\documentclass[11pt]{article}

\usepackage[margin=1in]{geometry}
\usepackage{graphicx}
\usepackage{amsmath, amssymb, amsfonts}
\usepackage{xr-hyper} 
\usepackage{hyperref}
\usepackage{url}
\usepackage{color}
\usepackage{cancel}
\usepackage{natbib}
\usepackage{adjustbox}
\usepackage{float}
\usepackage[dvipsnames]{xcolor}
\usepackage{comment}
\usepackage{lscape}
\usepackage{booktabs}
\usepackage{threeparttable}
\usepackage{tablefootnote}
\usepackage{xltabular}
\usepackage{multirow}
\usepackage{makecell}%
\usepackage{subcaption}
\usepackage{pgfplots}
\pgfplotsset{compat=1.14}
\usepackage{anyfontsize}
\usepackage{pgf}

\usepackage[normalem]{ulem}
\newcommand{\stkout}[1]{\ifmmode\text{\sout{\ensuremath{#1}}}\else\sout{#1}\fi}

\usepackage{tikz}
\usepackage{tikz-qtree}
\usetikzlibrary{trees}
\usetikzlibrary{shapes.multipart}
\usetikzlibrary{automata,positioning,decorations.pathmorphing}
\usetikzlibrary{arrows,shapes.arrows,shapes.geometric,shapes.multipart,decorations.pathmorphing,positioning,shapes.swigs, positioning, arrows.meta, swigs}

\allowdisplaybreaks

\newtheorem{theorem}{Theorem}

\newtheorem{lemma}[theorem]{Lemma}

\usepackage{xr}
\usepackage{mathtools}
\usepackage{algorithm}
\usepackage[noend]{algpseudocode}
\usepackage{enumerate}
\usepackage{parskip}
\newcommand{\E}{\mathbb{E}} 

\newcommand{\I}{\mathbb{I}}

\newcommand{\keywords}[1]{\par\addvspace\baselineskip \noindent\textbf{Keywords:}\enspace\ignorespaces#1}

\usepackage{authblk}

\author[1]{Shiyao Xu} 
\author[2]{Razieh Nabi} 
\author[3, 4]{Martin Underwood} 
\author[1]{Daniel Scharfstein} 
\affil[1]{Department of Population Health Sciences, University of Utah, Salt Lake City, UT, U.S.A.}
\affil[2]{Department of Biostatistics and Bioinformatics, Emory University, Atlanta, GA, U.S.A.}
\affil[3]{Warwick Clinical Trials Unit, University of Warwick, Coventry, U.K.}
\affil[4]{University Hospitals Coventry \& Warwickshire, Coventry, U.K.}

\title{Inferring Comprehensive Cohort Causal Effects in the Presence of Unmeasured Confounding and Missing Outcomes}

\date{}

\begin{document}

\maketitle

\begin{abstract}
This paper presents a methodological framework for estimating the comprehensive cohort causal effect (CCCE) in mixed-design clinical studies that combine randomized controlled trials (RCTs) and parallel observational study (OBS). Our approach is designed to evaluate robustness against unmeasured confounding in the OBS arm and to  handle outcomes that are missing at random in either the RCT or OBS arm. By employing a semiparametric theory-based sensitivity analysis framework, we derive the efficient influence function for the CCCE, parameterized by sensitivity parameters. We propose a one-step bias-corrected estimator that allows for flexible modeling and establish conditions under which our CCCE estimator is $\sqrt{n}$-consistent. To illustrate our methods, we apply them to the TOIB study, which evaluates the efficacy and safety of oral versus topical ibuprofen in managing chronic knee pain among older adults. We also evaluate the performance of the proposed methodology in a realistic simulation study.
\end{abstract}


\section{Introduction}
\label{sec:intro}

Randomized Controlled Trials (RCTs) are considered the gold standard for comparing treatments; however, their external validity is often questioned due to the enrollment of individuals who are not representative of the broader treatment-eligible population. To address this critique, methodologies have been developed that aim to generalize or transport inferences from RCTs to target populations \citep{review_lesko_2017, dahabreh2019extending, dahabreh2021study, review_shi_2022, review_Degtiar_2023, review_Colnet_2024}. These methods integrate data from RCTs and observational data sources. As noted by \citep{dahabreh2021study}, the integrated data sources can take a nested form where they come from a common cohort; or a non-nested form where they come from different cohorts.

In this paper, we consider the comprehensive cohort study (CCS) design \citep{olschewski1985CCS}, a nested study design in which individuals can opt to participate in an RCT or parallel observational study (OBS). In the latter study, individuals can choose among the treatment options of the RCT. The CCS design is often referred to as a (partially randomized) patient preference trial, as it accommodates individual preferences. The benefit of the design is that it systematically collects covariates, treatment and outcome information on broader range of eligible individuals than conventional RCTs \citep{brewin1989patient, brocklehurst1997partially, torgerson1998understanding}.

Two major threads of research have proposed methods for analyzing data arising from a CCS design. One thread focuses on estimating effects separately for the RCT and OBS components of the study (\citep{king1997angioplasty,schmoor2008evidence,olschewski1992analysis, schmoor1996randomized, brooks2000predictors, rovers2001generalizability,detre1999coronary, bedi2000assessing, kerry2000routine, jensen2003hormone}).   The other thread is concerned, as we are, with estimating an overall effect (\citep{olschewski1992analysis,brooks2000predictors,lu2019causal, dahabreh2022globalsensitivityanalysisstudies, Parikh_2025}). In particular, \citep{lu2019causal}, \citep{dahabreh2019generalizing}, \citep{dahabreh2022globalsensitivityanalysisstudies} and \citep{Parikh_2025} developed methods for estimating the \textit{comprehensive cohort causal effect} (CCCE), i.e., the average causal effect for the entire cohort. While targeting the same estimand, these latter approaches differ with regards to their identification assumptions. \citep{lu2019causal} and \citep{Parikh_2025} considered two assumptions, each sufficient for identification: explainable selection into RCT versus OBS and no unmeasured confounding in the OBS. \citep{dahabreh2022globalsensitivityanalysisstudies} considered a class of untestable assumptions anchored around explainable selection into RCT versus OBS. While methods have been developed for exploring the impact of deviations from the no unmeasured confounding assumption in observational studies (see \citep{nabi2024semiparametric} and references therein), they have not been combined, as we present here, with RCT data to evaluate the CCCE.

Our methodology is motivated by the TOIB study, a CCS comparing the effect of topical versus oral non-steroidal anti-inflammatory drugs (NSAIDs) on chronic knee pain \citep{underwood2008topical}. An added complexity in studies like the TOIB, which also apply to other similar studies \citep{comprehensive_Ashok_2005, comprehensive_Brinkhaus_2008, comprehensive_Kroz_2017, comprehensive_Glazener_2017}, is that outcomes are missing from both the RCT and OBS arms. This introduces an extra layer of inferential complexity \citep{nabi2023causal}, which we address under a missing at random assumption.  

The paper is organized as follows. 
In Section~\ref{sec:estimand_id}, we introduce notation, provide a mathematical representation of the CCCE, and discuss our identification assumptions. In Section~\ref{sec:infernece}, we describe our semiparametric estimation procedure along with its large sample properties. In Section~\ref{sec:alternatives}, we present an alternative sensitivity analysis framework by extending  \citep{dahabreh2022globalsensitivityanalysisstudies} to accommodate missing outcome data. In Section~\ref{sec:data}, we illustrate our methods using data from the TOIB study to evaluate the effect of oral versus topical ibuprofen on chronic knee pain in older adults. In Section~\ref{sec:sim} we present a simulation analysis to demonstrate the performance of our methodology. Section~\ref{sec:conc} provides a brief discussion. All proofs are deferred to the Appendix.

\section{Notation, Estimands, and Identification Assumptions}
\label{sec:estimand_id} 

\subsection{Notation and estimands}

Let $X$ denote a vector of measured pre-enrollment covariates, $R$ a randomization consent indicator ($1$ for RCT, $0$ for OBS), $T$ a binary treatment indicator, and $Y$ the outcome scheduled to be observed. Let $M$ denote a missing outcome indicator ($M=1$ if $Y$ is observed, $M=0$ if $Y$ is missing).  We denote the observed data in the absence of missingness by $O = (X, R, T, Y)$ and in the presence of missingness by $\widetilde{O} = (X,R,T,M,Y:M=1)$. We let $P$ and $\widetilde{P}$ denote the distribution of $O$ and $\widetilde{O}$, respectively. We assume that we observe $n$ i.i.d. copies of $\widetilde{O}$.

Note that $P$ is characterized by: (i) the cumulative conditional distribution function of outcome $Y$ given $T$, $R$ and $X$, defined as $F_{t, r}(y \,| \,x)=P(Y\leq y\mid T=t, R=r, X=x)$, (ii) the treatment propensity score, defined as $\pi_{t,r}(x) = P(T = t \mid R=r, X=x)$, (iii) the randomization consent propesnity score, defined as $g_r(x) = P(R = r \mid X=x)$, and (iv) the distribution of covariates, denoted by $F(x) = P(X \leq x)$. For notational convenience, we define the conditional regression of a specified function of the outcome $h(Y)$ as $\mu_{t,r}(h(Y); x) = \int h(y) dF_{t, r}(y|x)$. $\widetilde{P}$ is characterized by (ii), (iii), (iv) as well as (i') the cumulative distribution function of outcome $Y$ given $T$, $R$, $X$ and $M=1$, defined as $\widetilde{F}_{t, r}(y|x)=P(Y\leq y\mid T=t, R=r, X=x,M=1)$ and (v) the outcome missingness propensity score, denoted by $\eta_m(x, r, t)=P(M =m \mid X=x, R=r, T=t)$. We define $\widetilde{\mu}_{t,r}(h(Y); x) = \int h(y) d \widetilde{F}_{t, r}(y|x)$.

Let $T(r)$ be the treatment that would be received under study type $r$.  Let $Y(r,t)$ be the outcome that would be observed under selection into study type $r$ and treatment $t$. Throughout, we assume that study type selection has no impact on the outcome above and beyond treatment, so that $Y(r,t)=Y(t)$. The primary estimand is the CCCE, formally defined as $\E[Y(1)-Y(0)]$. Since the identification and estimation arguments for $\E[Y(1)]$ and $\E[Y(0)]$ are analogous, we focus on the generic mean potential outcome
\[
\psi_t=\E[Y(t)], \qquad t\in\{0,1\}.
\]

To characterize the contribution of the randomized and observational components of the study, we also consider the randomized trial causal effect (RTCE), $
\E[Y(1)-Y(0)\mid R=1]$, and the patient preference causal effect (PPCE), $\E[Y(1)-Y(0)\mid R=0]$. By the law of total expectation, $\text{CCCE} = \text{RTCE} \times P(R=1) + \text{PPCE} \times P(R=0)$.

\subsection{Identification assumptions}

We assume that in the randomized trial, there is full compliance with assigned treatment. We further assume: 
\begin{itemize}
    \setlength{\itemsep}{0.cm} 
    
    \item[(A1)\ ] \underline{Consistency}: The observed treatment is same as the counterfactual treatment had study type been set to the actual observed study type, i.e., $T(r)=T$ if $R=r$ for $r=0,1$. The observed outcome is the same as the counterfactual outcome had treatment been set to the actually observed treatment, i.e., $Y(t) = Y$ if $T=t$, for $t=0,1$. 
    
    \item[(A2)\ ] \underline{Positivity}: Treatment assignment has positive support in the RCT and OBS (conditioned on $X$) arms. That is, $0 < P(T=1 \mid R=1) < 1$ and $0 < P(T=1 \mid R=0, X=x) < 1$ for all $x$ in the support of $X$. 
    
    \item[(A3)\ ] \underline{Randomization in RCT}: Treatment assignment is independent of the potential outcomes and covariates in the RCT. That is, $T \perp \{ Y(t), X \} \mid R=1$, for $t = 0, 1$. 

    \item[(A4)\ ] \underline{Confounding Controlled by $\gamma_t$ in OBS}: Violation of no unmeasured confounding in OBS is controlled via $\gamma_t$ as follows: 
\begin{equation}
\label{nounmeasuredconfounding1}
   dF(y(t) \mid T=1-t,R=0,X)  =  dF(y(t) \mid  T=t,R=0,X) \times  \frac{ \exp\{  \gamma_t s_t(y(t)) \} }{ \mu_{t, 0} (\exp\{  \gamma_t s_t(Y(t)) \}; X)} \ , 
    \end{equation}
where $s_t(\cdot)$ is a pre-specified function and $\gamma_t$ is a sensitivity analysis parameter. 

     \item[(A5)\ ] \underline{Missing-At-Random Assumption for Outcome}: $M \perp Y \mid  T, X, R$. 

\end{itemize} 

By Assumption (A3), $\pi_{t, 1}(x) = \pi_{t, 1}$ for all $x$ in the support of the distribution of $X$. 

Assumption (A4) specifies a class of assumptions
indexed by a non-identified sensitivity analysis parameter $\gamma_t$, where $\gamma_t=0$ implies that $T$ is independent of $Y(t)$ given $X$ and $R=0$ (\cite{nabi2024semiparametric}). If $s_t(\cdot)$ is an increasing function, then $\gamma_t>0$ ($\gamma_t <0$) implies that the conditional (on $R=0$ and $X$) distribution of $Y(t)$ for those who opt for treatment $1-t$ is skewed to higher (lower) values of $Y(t)$ than those who opt for treatment $t$.  Using Bayes' rule, we can re-write (\ref{nounmeasuredconfounding1}) as
\begin{align}
\label{nounmeasuredconfounding2}
   \mbox{logit} &  \{ P[T=1-t|R=0,X,Y(t)] \} \nonumber \\
   & = \mbox{logit} \{ P[T=1-t|R=0,X] \} - \log\{\mu_{t, 0} (\exp\{  \gamma_t s_t(Y(t)) \}; X) \}  + \gamma_t  s_t(Y(t))\ .
\end{align}
    Thus, $\gamma_t$ is the conditional (on $R=0$ and $X$) log odds ratio of opting for treatment $1-t$ per unit change in $s_t(Y(t))$. 
    
    Assumptions (A3) and (A4) can be represented by the single-world intervention graph (SWIG) in Figure  \ref{fig:swig}. In this figure, the dashed lines are only present when $r=0$. Specifically, there are unmeasured factors denoted by $U$ that are associated with study type, $T(0)$ and $Y(t)$. The sensitivity analysis parameter $\gamma_t$ posits the association (above and beyond $X$) between $T(0)$ and $Y(t)$ that is induced by $U$.

Assumption (A5) implies that $F_{t, r}(y|x)=\widetilde{F}_{t, r}(y|x)$ and $\mu_{t,r}(h(Y); x) = \widetilde{\mu}_{t,r}(h(Y); x)$.

Under Assumptions (A1-A5), we can identify $\E[Y(t)]$ via the following functional of the observed data distribution, denoted by $\psi_t(\widetilde{P}; \gamma_t)$:
\begin{equation}
\label{eq:parameter_id1}
\psi_t(\widetilde{P}; \gamma_t) = \psi_{t,1}(\widetilde{P}) P(R=1) + \psi_{t,0}(\widetilde{P}; \gamma_t) P(R=0),
\end{equation}
where 
\[
\psi_{t,1}(\widetilde{P}) = \E[ \widetilde{\mu}_{t,1}(Y;X) | R=1]
\]

and
\[
\psi_{t,0}(\widetilde{P}; \gamma_t) = \E \bigg[ \Big\{ \widetilde{\mu}_{t,0}(Y;X) \pi_{t,0}(X) + \frac{ \widetilde{\mu}_{t,0}(Y \! \exp\{  \gamma_t s_t(Y)\} ;X) }{\widetilde{\mu}_{t,0}( \exp\{  \gamma_t s_t(Y)\} ;X) } \pi_{1-t,0}(X)  \Big\} \, \bigg| \, R=0\bigg]. 
\]

See Appendix~\ref{app:proofs_ID} for a proof. 

Note that $\psi_{t, 1}(\widetilde{P})$ and $\psi_{t, 0}(\widetilde{P}; \gamma_t)$ are the identification functionals of $\E[Y(t)|R=1]$ and $\E[Y(t)|R=0]$, which are the  mean of counterfactual outcomes in the subset of RCT and OBS participants. The \textit{comprehensive cohort causal effect} is defined as $\mbox{CCCE}(\gamma_1,\gamma_0) = \psi_1(\widetilde{P}; \gamma_1)-\psi_0(\widetilde{P}; \gamma_0)$;  the \textit{randomized trial causal effect} is defined as $\mbox{RTCE}=\psi_{1, 1}(\widetilde{P})-\psi_{0, 1}(\widetilde{P})$; and the \textit{patient preference causal effect} is defined as $\mbox{PPCE}(\gamma_1,\gamma_0) = \psi_{1, 0}(\widetilde{P}; \gamma_1)-\psi_{0, 0}(\widetilde{P}; \gamma_0)$.  It follows from (\ref{eq:parameter_id1}) that $\mbox{CCCE}(\gamma_1,\gamma_0) = \mbox{RTCE}  P(R=1) + \mbox{PPCE}(\gamma_1,\gamma_0) P(R=0)$.

\section{Estimation and Inference}
\label{sec:infernece}

\subsection{Influence function and remainder term}\label{sec:IF}

\begin{theorem}\label{EIF_miss}
As shown in Appendix~\ref{app:proofs_EIF}, an influence function for $\psi_t(\widetilde{P}; \gamma_t)$ is:
    \begin{align}\label{EIF_miss_eq}
    \phi_t(\widetilde{P}; \gamma_t)(\widetilde{O})
    &=  \underbrace{\upsilon_{t, 1}(\widetilde{P})(\widetilde{O})\times P(R=1)+\upsilon_{t, 0}(\widetilde{P}; \gamma_t)(\widetilde{O})\times P(R=0)}_{\upsilon_{t}(\widetilde{P}; \gamma_t)(\widetilde{O})} - \psi_t(\widetilde{P}; \gamma_t) \ , 
\end{align}
where $\upsilon_{t, 1}(\widetilde{P})(\widetilde{O})$ is given by
\begin{equation}\label{EIF_miss_eq_R1}
    \begin{aligned}
    \upsilon_{t, 1}(\widetilde{P})(\widetilde{O})
    =& \frac{M}{\eta_1(X, R, T)}\times\frac{\I(R=1)}{P(R=1)}\left\{\frac{\I(T=t)}{\pi_{t, 1}(X)}\left(Y-\widetilde{\mu}_{t, 1}(Y; X)\right)+\widetilde{\mu}_{t, 1}(Y; X))\right\} \\
    &\hspace{0.25cm}+\Big\{ 1 - \frac{M}{\eta_{1}(X, R, T)}  \Big\}\times \frac{\I(R=1)}{P(R=1)}\widetilde{\mu}_{t, 1}(Y; X )\ ,
\end{aligned}
\end{equation}
and $\upsilon_{t, 0}(\widetilde{P}; \gamma_t)(\widetilde{O})$ is given by
\begin{equation}
    \begin{aligned}
   \upsilon_{t, 0}(\widetilde{P}; \gamma_t)(\widetilde{O})
    =& \frac{M}{\eta_1(X, R, T)}\Bigg[\frac{\mathbb{I}(R = 0, T = t)}{P(R=0)}\Bigg\{ Y +  \frac{\pi_{1-t,0}(X)}{\pi_{t,0}(X)} \times \frac{\exp(\gamma_t s_t(Y))}{\widetilde{\mu}_{t,0}(\exp(\gamma_t s_t(Y)); X)}  \\
    &\hspace{6cm}\times\left\{Y  - \frac{\widetilde{\mu}_{t,0}(Y\exp(\gamma_t s_t(Y)); X)}{\widetilde{\mu}_{t,0}(\exp(\gamma_t s_t(Y)); X)} \right\} \Bigg\}\\
    &\hspace{2.5cm}+\frac{\mathbb{I}(R=0, T=1-t)}{P(R=0)} \times \frac{\widetilde{\mu}_{t,0}(Y\exp(\gamma_t s_t(Y)); X)}{\widetilde{\mu}_{t,0}(\exp(\gamma_t s_t(Y)); X)}\Bigg]\\
    &\hspace{0.25cm}+\Big\{ 1 - \frac{M}{\eta_{1}(X, R, T)}  \Big\}\Bigg[ \frac{\I(R=0, T=t)}{P(R=0)} \ \widetilde{\mu}_{t,0}(Y; X)\\
    &\hspace{2.5cm}+\frac{\I(R=0, T=1-t)}{P(R=0)}\ \frac{\widetilde{\mu}_{t,0}(Y\exp(\gamma_t s_t(Y)); X)}{\widetilde{\mu}_{t,0}(\exp(\gamma_t s_t(Y)); X)}\Bigg]\ .
\end{aligned}
\end{equation}
\end{theorem}

See Appendix~\ref{app:proofs_EIF} for a proof. 

\noindent {\bf Remark 1:} $\phi_t(\widetilde{P}; \gamma_t)(\widetilde{O})$ is the non-parametric influence function for $\psi_t(\widetilde{P}; \gamma_t)$ under a weaker version of Assumption (A3), i.e., $T \perp Y(t) \mid X, R=1$, for $t = 0, 1$ (Assumption (A3') - conditional ignorability of treatment in the RCT). Assumption (A3) places restrictions on the observed data distribution. As shown in Appendix~\ref{app:proofs_EIF}, the efficient influence function is given by (\ref{EIF_miss_eq}) with  $\pi_{t,1}(X)$ replaced by $\pi_{t,1}$. In our implementation, we use (\ref{EIF_miss_eq}) to adjust for potential covariate imbalance that can arise in finite samples \citep{cov_adjust_Benkeser_2021, cov_adjust_jin_2026}. As we will see in Section \ref{est_asymp}, our estimator will, nonetheless, be asymptotically efficient. In Appendix~\ref{app:proofs_R0R1}, we will see that $\upsilon_{t, 1}(\widetilde{P})(\widetilde{O})$ and $\upsilon_{t, 0}(\widetilde{P}; \gamma_t)(\widetilde{O})$ are related to the influence functions for $\psi_{t, 1}(\widetilde{P})$ and $\psi_{t, 0}(\widetilde{P}; \gamma_t)$. 

\noindent {\bf Remark 2:} In Lemma~\ref{remainder} of Appendix~\ref{sec:remainder}, we derive the remainder term $Rem_t(\widetilde{P}^\ast, \widetilde{P}) \equiv \psi_t(\widetilde{P}^\ast; \gamma_t) - \psi_t(\widetilde{P}; \gamma_t) + \E[\phi_t(\widetilde{P}^\ast; \gamma_t)(\widetilde{O})]$, where $\widetilde{P}^\ast$ is any distribution of $\widetilde{O}$. The remainder can be expressed as a sum of expectations involving products of nuisance estimation errors, a structure that is key to establishing the asymptotic distribution of our estimator of $\psi_t(\widetilde{P}; \gamma_t)$. Specifically, $Rem_t(\widetilde{P}^\ast, \widetilde{P})$ depends on products of the errors 
$\widetilde{\mu}^\ast_{t,r}(h(Y); X)-\widetilde{\mu}_{t,r}(h(Y); X)$ with $\pi_{1, r}^\ast(X)-\pi_{1, r}(X)$ and $\eta_1^\ast(X, r, t)-\eta_1(X, r, t)$ for $r=0, 1$. Consequently, asymptotic normality of the estimator for $\psi_t(\widetilde P;\gamma_t)$ can be achieved even when $\widetilde F_{t,r}(y\mid X)$, $\pi_{1,r}(X)$, and $\eta_1(X,r,t)$ are estimated at rates as slow as $n^{-1/4}$. More generally, if one set of nuisance functions is estimated slowly using flexible methods, asymptotic normality is still obtained provided that the resulting product of estimation errors is $o_p(n^{-1/2})$.

\subsection{Modeling and estimation of $\widetilde{P}$}\label{estimateP}

To fix ideas, we propose specific models whose parameters can be estimated at rates fast enough so that our resulting estimators converge to a normal distribution at $\sqrt{n}$-rates. Let $\lVert f(Z)\rVert_{L_2}\coloneq\sqrt{\int f^2(z)dF(z)}$. 

We posit a generalized additive model (GAM) with a logistic link function for $\pi_{1, r}(X)$ and $\eta_1(X, r, t)$ for $r=0, 1$ and $t=0, 1$ \citep{hastie2017generalized}.  Let $\widehat{\pi}_{1, r}(X)$ be the GAM estimator of 
$\pi_{1, r}(X)$.  \cite{horowitz2004nonparametric} showed that
\[
\Big|\Big| \ \widehat{\pi}_{1, r}(X) - \pi_{1, r}(X) \ \Big|\Big|_{L_2} =O_P(n^{-2/5}).
\]
for $r=0, 1$. Similarly, we have $\Big|\Big| \ \widehat{\eta}_1(X, r, t) - \eta_1(X, r, t) \ \Big|\Big|_{L_2} =O_P(n^{-2/5})$ for $r=0, 1$ and $t=0, 1$. 

We posit a single index model \citep{chiang2012new} for $\widetilde{F}_{t, r}(y|X)$. This model assumes that 
\[
\widetilde{F}_{t, r}(y|X)= \widetilde{F}_{t, r}(y,X'\beta_{t, r};\beta_{t, r}),
\]
where $F_{t, r}(y,u;\beta_{t, r})$ is a cumulative distribution function in $y$ for each $u$ ($\widetilde{F}_{t, r}(y,u;\beta_{t, r})=P(Y\leq y\mid T=t, R=r, X'\beta_{t, r}=u,M=1)$), $\beta_{t, r} = (\beta_{1,t, r},\ldots,\beta_{p,t, r})$ is a vector of unknown parameters and, for purposes of identifiability, $\beta_{t, r}$ has norm $1$ \citep{redd2025sensiatrpackageconducting}. We estimate $\widetilde{F}_{t, r}(y,x'\beta_{t, r};\beta_{t, r})$ by 
\[
\widehat{\widetilde{F}}_{t, r}(y,X'\widehat{\beta}_{t, r};\widehat{\beta}_{t, r}) = \frac{ \sum_{T_i=t} \I(Y_i \leq y) K_{\widehat{h}_{t, r}}(X_i'\widehat{\beta}_{t, r} - X'\widehat{\beta}_{t, r}) }{ \sum_{T_i=t} K_{\widehat{h}_{t, r}}(X_i'\widehat{\beta}_{t, r} - X'\widehat{\beta}_{t, r})) },
\]
where $K_h(v) = K(v/h)/h$, $K$ is a kernel function, $\widehat{\beta}_{t, r}$ is an estimator of $\beta_{t, r}$ and $\widehat{h}_{t, r}$ is an estimator of the bandwidth $h_{t, r}$ (see Appendix~\ref{nuisance} for estimation details). We then estimate  $\widetilde{\mu}_{t, 0}\big(Y\exp(\gamma_t s_t( Y)); X\big)$, $\widetilde{\mu}_{t, 0}\big(\exp(\gamma_t s_t(Y)); X\big)$ and $\widetilde{\mu}_{t, r}\big(Y; X\big)$ by
\begin{align*}
    \widehat{\widetilde{\mu}}_{t, 0}\big(Y\exp \{ \gamma_t s_t( Y)\} ; X\big) &= \int y\exp\{\gamma_t s_t(y)\} d \widehat{\widetilde{F}}_{t, 0}(y,X'\widehat{\beta}_{t, 0};\widehat{\beta}_{t, 0}), \text{ and } \\
    \widehat{\widetilde{\mu}}_{t, 0}\big(\exp(\gamma_t s_t(Y)); X\big) &= \int \exp\{\gamma_t s_t(y)\} d \widehat{\widetilde{F}}_{t, 0}(y,X'\widehat{\beta}_{t, 0};\widehat{\beta}_{t, 0}), \text{ and } \\
    \widehat{\widetilde{\mu}}_{t, r}\big(Y; X\big) &= \int yd \widehat{\widetilde{F}}_{t, r}(y,X'\widehat{\beta}_{t, r};\widehat{\beta}_{t, r}). 
\end{align*}

In Appendix~\ref{nuisance}, we establish conditions under which $\bigg|\bigg| \widehat{\widetilde{\mu}}_{t, 0}\big(\exp(\gamma_t s_t(Y)); X\big) - \widetilde{\mu}_{t, 0}\big(\exp(\gamma_t s_t(Y)); X\big) \bigg|\bigg|_{L_2}$, $\bigg|\bigg| \widehat{\widetilde{\mu}}_{t, 0}\big(Y\exp(\gamma_t s_t(Y)); X\big) - \widetilde{\mu}_{t, 0}\big(Y\exp(\gamma_t s_t(Y)); X\big) \bigg|\bigg|_{L_2}$ and $\bigg|\bigg| \widehat{\widetilde{\mu}}_{t, r}\big(Y; X\big) - \widetilde{\mu}_{t, r}\big(Y; X\big) \bigg|\bigg|_{L_2}$ are $O_P(n^{-2/5})$. 

\subsection{Estimation of $\psi_t(\widetilde{P}; \gamma_t)$ and asymptotic properties}\label{est_asymp}

We estimate $\psi_t(\widetilde{P};\gamma_t)$ using a one-step, split sample estimator \citep{kennedy2023semiparametricdoublyrobusttargeted} based on the influence function $\phi_t(\widetilde{P};\gamma_t)(\widetilde{O})$ presented above.
Specifically, we randomly split the observations into $K$ disjoint sets, where $S_i$ denotes the split membership of the $i$th observation (i.e, $S_i \in \{1, \ldots, K\}$). The size of the $k$th disjoint set is denoted by $n_k \approx n/K$ (i.e., $n_k = O(n)$).  We let $\widehat{\widetilde{P}}^{(-k)}$ be an estimator of $\widetilde{P}$ based on observations from all the splits except that of $k$th split. We let $P_n$ denote the empirical distribution of the observed data $\widetilde{O}$ based on $n$ observations, and $P_{n/n_k}^{(-k)}$ be the empirical distribution of the observed data from all the splits except that of $k$th split. 

Our one-step, split sample estimator takes the form:
\begin{equation*}
    	\widehat{\psi}_t(\gamma_t) = \frac{1}{K} \sum_{k=1}^K \underbrace{\left\{\frac{1}{n_k} \sum_{i: S_i=k}\Bigg[ \upsilon_{t, 1} \left(\widehat{\widetilde{P}}^{(-k)}\right) (\widetilde{O}_i)\times P_{n/n_k}^{(-k)}(R=1) +\upsilon_{t, 0} \left(\widehat{\widetilde{P}}^{(-k)} ; \gamma_t\right) (\widetilde{O}_i)\times P_{n/n_k}^{(-k)}(R=0)\Bigg]\right\}}_{\mbox{$k$th Split One-Step Estimator } \left(\widehat{\psi}_t^{(k)}(\gamma_t)\right)}. 
\end{equation*}

Similar to \citet{nabi2024semiparametric}, we apply the tuning-free Huberization procedure \citep{tuning_wang_2021} to avoid extreme values of $\upsilon_{t, 1}\left(\widehat{\widetilde{P}}^{(-k)}\right)(\widetilde{O}_i)$ and $\upsilon_{t, 0}\left(\widehat{\widetilde{P}}^{(-k)} ; \gamma_t\right)(\widetilde{O}_i)$ for some $i \in \{1,\ldots,n\}$. The $k$th split one-step truncated estimator $\widehat{\psi}_t^{(k)\dagger}(\gamma_t)$ is 
\begin{equation}\label{truncate_estimator}
\begin{aligned}
    	\widehat{\psi}_t^{(k)\dagger}(\gamma_t) = \frac{1}{n_k} \sum_{i: S_i=k}&\Bigg[\underbrace{\min\left\{\left|\upsilon_{t, 1}\left(\widehat{\widetilde{P}}^{(-k)} \right)(\widetilde{O}_i)\right|, \widehat\tau_{t, 1}^{(k)}\right\}\text{sign}\left\{\upsilon_{t, 1}\left(\widehat{\widetilde{P}}^{(-k)}\right)(\widetilde{O}_i)\right\}}_{\upsilon_{t, 1}^\dagger\left(\widehat{\widetilde{P}}^{(-k)}\right)(\widetilde{O}_i)}\times P_{n/n_k}^{(-k)}(R=1)\\
        &+\underbrace{\min\left\{\left|\upsilon_{t, 0}\left(\widehat{\widetilde{P}}^{(-k)} ; \gamma_t\right)(\widetilde{O}_i)\right|, \widehat\tau_{t, 0}^{(k)}\right\}\text{sign}\left\{\upsilon_{t, 0}\left(\widehat{\widetilde{P}}^{(-k)} ; \gamma_t\right)(\widetilde{O}_i)\right\}}_{\upsilon_{t, 0}^\dagger\left(\widehat{\widetilde{P}}^{(-k)} ; \gamma_t\right)(\widetilde{O}_i)}\times P_{n/n_k}^{(-k)}(R=0)\Bigg]\ ,
\end{aligned}
\end{equation}
where $\widehat\tau_{t, 0}^{(k)}$ and $\widehat\tau_{t, 1}^{(k)}$ are the non-negative solutions to
\begin{equation}\label{trunc1}
    	\frac{1}{n_k} \sum_{i: S_i=k}\frac{\min\left\{\upsilon_{t, 1}\left(\widehat{\widetilde{P}}^{(-k)}\right)(\widetilde{O}_i)^2, \left(\widehat\tau_{t, 1}^{(k)}\right)^2\right\}}{\left(\widehat\tau_{t, 1}^{(k)}\right)^2}=\frac{\log(n_k)}{n_k}
\end{equation}
and
\begin{equation}\label{trunc2}
    	\frac{1}{n_k} \sum_{i: S_i=k}\frac{\min\left\{\upsilon_{t, 0}\left(\widehat{\widetilde{P}}^{(-k)} ; \gamma_t\right)(\widetilde{O}_i)^2, \left(\widehat\tau_{t, 0}^{(k)}\right)^2\right\}}{\left(\widehat\tau_{t, 0}^{(k)}\right)^2}=\frac{\log(n_k)}{n_k}
\end{equation}
The one-step, split sample truncated estimator $\widehat{\psi}_t^\dagger(\gamma_t)$ is
\begin{equation*}
\widehat{\psi}_t^\dagger(\gamma_t)=\frac{1}{K}\sum_{k=1}^K\widehat{\psi}_t^{(k)\dagger}(\gamma_t)
\end{equation*}

In Appendix \ref{robustness}, we show that $\widehat{\psi}_t^\dagger(\gamma_t)$ is robust to mis-specification of $\pi_{1, r}(X)$ and $\eta_1(X, r, t)$ for $r=0, 1$. That is, $\widehat{\psi}_t^\dagger(\gamma_t)$ is a consistent estimator under correct specification of $\widetilde{F}_{t, r}(y|X)$  for $r=0, 1$. Note that correct specification of $\widetilde{F}_{t, r}(y|X)$ implies correct specification of $\widehat{\widetilde{\mu}}_{t,r}(h(Y); X)$.

In Appendix~\ref{asymptotic}, we prove the following theorem:

\begin{theorem}\label{asymptotic_theorem}
Assume that 
\begin{enumerate}
    \item $|Y|$ and $|\exp(\gamma_ts_t(Y))|$ are bounded in probability,  and 
    \item $\pi_{t, r}(X)$, $\eta_1(X, r, t)$ and $\mu_{t, 0}(\exp(\gamma_t s_t(Y);X)$ are bounded away from zero with probability one. 
\end{enumerate}
Under Assumptions (A1-A5), the modeling assumptions (correct specifications) in Section~\ref{estimateP}, and the estimation procedure described therein, 
$$\sqrt{n} \left\{ \widehat{\psi}_t^\dagger(\gamma_t) - \psi_t (\widetilde{P}; \gamma_t) \right\} \stackrel{D(\widetilde{P})}{\rightarrow} \mathcal{N}(0, \E[\phi_t(\widetilde{P}; \gamma_t)(\widetilde{O})^2]).$$
\end{theorem}

We estimate the variance of $\widehat{\psi}_t^\dagger(\gamma_t)$ by
$$\frac{1}{nK}\sum_{k=1}^K\left\{\frac{1}{n_k-1}\sum_{i: S_i=k}\left\{\upsilon_t^\dagger\left(\widehat{\widetilde{P}}^{(-k)} ; \gamma_t\right)(\widetilde{O}_i)-\widehat{\psi}_t^{(k)\dagger}(\gamma_t)\right\}^2\right\}$$
where 
$$\upsilon_t^\dagger\left(\widehat{\widetilde{P}}^{(-k)} ; \gamma_t\right)(\widetilde{O}_i)=\upsilon_{t, 1}^\dagger\left(\widehat{\widetilde{P}}^{(-k)}\right)(\widetilde{O}_i)\times P_{n/n_k}^{(-k)}(R=1)+\upsilon_{t, 0}^\dagger\left(\widehat{\widetilde{P}}^{(-k)} ; \gamma_t\right)(\widetilde{O}_i)\times P_{n/n_k}^{(-k)}(R=0)\ .$$
We construct a $95\%$ Wald confidence interval for $\psi_t(\widetilde{P}; \gamma_t)$ as well as a symmetric-t $95\%$ confidence interval using parametric bootstrap \citep{efron1993introduction}. 

In order to obtain a $\sqrt{n}-$consistent estimator of $\psi_t (\widetilde{P}; \gamma_t)$ under mis-specification of $\pi_{t, r}(X)$ and $\eta_1(X, r, t)$, we need to apply parametric estimation methods for $\widetilde{F}_{t, r}(y|X)$ such that $\lVert \widehat{\widetilde{\mu}}_{t, r}(h(Y); X) - \widetilde{\mu}_{t, r}(h(Y); X) \rVert_{L_2}=o_p(n^{-1/2})$ (see Appendix~\ref{asymptotic}, Lemma~\ref{mis_asymp}). 

Under Assumption (A3), the asymptotic variance $\E[\phi_t(\widetilde{P}; \gamma_t)(\widetilde{O})^2]$ equals the variance of the efficient influence function. Therefore, under the conditions of Theorem \ref{asymptotic_theorem}, our estimator for $\psi_t(\widetilde{P}; \gamma_t)$ is semiparametric efficient.

\section{Alternative Estimation Approach}
\label{sec:alternatives}

Another way to identify $\E[Y(t)]$ is to assume individuals in the RCT and OBS are {\em exchangeable} within levels of $X$. That is, $R \perp Y(t) \mid X$. This assumption is commonly applied to generalize RCT results to a larger trial-eligible population \citep{exchange_Stuart_2001, lu2019causal, review_Colnet_2024}. \citet{dahabreh2022globalsensitivityanalysisstudies} embedded this assumption within a sensitivity model of the following form:
\begin{equation}
\label{exchangeability}
   dF(y(t) \mid R=0, X) = dF(y(t) \mid R=1,X)\times\frac{\exp\{  \gamma_t' s_t(y(t)) \} }{ \mu_{t, 1}(\exp(\gamma_t' s_t(Y)); X)},
\end{equation}
where $s_t(\cdot)$ is a specified function and $\gamma_t'$ is a non-identified sensitivity analysis parameter.  Note that exchangeability is obtained when  $\gamma_t'=0$. 
If $s_t(\cdot)$ is an increasing function, then $\gamma_t>0$ ($\gamma_t <0$) implies that the conditional (on $R=0$ and $X$) distribution of $Y(t)$ for those who opt for treatment $1-t$ is skewed to higher (lower) values of $Y(t)$ than those who opt for treatment $t$.  Using Bayes' rule, we can re-write (\ref{nounmeasuredconfounding1}) as
\begin{align*}
   \mbox{logit} &  \{ P[R=0|X,Y(t)] \} \nonumber \\
   & = \mbox{logit} \{ P[R=0|X] \} - \log\{\mu_{t, 1}(\exp(\gamma_t' s_t(Y)); X)\}  + \gamma'_t  s_t(Y(t))\ , 
    \end{align*}
    Thus, $\gamma'_t$ is the conditional (on $X$) log odds ratio of opting for the observational study per unit change in $s_t(Y(t))$.

Under Assumptions (A1-A3), (A5) and (\ref{exchangeability}), $\E[Y(t)]$ is identified according to the following formulae:
\begin{equation}\label{eq:parameter_id_ex}
    \psi_t'(\widetilde{P}; \gamma_t') = \psi_{t,1}(\widetilde{P}) P(R=1)+\psi_{t, 0}'(\widetilde{P}; \gamma_t')P(R=0), 
\end{equation}
where $\psi_{t,1}(\widetilde{P})$ is the same as in (\ref{eq:parameter_id1}), and $\psi_{t, 0}'(\widetilde{P}; \gamma_t')$ is the identification functional of $\E[Y(t)|R=0]$:
\[
    \psi_{t, 0}'(\widetilde{P}; \gamma_t') = \E\left[ \frac{\mu_{t, 1}(Y \exp(\gamma_t' s_t(Y)); X)}{\mu_{t, 1}(\exp(\gamma_t' s_t(Y)); X)} \, \bigg| \, R=0\right]
\]

See Appendix~\ref{app:model_2} for details on the estimation procedures for $\psi_t'(\widetilde{P}; \gamma_t')$ and $\psi_{t, 0}'(\widetilde{P}; \gamma_t')$. Identification and estimation procedure for $\E[Y(t)|R=1]$ is the same as in Section~\ref{sec:IF}. 

The sensitivity models in  (\ref{nounmeasuredconfounding1}) and (\ref{exchangeability}) each identify  $\E[Y(t)]$ without imposing restrictions on the observed data law $\widetilde{P}$.  In Lemma~\ref{onetoone} of Appendix~\ref{app:model_2}, we show that in models (\ref{nounmeasuredconfounding1}) and (\ref{exchangeability}), $\psi_t(\widetilde{P},\gamma_t)$  and $\psi'_t(\widetilde{P},\gamma_t')$ are increasing functions of $\gamma_t$ and $\gamma_t'$, respectively. This implies that, for each $\gamma_t$, there exists a $\gamma_t'$ (possibly infinite) that minimizes $\{\psi_t(\widetilde{P},\gamma_t) - \psi'_t(\widetilde{P},\gamma_t')\}^2$ and vice-versa. Thus, having knowledge about the level of departure from the unmeasured confounding assumption in OBS ($\gamma_t$) will inform the level of deviation from the exchangeability assumption over $R$ ($\gamma_t'$), and vice versa. We illustrate this point in Section~\ref{sec:data}. This is a generalization of the falsification test for external validity and unmeasured confounding in \citet{Parikh_2025}, where we can test for violation of either $R\perp Y(t)|X$ or $T\perp Y(t)|X, R=0$, given the other.

\section{The TOIB Study}
\label{sec:data}

The TOIB study was designed to determine whether general practitioners should advise their older patients with chronic knee pain to use topical or oral NSAIDs \citep{underwood2008topical}. As an equivalence trial, its primary objective was to evaluate whether the effects of advice to use topical NSAIDs were equivalent to those of advice to use oral NSAIDs on knee pain and disability outcomes \citep{toib_underwood_2008}. The study included two components: RCT and OBS. In the RCT, patients were randomized to receive either a recommendation to use a topical or an oral anti-inflammatory drug. In the OBS, patients received a recommendation based on their expressed preference for a topical or oral anti-inflammatory drug. After excluding 22 patients (9 in RCT, 13 in OBS) with missing baseline covariate data, the RCT includes 273 participants: 140 in the oral group and 133 in the topical group and the OBS includes 290 participants: 75 in the oral group and 215 in the topical group. Our goal is to use the proposed methodology to estimate the causal effect of recommendation for topical ($t=1$) versus oral anti-inflammatory drug ($t=0$) on Western Ontario and McMaster Universities Osteoarthritis Index (WOMAC) pain score evaluated at 12 months after enrollment ($Y$). We also report the effect of treatment on $Y$ among RCT and among OBS participants. The effects of advice to use topical versus oral NSAIDs on knee pain were considered equivalent if the $95\%$ confidence interval for the between-group difference in WOMAC pain scores lay entirely within the prespecified equivalence bounds of -10 to 10 points \citep{toib_underwood_2008}. 

The WOMAC pain score in \citet{underwood2008topical} was computed as the average of five questions about pain level (pain walking on a flat surface, pain going up or down stairs, pain at night while in bed, pain sitting or lying, pain standing upright), each recorded as a score ranging from 0 (no pain) to 100 (extreme pain). If the answer to any of the five questions is missing, then the WOMAC pain score is the average of four. Having more than one missing answer is defined as having missing WOMAC pain score. At 12 month, $19.9\%$ of the participants are missing WOMAC pain score. It is believed that the WOMAC pain score is less discriminating at the extremes. Thus we set $s_t(Y)=\Phi(\frac{Y-60}{25})$, where $\Phi$ is the cumulative distribution function of a normal distribution (see Figure~\ref{fig_s_t_y}).

In our analysis, baseline covariates $X$ included age, baseline WOMAC pain score, expected pain a year from now (much/a little worse, about the same, a little/much better/free of pain) and chronic pain grade (grade 1-2, grade 3-4). We specified the regression models for treatment, outcome and missing outcome mechanisms to have additive effects of the baseline covariates. We allowed for non-linear effects of age and pain in the treatment and missing outcome models. Table~\ref{Table1} shows descriptive statistics of $X$ by study enrollment and treatment received. Due to small sample sizes, we were limited in how many covariates we could include in our analysis. For the OBS arm, it is unlikely that the four covariates included in the analysis  fully account for the differences in patient characteristics between treatment groups, leading to issues of unmeasured confounding. In addition, in the RCT arm, we observe covariate imbalance between treatment groups: on average, participants in the topical group have lower baseline WOMAC pain score and lower chronic pain grade than the oral group. Covariate adjustment in the RCT can address such imbalances.

Figure~\ref{fig1} shows the estimated $\E[Y(t)]$, $\E[Y(t)|R=0]$ and $\E[Y(t)|R=1]$ for $t=0, 1$ (solid lines) as a function of $\gamma_1$ and $\gamma_0$. When there is no unmeasured confounding ($\gamma_1=\gamma_0=0$), the estimated difference in WOMAC pain score had all patients in the comprehensive cohort received a recommendation for topical versus oral NSAIDs is $-0.24$ (symmetric-t $95\%$ CI: ($-4.28$, $3.81$)); the estimated difference in WOMAC pain score had all patients in the OBS received topical versus oral NSAIDs is $-0.02$ ($95\%$ CI: ($-5.79$, $5.75$)). Since $\E[Y(t)|R=1]$ for $t=0, 1$ does not depend on $\gamma_t$, the estimated effect of topical versus oral NSAIDs among patients enrolled in the RCT is $-0.13$ ($95\%$ CI: ($-6.18$, $5.92$)), regardless of $\gamma_1$ and $\gamma_0$. In our sensitivity analysis, we restricted the range of $\gamma_t$ to $[-2, 2]$ so that estimated change from baseline in mean WOMAC pain score under the topical or oral recommendation with varying $\gamma_t$ value would not exceed 5 points. 

In Figure~\ref{fig2}, we plot the induced estimates of $\E[Y(t)|T=1-t]$, $\E[Y(t)|T=1-t, R=0]$ and  $\E[Y(t)|T=1-t, R=1]$ as functions of $\gamma_t$. See Appendix~\ref{app:data_analysis} for details of estimation. The induced estimates reflect the average WOMAC score at 12 months under topical (oral) NSAIDs recommendation for those who actually received a recommendation for  oral (topical) NSAIDs in the entire cohort, the OBS arm and the RCT arm, respectively, under different $\gamma_t$ specifications. Subject matter experts could utilize this information to determine plausible range of sensitivity parameters for analysis. 

Another way to understand our choice of sensitivity parameters is to observe how different combinations of $\gamma_1$ and $\gamma_0$ influence treatment effects. Figure~\ref{fig3} shows the contour plots of $\E[Y(1)]-\E[Y(0)]$ and $\E[Y(1)|R=0]-\E[Y(0)|R=0]$ under different values of $\gamma_1$ and $\gamma_0$. The bold black curve represent the contour where treatment effects are estimated to be zero. We include regions of $\gamma_1$ and $\gamma_0$ (defined by the blue curves) that would lead to equivalence results about the effect of topical versus oral NSAIDs on pain level at 12 months, under a equivalence margin of 10 points \citep{toib_underwood_2008}. Since all (almost all) of the $(\gamma_1, \gamma_0)$ sets within the $[-2, 2]$ range would result in equivalence conclusions among the entire cohort (among OBS participants), we demonstrate the robustness of our result against unmeasured confounding in the OBS arm. 

To demonstrate the relationship between $\gamma_t$ and $\gamma_t'$, we plot the estimated $\E[Y(t)]$ and $\E[Y(t)|R=0]$ as a function of $\gamma_t$ (solid lines), and intersect with the estimated $\E[Y(t)]$ and $\E[Y(t)|R=0]$ under $\gamma_t'=0, 0.5, 1$ (dashed horizontal lines) in Figure~\ref{fig4}. The horizational axis coordinates where the lines meet are the corresponding $\gamma_t$ value when $\gamma_t'=0, 0.5, 1$. For instance, $\gamma_1'=0$ corresponds to $\gamma_1=-1.5$. And $\gamma_0'=0$ corresponds to $\gamma_0=-2.2$. That is, if we assume exchangeability with respect to consent into the RCT, there is evidence of unmeasured confounding with respect to the recommendation received in OBS. Specifically, for participants who enroll in OBS, the conditional (on $X$) probability of opting for an oral (topical) recommendation decreases as WOMAC pain score under a topical (oral) recommendation increases. 
Similary, if we assume no unmeasured confounding in OBS, there is evidence that exchangeability with respect to consent into the RCT does not hold.  That is, the conditional (on $X$) probability of selecting the OBS study decreases as WOMAC pain score under a topical (oral) recommendation increases.

\section{Simulation Study}
\label{sec:sim}

We conducted a realistic simulation study to evaluate the performance of our estimation approach for $\E[Y(t)]$, $\E[Y(t)|R=0]$ and $\E[Y(t)|R=1]$. We used data from the TOIB study to build the true observed data generating mechanisms. Specifically, we used the empirical distribution of $X$, the estimated GAM models for treatment and missing outcome mechanisms, and the estimated parameters from beta regression models for scaled outcome ($Y/100$) to generate simulated data. We considered $\gamma_1=(-2, -1.5, -1, 0.5, 0, 0.5, 1, 1.5, 2)$ and $\gamma_0=(-2, -1.5, -1, 0.5, 0, 0.5, 1, 1.5, 2)$. For each choice of $\gamma_t$, we used the observed data distribution and the functional form of $s_t(Y)$ discussed in Section~\ref{sec:data} to compute the true value of $\E[Y(t)]$, $\E[Y(t)|R=0]$ and $\E[Y(t)|R=1]$ using (\ref{eq:parameter_id1}). We considered sample sizes of 563 (the sample size of the original dataset), 1000 and 2000. For each sample size, we simulated 2000 datasets. We evaluated estimation bias, $95\%$ Wald confidence interval coverage and  $95\%$ symmetric-t confidence interval coverage from parametric bootstrap. 

Table~\ref{coverage} presents the result of our simulation study for different $\gamma_t$ and sample sizes. The performance for estimating $\E[Y(t)]$, $\E[Y(t)|R=0]$ and $\E[Y(t)|R=1]$ are listed in columns, with topical and oral treatment reported side by side for each estimand of interest. Regardless of the value of $\gamma_t$, bias is low and decreases with sample size. Wald confidence intervals have coverage close to the nominal level, with the exception of oral NSAIDs for the estimation of $\psi_0(\widetilde{P}; \gamma_0)$ and $\psi_{0, 0}(\widetilde{P}; \gamma_0)$ with larger $\gamma_0$ values. Symmetric-t confidence interval using parametric bootstrap have coverage close to the nominal level.

\section{Discussion}
\label{sec:conc}

In this paper, we developed a semiparametric sensitivity analysis approach to address potential unmeasured confounding in estimating comprehensive cohort causal effects when outcomes are missing at random. The methods can be further extended to address deviations from the missing at random assumption along the lines of \cite{scharfstein2003generalized}.

Assumption (A3) as formulated in (\ref{nounmeasuredconfounding1}) is not  the only way to identify $\E[Y(t)]$.  In Section~\ref{sec:alternatives}, we discussed an alternative class of identfiying assumptions that explore deviations from conditional exchangeability over consent into RCT. Another alternative (\cite{robins2000sensitivity}) would be to replace Assumption (A3) by an assumption that posits that
\begin{align*}
g(\E[Y(t)|T=1-t,R=0,X]) - g(\E[Y(t)|T=t,R=0,X]) = \delta_t(X), 
\end{align*}
where $g(\cdot)$ is a specified function whose range is the entire real line and $\delta_t(X)$ is a non-identified sensitivity analysis function. For tractability of the sensitivity analysis, it is typically assumed that $\delta_t(x) = \delta_t$ for all $x$, which may be unrealistic. In contrast, (\ref{nounmeasuredconfounding1}) implies that 
\begin{align*}
& g(\E[Y(t)|T=1-t,R=0,X]) - g(\E[Y(t)|T=t,R=0,X]) \\
& = g \left(\frac{\E[ Y(t) \exp\{ \gamma_t s_t Y(t) \} |T=t,R=0,X]}{\E[\exp\{ \gamma_t s_t Y(t) \} |T=t,R=0,X]}\right) -  g(\E[Y(t)|T=t,R=0,X]),
\end{align*}
where the right hand side is a function of $X$ and $\gamma_t$. Of course, our approach assumes that the effect of $Y(t)$ on treatment selection in the observational study does not vary by $X$ (see \eqref{nounmeasuredconfounding2}). Ultimately, the approach adopted needs to rely on expert judgment about the underlying assumptions.

In this paper, we assumed full compliance in the RCT, i.e., treatment assigned equals treatment received.  This is plausible in the TOIB study because the intervention was a treatment recommendation provided to the patient by their physician. When there is non-compliance in the RCT, the notion of what constitutes treatment is different between RCT and OBS. In the RCT, it would be treatment assigned and in the OBS it would be treatment received. To align the definition of treatment between the two study types, additional assumptions would be required as treatment received in the RCT would no longer benefit from randomization. As currently formulated, our approach will provide valid inferences if treatment received in the RCT is randomized within levels of X; this is because we allowed  the treatment model in the RCT to depend on $X$. If conditional randomization is suspect, one would then need to explore deviations from this assumption through an RCT-specific  sensitivity analysis model like \eqref{nounmeasuredconfounding1}. 

\newpage

\section*{Acknowledgements}
The research reported in this publication was supported by the National Center for Advancing Translational Sciences of the National Institutes of Health under Award Number UM1 TR004409. The content is solely the responsibility of the authors and does not necessarily represent the official views of the National Institutes of Health. The TOIB  study was funded by the NIHR Health Technology Assessment programme (Award ID 01/09/02). Views expressed in the publication are representative of the authors and not necessarily those of the NIHR or the Department of Health and Social Care.

\section*{Data availability}
The data that support the findings in this paper cannot be shared because of confidentiality and/or data use restrictions. The R code used for the analysis, together with a vignette and sample data, is available at \url{https://github.com/UofUEpiBio/CompCausal}.

\section*{Conflict of Interest}
Martin Underwood has been chief investigator or co-investigator on multiple previous and current research grants from the UK National Institute for Health Research, and is a co-investigator on current grants funded by the Australian NHMRC and Norwegian MRC. He is a director and shareholder of Clinvivo Ltd that provides electronic data collection for health services research. He has accepted honoraria for examining theses, and performing peer review. He receives some salary support from University Hospitals Coventry and Warwickshire. He is a co-investigator on two current and one completed NIHR funded studies that have, or have had, additional support from Stryker Ltd. He has accepted travel expenses and accommodation from the organisers for speaking at academic meetings.

\newpage
\bibliographystyle{plainnat}
\bibliography{references}

\newpage

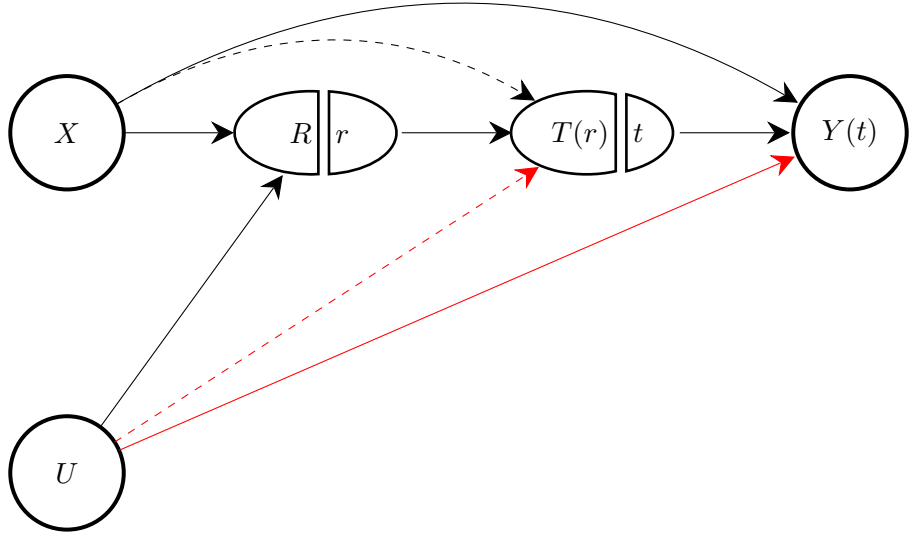
\begin{figure}
\centering

\begin{tikzpicture}
\tikzset{
  line width=1.5pt,
  outer sep=0pt,
  >={Stealth[length=4mm, width=4mm]},
  ell/.style={
    draw,
    fill=white,
    inner sep=0pt,
    line width=1.5pt,
    minimum size=15mm
  },
  swignode/.style={
    shape=swig vsplit,
    draw,
    fill=white,
    inner sep=2pt,
    minimum width=20mm,
    minimum height=11mm,
    text height=1.5ex,
    text depth=.25ex,
    line width=1pt,
    swig vsplit={
      gap=4pt,
      line color right=black
    }
  }
}

\node[name=x, ell, shape=circle] {$X$};

\node[name=u, ell, below=30mm of x, shape=circle] {$U$};

\node[name=r, swignode, right=15mm of x] {
  \nodepart{left}{$R$}
  \nodepart{right}{$r$}
};

\node[name=t, swignode, right=15mm of r] {
  \nodepart{left}{$T(r)$}
  \nodepart{right}{$t$}
};

\node[name=y, right=15mm of t, ell, shape=circle] {$Y(t)$};

\draw[->]
  (x) edge (r)
  (x) edge[out=30,in=150, dashed] (t)
  (x) edge[out=30,in=150] (y)
  (r) edge (t)
  (u) edge (r)
  (u) edge[red, dashed] (t)
  (u) edge[red] (y)
  (t) edge (y);

\end{tikzpicture}
\caption{SWIG for comprehensive cohort study with RCT and OBS arms.  Dashed arrows into $T(r)$ represent preference-based treatment selection in the OBS arm $(r=0)$. When $r=1$, treatment is randomized, so there is no $X \to T(r)$ or $U \to T(r)$ confounding path. The sensitivity parameter $\gamma_t$ governs the level of unmeasured confounding in the OBS arm, represented by the edge from $U$ to $T(r)$, and $U$ to $Y(t)$.}
\label{fig:swig}
\end{figure}

\begin{figure}[hbt!]
\centering
\includegraphics[width=0.5\textwidth,height=0.5\textheight,keepaspectratio]{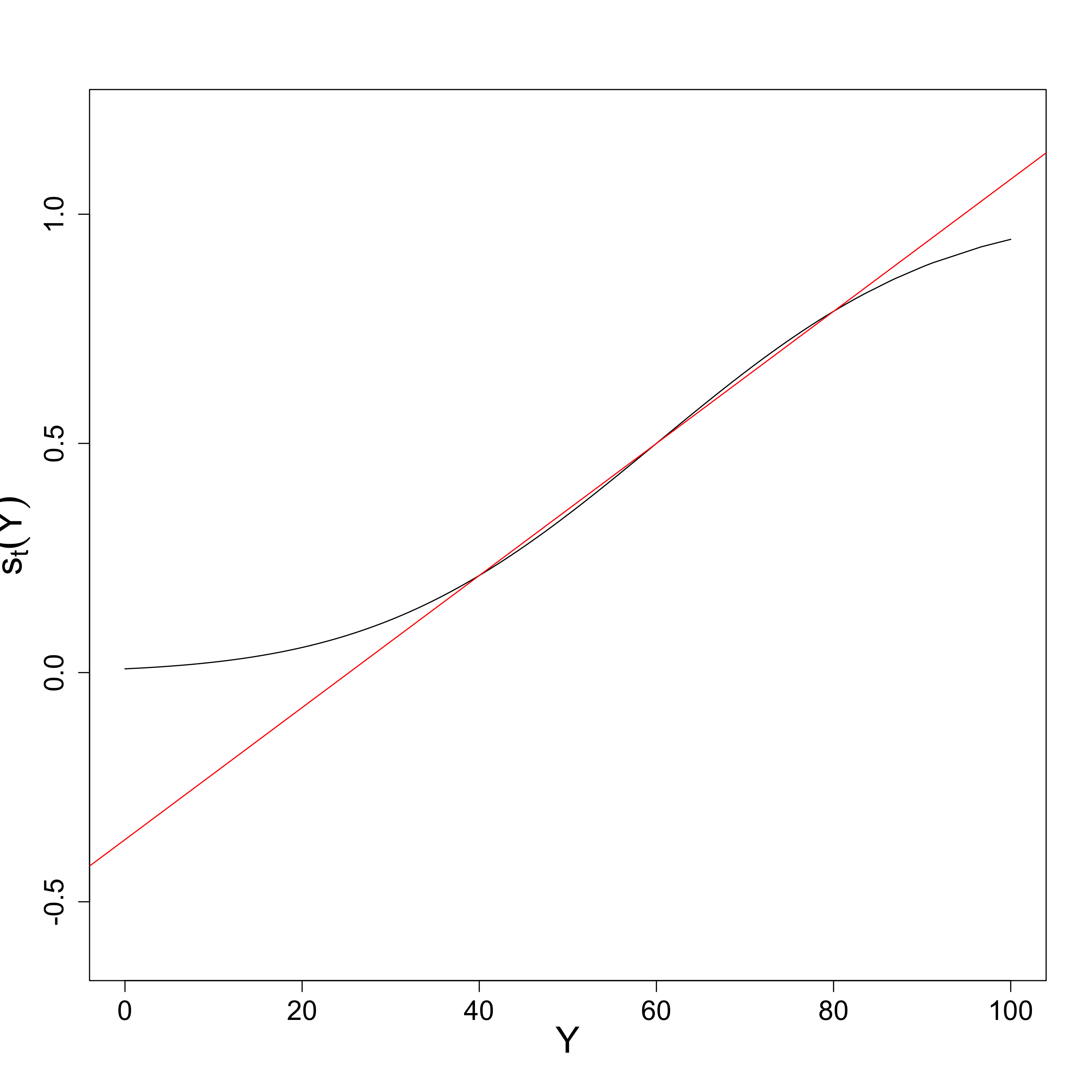}
\caption{$s_t(Y)=\Phi(\frac{Y-60}{25})$ (black curve). Red line is the linear line going through $(40, \Phi(\frac{40-60}{25}))$ and $(80, \Phi(\frac{80-60}{25}))$}
\label{fig_s_t_y}
\end{figure}

\begin{landscape}
\setlength\tabcolsep{5pt}
\begin{xltabular}{\textheight}{|ccccccc|}
\caption{Descriptive Statistics for covariates by study enrollment (RCTs versus OBS) and treatment (topical versus oral ibuprofen).}\label{Table1}\\
\toprule
& \multicolumn{3}{c}{\makecell{\large RCT}} & \multicolumn{3}{c|}{\makecell{\large OBS}}\\
\midrule
\textbf{Covariates} & \textbf{\makecell{ Topical\\ (N=133)}} & \textbf{ \makecell{Oral \\(N=140)}} & \textbf{\makecell{Overall \\ (N=273)}} & \textbf{\makecell{Topical\\ (N=215)}} & \textbf{ \makecell{Oral \\(N=75)}} & \textbf{ \makecell{Overall \\(N=290)}}\\
\midrule
\endfirsthead
\multicolumn{7}{c}%
{\tablename\ \thetable{} -- continued from previous page} \\
\toprule
& \multicolumn{3}{c}{\makecell{\large RCT}} & \multicolumn{3}{c|}{\makecell{\large OBS}}\\
\midrule
\textbf{Covariates} & \textbf{\makecell{ Topical\\ (N=133)}} & \textbf{ \makecell{Oral \\(N=140)}} & \textbf{\makecell{Overall \\ (N=273)}} & \textbf{\makecell{Topical\\ (N=215)}} & \textbf{ \makecell{Oral \\(N=75)}} & \textbf{ \makecell{Overall \\(N=290)}}\\
\midrule
\endhead

\midrule
\multicolumn{7}{|r|}{{Continued on next page}} \\ 
\bottomrule
\endfoot

\bottomrule
\endlastfoot

\multicolumn{1}{|l}{\textbf{Age (in years):} Mean (SD)} & 61.9 (8.4) & 62.3 (8.3) & 62.1 (8.3) & 65.8(8.0) & 64.0(8.8) & 65.3(8.2)\\
\multicolumn{7}{|l|}{\textbf{Expected pain a year from now}} \\
A little/Much worse & 73 (54.9\%)&64 (45.7\%)&137 (50.2\%)&96 (44.7\%)&40 (53.3\%)&136 (46.9\%)\\
About the same&33 (24.8\%)&42 (30.0\%)&75 (27.5\%)&59 (27.4\%)&23 (30.7\%)&82 (28.3\%)\\
A little/Much better/Free of pain&27 (20.3\%)&34 (24.3\%)&61 (22.3\%)&60 (27.9\%)& 12 (16.0\%)&72 (24.8\%)\\
\multicolumn{1}{|l}{\textbf{Baseline WOMAC pain score:} Mean (SD)} &38.0 (19.2) &39.5 (21.6)&38.8 (20.5)&41.3 (20.1)&38.6 (19.4)&40.6 (20.0)\\
\multicolumn{7}{|l|}{\textbf{Chronic pain grade}} \\
Grade 1-2&99 (74.4\%)&90 (64.3\%) & 189 (69.2\%) & 150 (69.8\%) & 50 (66.7\%) & 200 (69.0\%)\\
Grade 3-4&34 (25.6\%) & 50 (35.7\%) & 84 (30.8\%) & 65 (30.2\%) & 25 (33.3\%) & 90 (31.0\%)\\
\end{xltabular}
\end{landscape}

\begin{table}[ht!]
\caption{Simulation results averaged over 2000 replicates. Percent bias, 95\% Wald coverage and 95\% Bootstrap coverage.}
\centering
\scriptsize
\setlength{\tabcolsep}{3pt}
\begin{adjustbox}{max width=\linewidth}
\begin{threeparttable}
\begin{tabular}{r|cc|cc|cc|cc|cc|cc|cc|cc|cc}
\toprule
\multicolumn{1}{c}{} & \multicolumn{6}{c}{Percent bias (\%)} & \multicolumn{6}{c}{95\% Wald coverage} & \multicolumn{6}{c}{95\% Bootstrap coverage} \\
\cmidrule(lr){2-7} \cmidrule(lr){8-13} \cmidrule(lr){14-19}
\multicolumn{1}{c}{$\gamma_t$}
& \multicolumn{2}{c}{$\psi_t(\widetilde{P};\gamma_t)$}
& \multicolumn{2}{c}{$\psi_{t,1}(\widetilde{P})$}
& \multicolumn{2}{c}{$\psi_{t,0}(\widetilde{P};\gamma_t)$}
& \multicolumn{2}{c}{$\psi_t(\widetilde{P};\gamma_t)$}
& \multicolumn{2}{c}{$\psi_{t,1}(\widetilde{P})$}
& \multicolumn{2}{c}{$\psi_{t,0}(\widetilde{P};\gamma_t)$}
& \multicolumn{2}{c}{$\psi_t(\widetilde{P};\gamma_t)$}
& \multicolumn{2}{c}{$\psi_{t,1}(\widetilde{P})$}
& \multicolumn{2}{c}{$\psi_{t,0}(\widetilde{P};\gamma_t)$}\\
\cmidrule(lr){2-3} \cmidrule(lr){4-5} \cmidrule(lr){6-7}
\cmidrule(lr){8-9} \cmidrule(lr){10-11} \cmidrule(lr){12-13}
\cmidrule(lr){14-15} \cmidrule(lr){16-17} \cmidrule(lr){18-19}
\multicolumn{1}{c}{} & Topical & \multicolumn{1}{c}{Oral}
& Topical & \multicolumn{1}{c}{Oral}
& Topical & \multicolumn{1}{c}{Oral}
& Topical & \multicolumn{1}{c}{Oral}
& Topical & \multicolumn{1}{c}{Oral}
& Topical & \multicolumn{1}{c}{Oral} 
& Topical & \multicolumn{1}{c}{Oral}
& Topical & \multicolumn{1}{c}{Oral}
& Topical & Oral\\
\midrule

\multicolumn{19}{c}{\textbf{Sample Size = 563}\tnote{1}} \\
\midrule
-2 & -0.30 & -0.69 & \multirow{9}{*}{-0.30} & \multirow{9}{*}{-0.15} & -0.30 & -1.20 & 0.95 & 0.95 & \multirow{9}{*}{0.95} & \multirow{9}{*}{0.95} & 0.95 & 0.95 & 0.96 & 0.97 & \multirow{9}{*}{0.96} & \multirow{9}{*}{0.96} & 0.96 & 0.97\\
-1.5 & -0.27 & -0.70 &  &  & -0.26 & -1.20 & 0.95 & 0.95 &  &  & 0.95 & 0.95 & 0.96 & 0.97 & & & 0.96 & 0.97\\
-1 & -0.24 & -0.70 &  & & -0.20 & -1.19 & 0.95 & 0.95 &  &  & 0.96 & 0.95 & 0.96 & 0.97 & & & 0.96 & 0.97\\
-0.5 & -0.21 & -0.71 &  &  & -0.14 & -1.19 & 0.95 & 0.95 &  &  & 0.96 & 0.95 & 0.96 & 0.97 & & & 0.96 & 0.97 \\
0  &  -0.17 & -0.74 &  &  & -0.07 & -1.23 & 0.95 & 0.94 &  &  & 0.96 & 0.94 & 0.96 & 0.97 & & & 0.96 & 0.97  \\
0.5 & -0.13 & -0.81 &  &  & 0.01 & -1.35 & 0.95 & 0.94 &  &  & 0.96 & 0.93 & 0.96 & 0.96 & & & 0.96 & 0.97 \\
1  &  -0.09 & -0.94 &  &  & 0.08 & -1.56 & 0.95 & 0.93 &  &  & 0.96 & 0.92 & 0.96 & 0.96 & & & 0.96 & 0.96\\
1.5  &  -0.05 & -1.12 &  &  & 0.15 & -1.87 & 0.95 & 0.92 &  &  & 0.96 & 0.90 & 0.96 & 0.96 & & & 0.96 & 0.96 \\
2  &  -0.02 & -1.36 &  & &  0.20 & -2.26 & 0.95 & 0.90 &  &  & 0.95 & 0.89 & 0.96 & 0.95 & & & 0.96 & 0.94 \\

\midrule
\multicolumn{19}{c}{\textbf{Sample Size = 1000}\tnote{1}} \\
\midrule
-2 & -0.17 & -0.37 & \multirow{9}{*}{-0.18} & \multirow{9}{*}{-0.13} & -0.16 & -0.59 & 0.95 & 0.95 & \multirow{9}{*}{0.95} & \multirow{9}{*}{0.95} & 0.95 &  0.95 & 0.95 & 0.96 & \multirow{9}{*}{0.95} & \multirow{9}{*}{0.95} & 0.96 & 0.96\\
-1.5 & -0.16 & -0.34 &  &  & -0.13 & -0.53 & 0.95 & 0.95 &  &  & 0.95 & 0.95 & 0.96 & 0.96 & & & 0.96 & 0.96\\
-1 & -0.14 & -0.29 &  & & -0.10 & -0.43 & 0.95 & 0.96 &  &  & 0.95 & 0.95 & 0.95 & 0.96 & & & 0.96 & 0.97\\
-0.5 & -0.13 & -0.24 &  &  & -0.07 & -0.33 & 0.95 & 0.96 &  &  & 0.96 & 0.95 & 0.95 & 0.96 & & & 0.96 & 0.96 \\
0  &  -0.11 & -0.20 &  &  & -0.04 & -0.25 & 0.95 & 0.96 &  &  & 0.96 & 0.95 & 0.95 & 0.96 & & & 0.96 & 0.97  \\
0.5 & -0.09 & -0.19 &  &  & -0.01 & -0.23 & 0.95 & 0.95 &  &  & 0.96 & 0.95 & 0.95 & 0.97 & & & 0.96 & 0.96\\
1  &  -0.07 & -0.22 &  &  & 0.03 & -0.28 & 0.95 & 0.95 &  &  & 0.95 & 0.95 & 0.95 & 0.96 & & & 0.96 & 0.96 \\
1.5  &  -0.05 & -0.31 &  &  & 0.06 & -0.44 & 0.95 & 0.94 &  &  & 0.96 & 0.94 & 0.95 & 0.96 & & & 0.96 & 0.96 \\
2  &  -0.03 & -0.45 &  & &  0.10 & -0.69 & 0.95 & 0.94 &  &  & 0.96 & 0.93 & 0.95 & 0.96 & & & 0.96 & 0.96\\

\midrule
\multicolumn{19}{c}{\textbf{Sample Size = 2000}} \\
\midrule
-2 & -0.03 & -0.21 & \multirow{9}{*}{-0.03} & \multirow{9}{*}{0.01} & -0.02 & -0.41 & 0.96 & 0.95 & \multirow{9}{*}{0.96} & \multirow{9}{*}{0.95} & 0.94 & 0.96 & 0.96 & 0.96 & \multirow{9}{*}{0.96} & \multirow{9}{*}{0.95} & 0.95 & 0.97\\
-1.5 & -0.02 & -0.17 &  &  & 0.00 & -0.35 & 0.96 & 0.95 &  &  & 0.95 & 0.96 & 0.96 & 0.96 & & & 0.95 & 0.97\\
-1 & -0.01 & -0.13 &  & & 0.01 & -0.26 & 0.95 & 0.95 &  &  & 0.95 & 0.96 & 0.96 & 0.96 & & & 0.95 & 0.96\\
-0.5 & 0.00 & -0.07 &  &  & 0.03 & -0.15 & 0.95 & 0.95 &  &  & 0.95 & 0.96 &  0.96 & 0.96 & & & 0.95 & 0.96 \\
0  &  0.01 & -0.02 &  &  & 0.04 & -0.05 & 0.95 & 0.95 &  &  & 0.95 & 0.96 &  0.96 & 0.96 & & & 0.95 & 0.96 \\
0.5 & 0.02 & 0.02 &  &  & 0.06 & 0.03 & 0.96 & 0.95 &  &  & 0.95 & 0.95  & 0.96 & 0.96 & & & 0.95 & 0.96\\
1  &  0.03 & 0.04 &  &  & 0.08 & 0.06 & 0.96 & 0.95 &  &  & 0.95 &  0.95 & 0.96 & 0.96 & & & 0.95 & 0.96\\
1.5  &  0.04 & 0.02 &  &  & 0.10 & 0.03 & 0.96 & 0.95 &  &  & 0.95 & 0.95 & 0.96 & 0.96 & & & 0.95 & 0.96 \\
2  &  0.05 & -0.04 &  & &  0.12 & -0.08 & 0.96 & 0.95 &  &  & 0.95 & 0.94 & 0.96 & 0.96 & & & 0.95 &  0.96\\

\bottomrule
\end{tabular}
\begin{tablenotes}
\footnotesize
\item[1] For sample size 563, 2 (0.1\%) simulation trials for topical NSAIDs ($t=1$) and 50 (2.5\%) simulation trials for oral NSAIDs ($t=0$) have bootstrap samples that fail to fit the single index model under both initial coefficients estimation methods. For sample size 1000, 1 (0.05\%) simulation trial for oral NSAIDs ($t=0$) have bootstrap samples that fail to fit the single index model under both initial coefficients estimation methods. We exclude these simulation trials from the reported results. 
\end{tablenotes}
\end{threeparttable}
\end{adjustbox}
\label{coverage}
\end{table}

\begin{figure}[hbt!]
\centering
\begin{subfigure}{0.8\textwidth}
\includegraphics[width=\textwidth, height=\textheight, keepaspectratio]{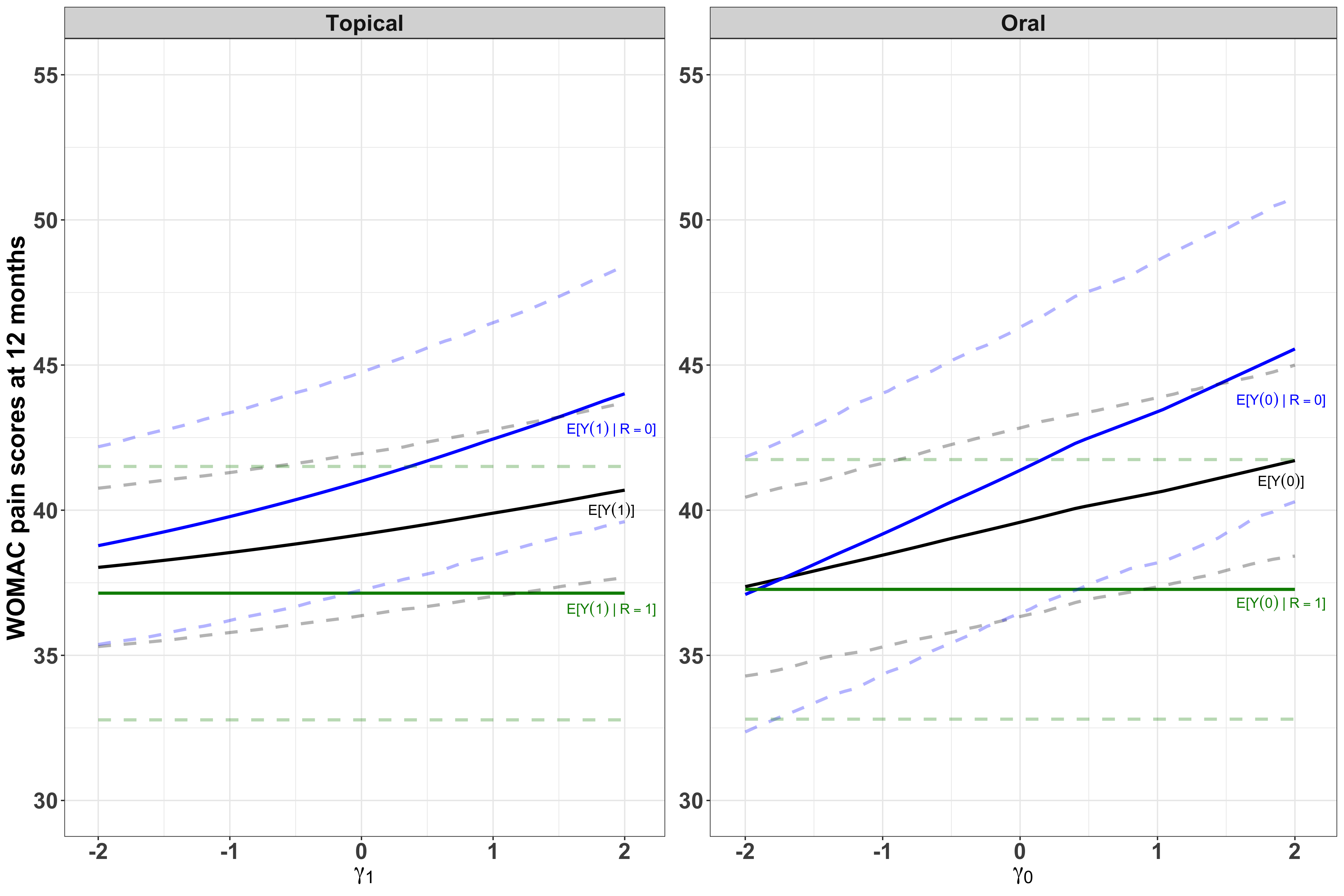}
\caption{Estimates as functions of $\gamma_t$}
\label{fig1}
\end{subfigure}
\hfill
\begin{subfigure}{0.8\textwidth}
\includegraphics[width=\linewidth]{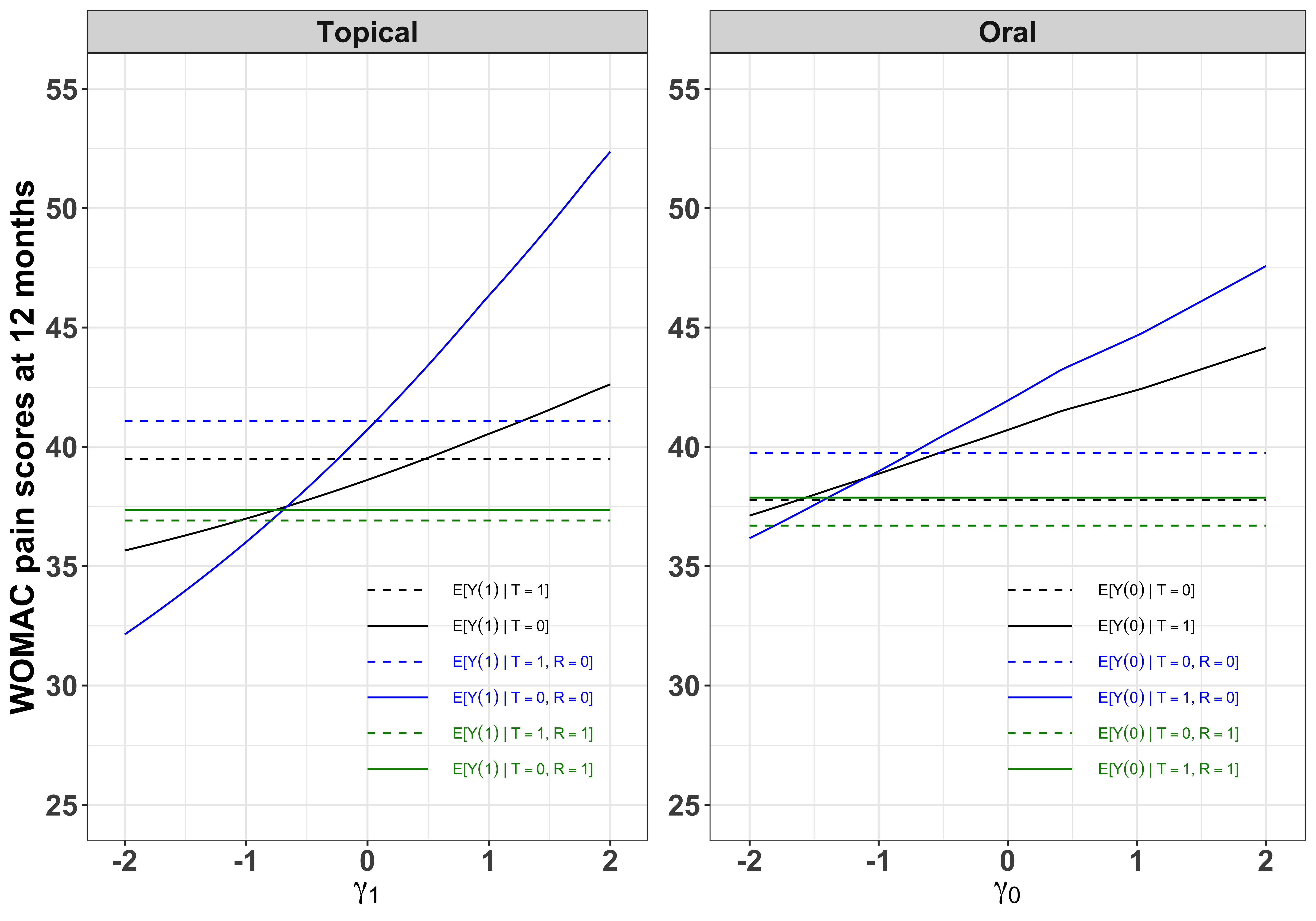}
\caption{Induced estimates}
\label{fig2}
\end{subfigure}
\caption{(a). Estimates of $\E[Y(t)]$, $\E[Y(t)|R=0]$ and $\E[Y(t)|R=1]$ as functions of $\gamma_t$. Pointwise symmetric-t $95\%$ confidence interval (dashed lines) included. (b). Induced estimates of WOMAC pain score at 12 months under receiving a recommendation of topical NSAIDs for those who actually received an oral NSAIDs recommendation (left panel), and under receiving a recommendation of oral NSAIDs for those who actually received a  topical NSAIDs recommendation (right panel) in the entire cohort, OBS, and RCT.} 
\end{figure}

\begin{figure}[hbt!]
\begin{subfigure}{0.5\textwidth}
\centering
\includegraphics[width=\textwidth, height=\textheight, keepaspectratio]{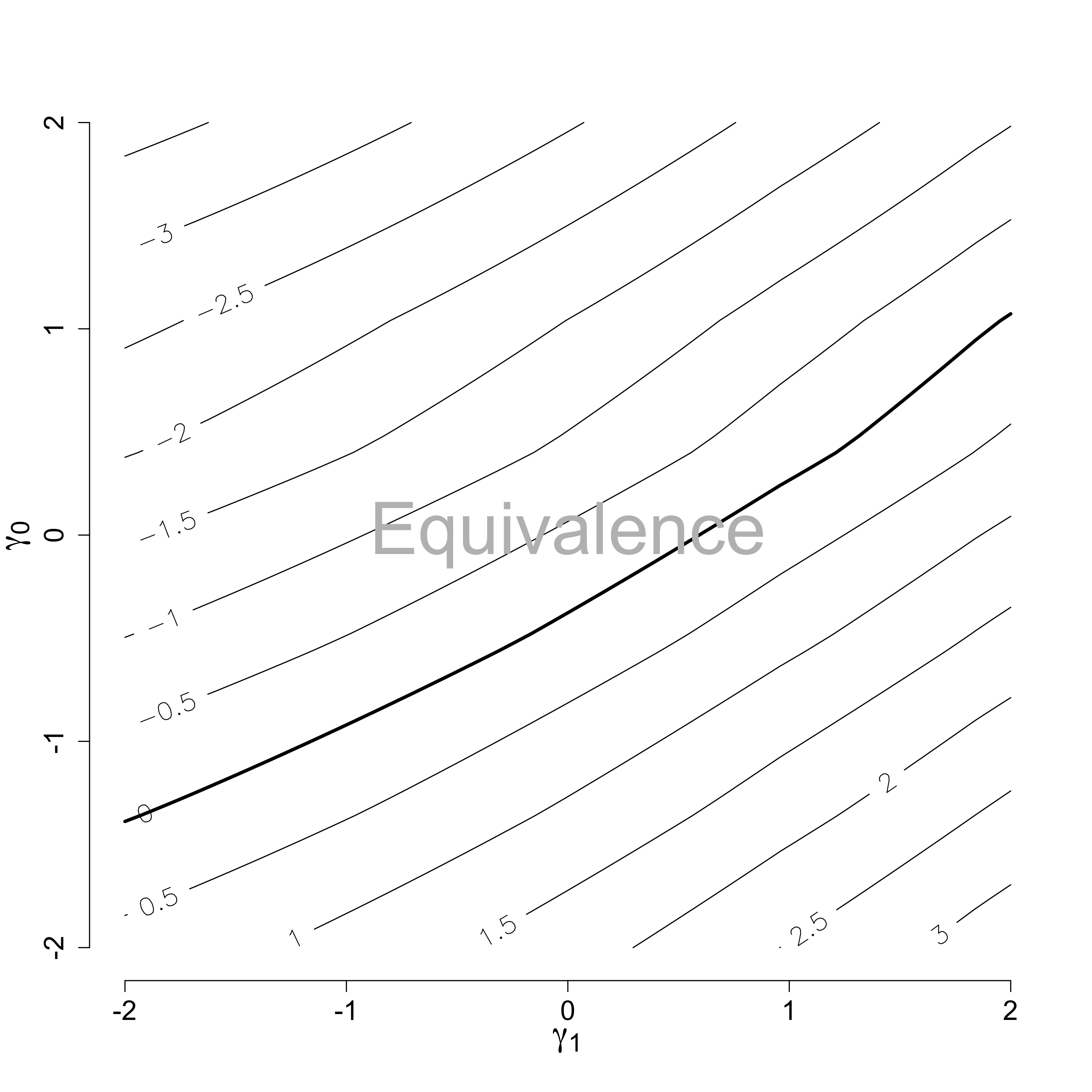}
\caption{Contour plot of $\E[Y(1)]-\E[Y(0)]$}
\end{subfigure}
\hfill
\begin{subfigure}{0.5\textwidth}
\centering
\includegraphics[width=\linewidth]{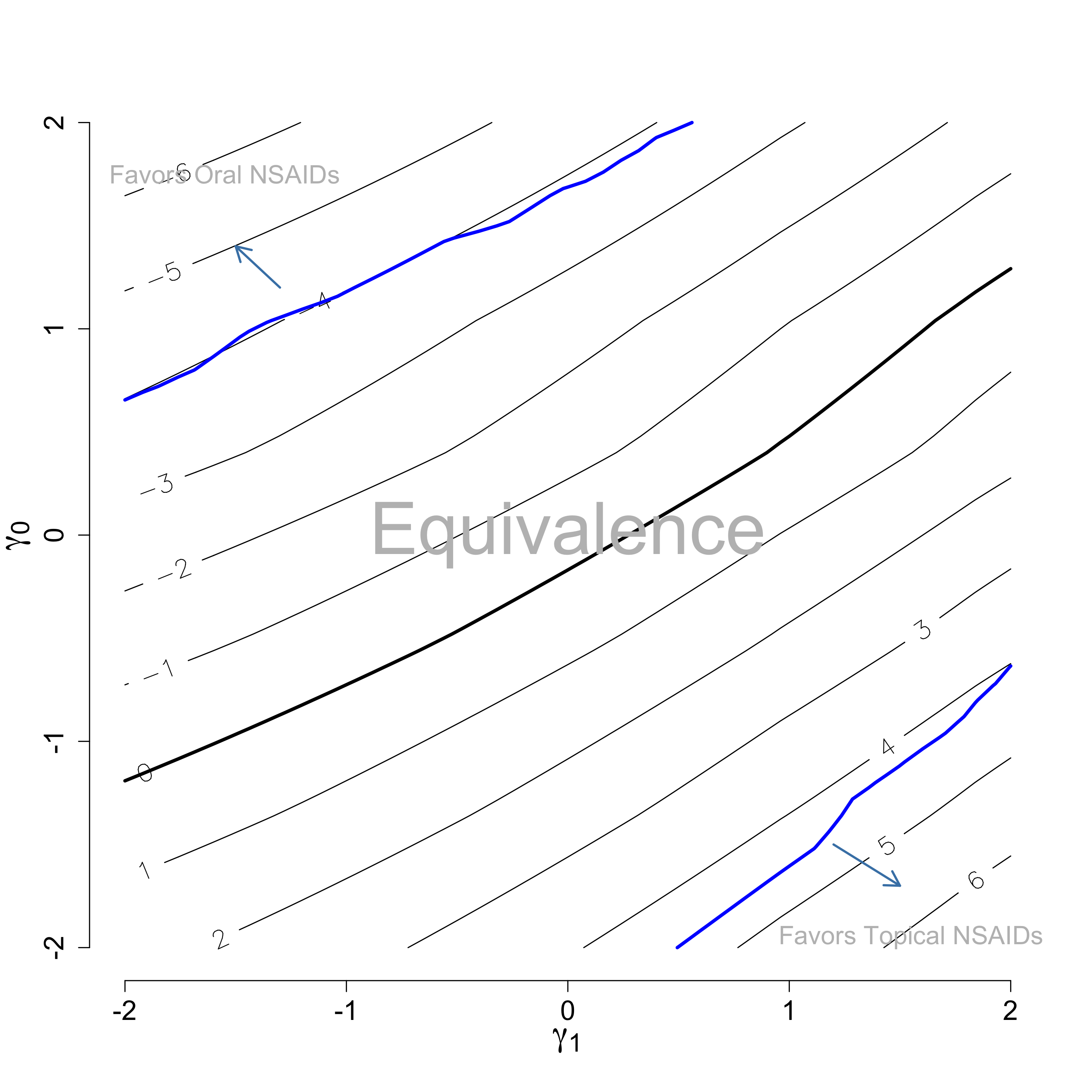}
\caption{Contour plot of $\E[Y(1)|R=0]-\E[Y(0)|R=0]$}
\end{subfigure}
\caption{Contour plot of estimated treatment effects: (a) among all individuals in the comprehensive cohort; (b). among individuals enrolled in the OBS arm, under different values of $\gamma_1$ and $\gamma_0$. The bold black curve indicate combinations of $\gamma_1$ and $\gamma_0$ that would lead to estimated treatment effects of zero. The area defined by the blue curves indicate the range of $\gamma_1$ and $\gamma_0$ with statistically equivalent effects. }
\label{fig3}
\end{figure}

\begin{figure}[hbt!]
\centering
\begin{subfigure}{0.7\textwidth}
\centering
\includegraphics[width=\textwidth, height=\textheight, keepaspectratio]{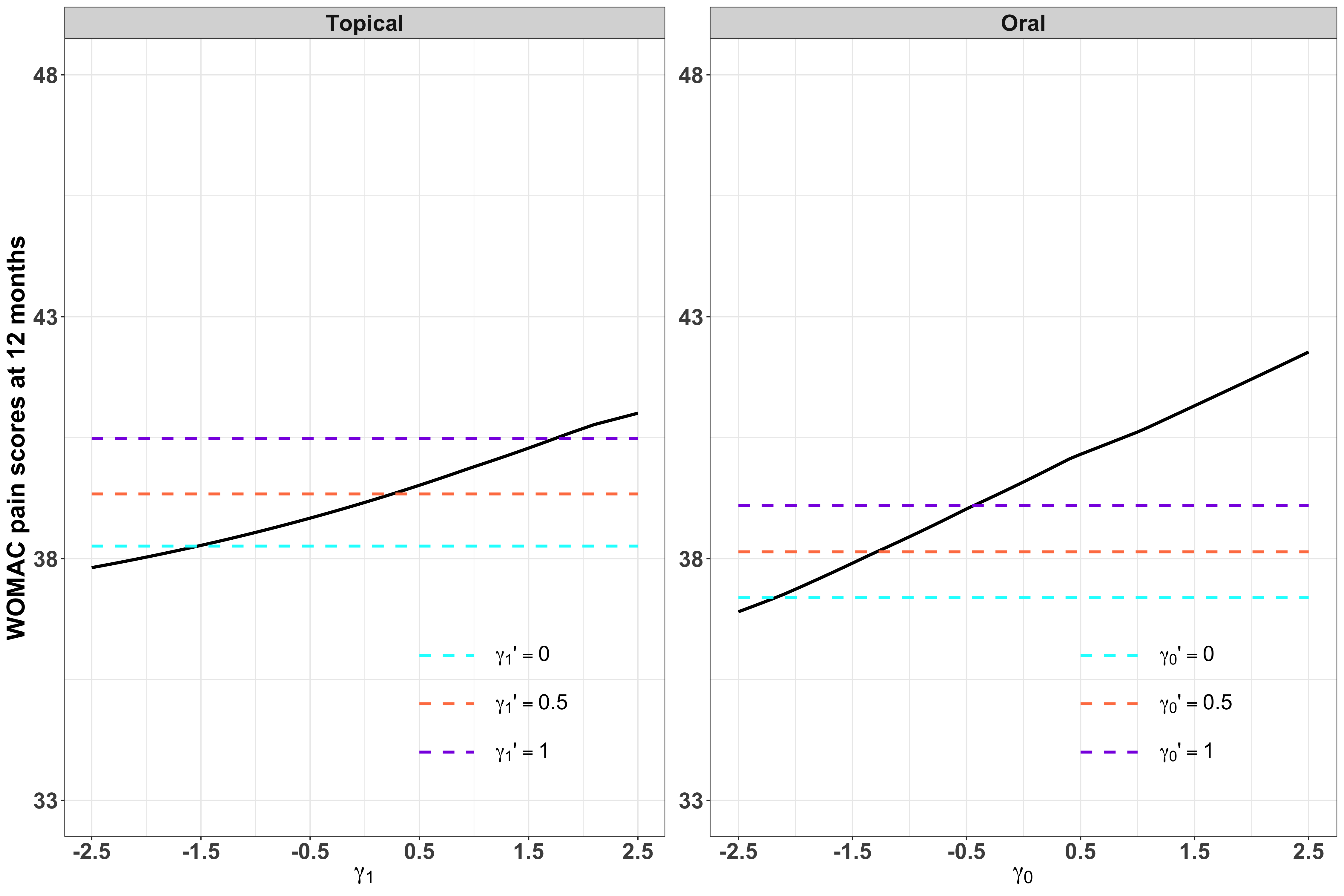}
\caption{Estimated $\E[Y(t)]$ under Assumption (A4) (solid lines) and under (\ref{exchangeability}) (dashed lines). }
\end{subfigure}
\hfill
\begin{subfigure}{0.7\textwidth}
\centering
\includegraphics[width=\linewidth]{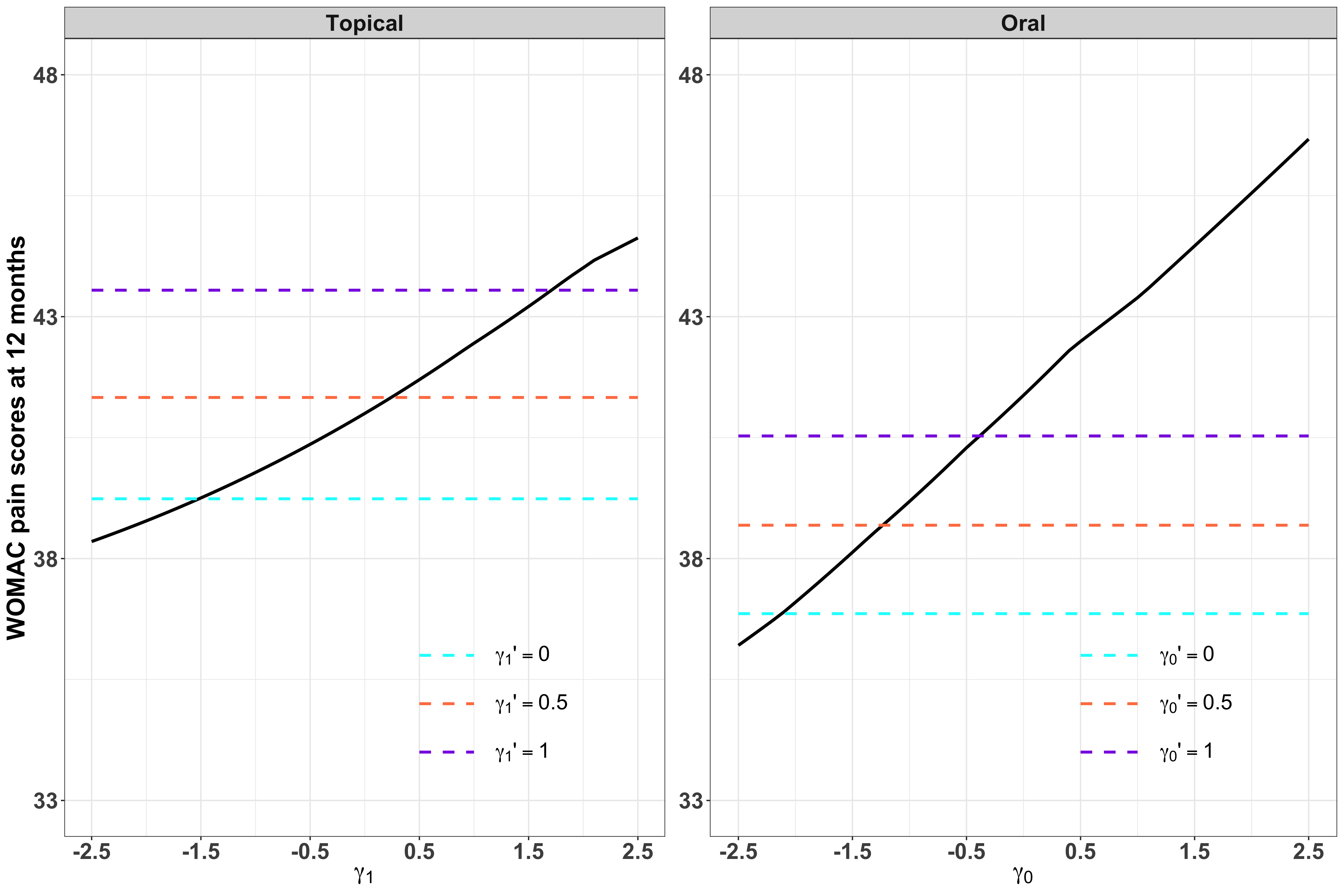}
\caption{Estimated $\E[Y(t)|R=0]$ under Assumption (A4) (solid lines) and under (\ref{exchangeability}) (dashed lines). }
\end{subfigure}
\caption{Relationship between $\gamma_t$ and $\gamma_t'$}
\label{fig4}
\end{figure}

\clearpage
\appendix  

\noindent {\bf \Large Appendix} \vspace{0.25cm}

\noindent The appendix is structured as follows. 
Appendix~\ref{app:notation} offers a summary of the notations used throughout the manuscript. Appendix~\ref{app:proofs} provides proofs for estimation and inference for $\E[Y(t)]$ under Assumptions (A1-A5). Appendix~\ref{app:proofs_R0R1} presents the influence function and main steps for estimating $\E[Y(t)|R=0]$ and $\E[Y(t)|R=1]$. Appendix~\ref{app:model_2} present the estimation procedures for $\E[Y(t)]$ and $\E[Y(t)|R=0]$ under Assumptions (A1-A3), (A5) and (\ref{exchangeability}). Appendix~\ref{app:data_analysis} contains the identification and estimation discussions of $\E[Y(t)|1-t]$,  $\E[Y(t)|1-t, R=0]$ and $\E[Y(t)|1-t, R=1]$.

\section{Notation} 
\label{app:notation} 

We provide a comprehensive list of notation in Table~\ref{tab:notations}. 

\begin{table}[!h]
\begin{center}
\caption{\centering Glossary of terms and notations}
\label{tab:notations}
\addtolength{\tabcolsep}{8pt}
{\small
\begin{tabular}{ ll} 
    \hline  
    \textbf{Symbol} & \textbf{Definition} 
    \\ \hline 
     $T$ & Treatment indicator \\ 
     $R$ & Randomization consent indicator (1 for RCT, 0 for OBS)  \\ 
     $X$ & Baseline covariates \\ 
     $Y$ & Outcome  \\ 
     $M$ & Outcome missingness indicator ($M=1$ if $Y$ is observed, $M=0$ if $Y$ is missing)\\
     $Y(t)$ & Counterfactual outcome under treatment $t$\\
     $F_{t, r}(y|x)$ & $P(Y\leq y\mid T=t, R=r, X=x)$\\
     $\widetilde{F}_{t, r}(y|x)$ & $P(Y\leq y\mid T=t, R=r, X=x, M=1)$\\
     $\mu_{t,r}(h(Y); x)$ & $\E[h(Y) \mid T=t, R=r, X=x]$\\
     $\widetilde{\mu}_{t,r}(h(Y); x)$ & $\E[h(Y) \mid T=t, R=r, X=x, M=1]$\\
     $\pi_{t,r}(x)$ & $P(T = t \mid R=r, X=x)$\\
     $g_r(x)$ & $P(R = r \mid X=x)$\\
     $F(x)$ & $P(X\leq x)$\\
     $\eta_m(x, r, t)$ & $P(M =m \mid X=x, R=r, T=t)$\\
         \hline  
\end{tabular}
}
\end{center}
\end{table}

\section{Proofs for $\E[Y(t)]$ under Assumptions (A1-A5)}
\label{app:proofs}

\subsection{Identification Functional in \eqref{eq:parameter_id1}}
\label{app:proofs_ID}

\begin{align*}
\psi_t(\widetilde{P}; \gamma_t) &=  \E[Y(t)|R=1]P(R=1)+\E[Y(t)|R=0]P(R=0)\\
&= \sum_{r\in\{0, 1\}}\E[\E[Y(t)\mid R=r, X]\mid R=r]P(R=r)
\end{align*}
Under Assumptions (A1), (A3) and (A5), 
\begin{align*}
    \E[Y(t)\mid R=1, X] &= \E[Y(t)\mid T=t, R=1, X]\\
    &= \E[Y(t)\mid T=t, R=1, X, M=1]\\
    &= \widetilde{\mu}_{t, 1}(Y; X)
\end{align*}
and under Assumptions (A1), (A4) and (A5), 
\begin{align*}
    \E[Y(t)\mid R=0, X] &= \E[Y(t)\mid T=t, R=0, X]\pi_{t, 0}(X)+\E[Y(t)\mid T=1-t, R=0, X]\pi_{1-t, 0}(X)\\
    &=\E[Y\mid T=t, R=0, X, M=1]\pi_{t, 0}(X)+\pi_{1-t, 0}(X)\int y(t) dF(y(t)\mid T=1-t, R=0, X=X)\\
    &= \widetilde{\mu}_{t, 0}(Y; X)\pi_{t, 0}(X)+\pi_{1-t, 0}(X)\int \frac{ y(t)\exp\{  \gamma_t s_t(y(t)) \} }{ \mu_{t, 0} (\exp\{  \gamma_t s_t(Y(t)) \}; X)}dF(y(t)\mid T=t, R=0, X)\\
    &= \widetilde{\mu}_{t, 0}(Y; X)\pi_{t, 0}(X)+\pi_{1-t, 0}(X)\int \frac{ y\exp\{  \gamma_t s_t(y) \} }{\widetilde{\mu}_{t, 0} (\exp\{  \gamma_t s_t(Y) \}; X)}dF(y\mid T=t, R=0, X, M=1)\\
    &= \widetilde{\mu}_{t, 0}(Y; X)\pi_{t, 0}(X)+\frac{\widetilde{\mu}_{t, 0}(Y\exp\{  \gamma_t s_t(Y) \}; X)}{\widetilde{\mu}_{t, 0} (\exp\{  \gamma_t s_t(Y) \}; X)}\pi_{1-t, 0}(X)
\end{align*}
Thus, 
\begin{equation*}
\psi_t(\widetilde{P}; \gamma_t) = \psi_{t,1}(\widetilde{P}) P(R=1) + \psi_{t,0}(\widetilde{P}; \gamma_t) P(R=0),
\end{equation*}
where 
\[
\psi_{t,1}(\widetilde{P}) = \E[ \widetilde{\mu}_{t,1}(Y;X) | R=1].
\]
and
\[
\psi_{t,0}(\widetilde{P}; \gamma_t) = \E \bigg[ \Big\{ \widetilde{\mu}_{t,0}(Y;X) \pi_{t,0}(X) + \frac{ \widetilde{\mu}_{t,0}(Y \! \exp\{  \gamma_t s_t(Y)\} ;X) }{\widetilde{\mu}_{t,0}( \exp\{  \gamma_t s_t(Y)\} ;X) } \pi_{1-t,0}(X)  \Big\} | R=0\bigg].
\]

\subsection{Nonparametric Efficient Influence Function} 
\label{app:proofs_EIF}

Here we introduce an alternative to Assumption (A3):
\begin{itemize}
    \setlength{\itemsep}{0.cm} 
    \item[(A3')\ ] \underline{Conditional ignorability under RCT}: Treatment assignment is independent of the potential outcomes given covariates in the RCT . That is, $T \perp Y(t) \mid X, R=1$, for $t = 0, 1$. 
\end{itemize} 

We start by first deriving the non-parametric influence function for $\psi_t(P; \gamma_t)$ in a  model governed by Assumptions (A1-A2), (A3') and (A4) and no missing outcome data. Then we argue that under Assumption (A3), the form of influence function stays the same. Finally, we derive the influence function for $\psi_t(\widetilde{P}; \gamma_t)$ in a  model governed by Assumptions (A1-A5) and missing outcome data.

\textbf{Proof of Theorem \ref{EIF_miss}}
Under Assumptions (A1-A2), (A3') and (A4) and no missing outcome data, we have
\begin{equation*}
\begin{aligned}
    \psi_t(P; \gamma_t) &= \E[\E[Y(t)|X]]\\
    &=\sum_{r\in\{0, 1\}}\E[\E[Y(t)|R=r, X]g_r(X)]\\
    &=\sum_{r\in\{0, 1\}}\E[\E[Y(t)|T=t, R=r, X]g_r(X)]\\
    &=\E[\E[Y|T=t, R=1, X]g_1(X)]\\
    &\hspace{1cm}+\E\bigg[g_0(X)\Big\{\E[Y\mid T=t, R=0, X]\pi_{t, 0}(X)+\E[Y(t)\mid T=1-t, R=0, X]\pi_{1-t, 0}(X)\Big\}\bigg]\\
    &=\E \big[ g_1(X) \, \mu_{t,1}(Y;X)\big] \\
    &\hspace{1cm}+\E \bigg[g_0(X) \Big\{ \mu_{t,0}(Y;X) \pi_{t,0}(X) +  \pi_{1-t,0}(X)\int y(t) dF(y(t)\mid T=1-t, R=0, X=X)  \Big\} \bigg]\\
    &=\E \big[ g_1(X) \, \mu_{t,1}(Y;X)\big] \\
    &\hspace{1cm}+\E \bigg[g_0(X) \Big\{ \mu_{t,0}(Y;X) \pi_{t,0}(X) +  \pi_{1-t,0}(X)\int \frac{ y\exp\{  \gamma_t s_t(y) \} }{ \mu_{t, 0} (\exp\{  \gamma_t s_t(Y(t)) \}; X)}dF(y\mid T=t, R=0, X) \Big\} \bigg]\\
    &=\E \big[ g_1(X) \, \mu_{t,1}(Y;X)\big] \\
    &\hspace{1cm}+\E \bigg[g_0(X) \Big\{ \mu_{t,0}(Y;X) \pi_{t,0}(X) + \frac{ \mu_{t,0}(Y \! \exp\{  \gamma_t s_t(Y)\} ;X) }{\mu_{t,0}( \exp\{  \gamma_t s_t(Y)\} ;X) } \pi_{1-t,0}(X)  \Big\} \bigg]\\
    &=\int \bigg(g_1(x)\mu_{t, 1}(Y; x)+g_0(x)\bigg\{\mu_{t, 0}(Y; x)\pi_{t, 0}(x)+\frac{\mu_{t, 0}(Y\exp\{  \gamma_t s_t(Y) \}; x)}{\mu_{t, 0} (\exp\{  \gamma_t s_t(Y) \}; x)}\pi_{1-t, 0}(x)\bigg\}\bigg)dF(x).
\end{aligned}
\end{equation*}

Consider a statistical model $\mathcal{M}$ containing distributions $P^\ast$: $P^\ast\in \mathcal{M}$. Consistent with the main text, $P$ is the true distribution. $P^\ast$ is characterized by $F^\ast(x) = P^\ast(X \leq x)$, $g^\ast_{r}(x)=P^\ast(R=r\mid X=x)$, $\pi^\ast_{t, r}(x)=P^\ast(T=t\mid R=r, X=x)$ and $F^\ast_{t, r}(y\mid x)=P^\ast(Y\leq y\mid T=t, R=r, X=x)$. Let $\{ P^\ast_{\theta}: P^\ast_{\theta} \in \mathcal{M}\}$. We consider the following parametric submodels:
\begin{align*}
dF^\ast_{\theta}(x)&=dF(x)\{1+\epsilon h(x)\}\\
g^\ast_{r, \theta}(x) &= \frac{\{ g_1(x) \exp \{ \delta l(x) \} \}^r g_0(x)^{1-r}  }{ g_1(x) \exp \{ \delta l(x) \} + g_0(x)}\\
\pi^\ast_{t, 0, \theta}(x) &= \frac{\{ \pi_{1, 0}(x) \exp \{ \alpha j(x) \} \}^t \pi_{0, 0}(x)^{1-t}  }{ \pi_{1, 0}(x) \exp \{ \alpha j(x) \} + \pi_{0, 0}(x)}\\
\pi^\ast_{t, 1, \theta}(x) &= \frac{\{ \pi_{1, 1}(x) \exp \{ \beta m(x) \} \}^t \pi_{0, 1}^{1-t}(x)  }{ \pi_{1, 1}(x) \exp \{ \beta m(x) \} + \pi_{0, 1}(x)}\\
dF^\ast_{t, r, \theta}(y\mid x) &= dF_{t, r}(y\mid x)\{1+\eta_{t, r}k_{t, r}(y, x)\}
\end{align*}
where $\theta=(\epsilon, \delta, \alpha, \eta_{1, 1}, \eta_{0, 1}, \eta_{1, 0}, \eta_{0, 0})$, $\mathbb{E}[h(X)]=0$, $\mathbb{E}[k_{t, r}(Y,X) \mid T=t, R=r, X]=0$, $l(X)$, $m(X)$ and $j(X)$ are any functions of $X$. The associated score functions are $h(X)$, $\{R - g_1(X)\}l(X)$, $R\{T-\pi_{1, 1}(X)\}m(X)$, $(1-R)\{T-\pi_{1, 0}(X)\}j(X)$, $\I(T=1, R=1)k_{1, 1}(Y,X)$, $\I(T=0, R=1)k_{0, 1}(Y,X)$, $\I(T=1, R=0)k_{1, 0}(Y,X)$ and $\I(T=0, R=0)k_{0, 0}(Y,X)$. 


We can express the target parameter as a function of $P^\ast_{\theta}$:
\begin{align*}
    \psi_t(P^\ast_{\theta}; \gamma_t) = \int_x \bigg(&\bigg[\int_y ydF^\ast_{t, 1, \theta}(y\mid x)\bigg]g^\ast_{1, \theta}(x)+\\
    &g^\ast_{0, \theta}(x)\bigg\{\bigg[\int_y ydF^\ast_{t, 0, \theta}(y\mid x)\bigg]\pi^\ast_{t, 0, \theta}(x)+\frac{\int_y y\exp\{  \gamma_t s_t(y) \}dF^\ast_{t, 0, \theta}(y\mid x)}{\int_y \exp\{  \gamma_t s_t(y) \}dF^\ast_{t, 0, \theta}(y\mid x)}\pi^\ast_{1-t, 0, \theta}(x)\bigg\}\bigg)dF^\ast_{\theta}(x).
\end{align*}
The derivative of $\psi_t(P^\ast_{\theta}; \gamma_t)$ with respect to $\epsilon$ evaluated at $\theta=0$ is
\begin{align*}
    \frac{\partial \psi_t(P^\ast_{\theta}; \gamma_t)}{\partial\epsilon}\bigg|_{\theta=0}=\int_x \bigg(\mu_{t, 1}(Y; x)g_1(x)+g_0(x)\bigg\{\mu_{t, 0}(Y; x)\pi_{t, 0}(x)+\frac{\mu_{t, 0}(Y\exp\{  \gamma_t s_t(Y) \}; x)}{\mu_{t, 0} (\exp\{  \gamma_t s_t(Y) \}; x)}\pi_{1-t, 0}(x)\bigg\}\bigg)h(x)dF(x).
\end{align*}
The derivative of $\psi_t(P^\ast_{\theta}; \gamma_t)$ with respect to $\delta$ evaluated at $\theta=0$ is
\begin{align*}
    \frac{\partial \psi_t(P^\ast_{\theta}; \gamma_t)}{\partial\delta}\bigg|_{\theta=0} = \int_x\bigg(\mu_{t, 1}(Y; x)-\bigg\{\mu_{t, 0}(Y; x)\pi_{t, 0}(x)+\frac{\mu_{t, 0}(Y\exp\{  \gamma_t s_t(Y) \}; x)}{\mu_{t, 0} (\exp\{  \gamma_t s_t(Y) \}; x)}\pi_{1-t, 0}(x)\bigg\}\bigg)g_1(x)g_0(x)l(x)dF(x). 
\end{align*}
The derivative of $\psi_t(P^\ast_{\theta}; \gamma_t)$ with respect to $\alpha$ evaluated at $\theta=0$ is
\begin{align*}
    \frac{\partial \psi_t(P^\ast_{\theta}; \gamma_t)}{\partial\alpha}\bigg|_{\theta=0}=\int_x g_0(x)\bigg\{\mu_{t, 0}(Y; x)-\frac{\mu_{t, 0}(Y\exp\{  \gamma_t s_t(Y) \}; x)}{\mu_{t, 0} (\exp\{  \gamma_t s_t(Y) \}; x)}\bigg\}\times(-1)^{t+1}\pi_{1, 0}(x)\pi_{0, 0}(x)j(x)dF(x).
\end{align*}
The derivative of $\psi_t(P^\ast_{\theta}; \gamma_t)$ with respect to $\beta$ evaluated at $\theta=0$ is 0.

The derivative of $\psi_t(P^\ast_{\theta}; \gamma_t)$ with respect to $\eta_{t, 1}$ evaluated at $\theta=0$ is
\begin{align*}
    \frac{\partial \psi_t(P^\ast_{\theta}; \gamma_t)}{\partial\eta_{t, 1}}\bigg|_{\theta=0}=\int_x\bigg[\int_y yk_{t, 1}(y, x)dF_{t, 1}(y\mid x)\bigg]g_1(x)dF(x).
\end{align*}
The derivative of $\psi_t(P^\ast_{\theta}; \gamma_t)$ with respect to $\eta_{t, 0}$ evaluated at $\theta=0$ is
\begin{align*}
    \frac{\partial \psi_t(P^\ast_{\theta}; \gamma_t)}{\partial\eta_{t, 0}}\bigg|_{\theta=0}=\int_x g_0(x)\bigg\{&\bigg[\int_y yk_{t, 0}(y, x)dF_{t, 0}(y\mid x)\bigg]\pi_{t, 0}(x)\\
    +\bigg(&\frac{\int_y y\exp\{  \gamma_t s_t(y) \}k_{t, 0}(y, x)dF_{t, 0}(y\mid x)}{\mu_{t, 0} (\exp\{  \gamma_t s_t(Y) \}; x)}\\
    &-\frac{\mu_{t, 0} (Y\exp\{  \gamma_t s_t(Y) \}; x)}{\mu_{t, 0} (\exp\{  \gamma_t s_t(Y) \}; x)^2}\int_y \exp\{  \gamma_t s_t(y) \}k_{t, 0}(y, x)dF_{t, 0}(y\mid x)\bigg)\pi_{1-t, 0}(x)\bigg\}dF(x).
\end{align*}
The derivative of $\psi_t(P^\ast_{\theta}; \gamma_t)$ with respect to $\eta_{1-t, 1}$ evaluated at $\theta=0$ is 0. 

The derivative of $\psi_t(P^\ast_{\theta}; \gamma_t)$ with respect to $\eta_{1-t, 0}$ evaluated at $\theta=0$ is 0. 

Any mean zero observed data random variable can be expressed as
\begin{align*}
d(O)=&a(X)+\left(\mathbb{I}(R=1)-g_1(X)\right)b(X)+\mathbb{I}(R=0)\left(\mathbb{I}(T=1)-\pi_{1, 0}(X)\right)c(X)\\
&+\mathbb{I}(R=1)\left(\mathbb{I}(T=1)-\pi_{1, 1}(X)\right)w(X)
+\sum_{t, r\in\{0, 1\}}\mathbb{I}(T=t, R=r)d_{t, r}(Y, X)
\end{align*}
where $\mathbb{E}[a(X)]=0$, $\mathbb{E}[d_{t, r}(Y,X)|T=t, R=r, X]=0$, $b(X)$, $w(X)$ and $c(X)$ are unspecified functions of X. To find the non-parametric influence
function, we need to find $a(X)$, $b(X)$, $w(X)$, $c(X)$, and $d_{t, r}(Y, X)$ such that $\mathbb{E}[a(X)h(X)]=\frac{\partial \psi_t(P^\ast_{\theta}; \gamma_t)}{\partial\epsilon}\bigg|_{\theta=0}$, $\mathbb{E}[(R-g_1(X))^2b(X)l(X)]=\frac{\partial \psi_t(P^\ast_{\theta}; \gamma_t)}{\partial\delta}\bigg|_{\theta=0}$, $\mathbb{E}[(1-R)\{T-\pi_{1, 0}(X)\}^2c(X)j(X)]=\frac{\partial \psi_t(P^\ast_{\theta}; \gamma_t)}{\partial\alpha}\bigg|_{\theta=0}$, $\mathbb{E}[R\{T-\pi_{1, 1}(X)\}^2m(X)w(X)]=\frac{\partial \psi_t(P^\ast_{\theta}; \gamma_t)}{\partial\beta}\bigg|_{\theta=0}$, and $\mathbb{E}[\mathbb{I}(T=t, R=r)d_{t, r}(Y, X)k_{t, r}(Y, X)]=\frac{\partial \psi_t(P^\ast_{\theta}; \gamma_t)}{\partial\eta_{t, r}}\bigg|_{\theta=0}$. We can show that 
\begin{align*}
    a(X) &= \mu_{t, 1}(Y; X)g_1(x)+g_0(X)\bigg\{\mu_{t, 0}(Y; X)\pi_{t, 0}(X)+\frac{\mu_{t, 0}(Y\exp\{  \gamma_t s_t(Y) \}; X)}{\mu_{t, 0} (\exp\{  \gamma_t s_t(Y) \}; X)}\pi_{1-t, 0}(X)\bigg\}-\psi_t(P; \gamma_t)\\
    b(X) &= \mu_{t, 1}(Y; X)-\bigg\{\mu_{t, 0}(Y; X)\pi_{t, 0}(X)+\frac{\mu_{t, 0}(Y\exp\{  \gamma_t s_t(Y) \}; X)}{\mu_{t, 0} (\exp\{  \gamma_t s_t(Y) \}; X)}\pi_{1-t, 0}(X)\bigg\}\\
    c(X) &= (-1)^{t+1}\times\bigg\{\mu_{t, 0}(Y; X)-\frac{\mu_{t, 0}(Y\exp\{  \gamma_t s_t(Y) \}; X)}{\mu_{t, 0} (\exp\{  \gamma_t s_t(Y) \}; X)}\bigg\}\\
    w(X) &=0\\
    d_{t, 1}(Y, X) &= \frac{1}{\pi_{t, 1}(X)}(Y-\mu_{t, 1}(Y; X))\\
    d_{t, 0}(Y, X) &= Y-\mu_{t, 0}(Y; X)+\frac{\pi_{1-t, 0}(X)}{\pi_{t, 0}(X)}\frac{\exp\{  \gamma_t s_t(Y) \}}{\mu_{t, 0} (\exp\{  \gamma_t s_t(Y) \}; X)}\left\{Y-\frac{\mu_{t, 0} (Y\exp\{  \gamma_t s_t(Y) \}; X)}{\mu_{t, 0} (\exp\{  \gamma_t s_t(Y) \}; X)}\right\}\\
    d_{1-t, 1}(Y, X) &= 0\\
    d_{1-t, 0}(Y, X) &= 0.
\end{align*}
Thus, the non-parametric influence function for $\psi_t(P; \gamma_t)$ under Assumptions (A1-A2), (A3') and (A4)  no outcome missingness , denoted by $\phi^c_t(P; \gamma_t)(O)$, is:
\begin{equation}
\begin{aligned}\label{IF_proof}
    \phi^c_t(P; \gamma_t)(O) 
    &=  \mathbb{I}(R = 0, T = t) \times \Bigg\{ Y +  \frac{\pi_{1-t,0}(X)}{\pi_{t,0}(X)} \times \frac{\exp(\gamma_t s_t(Y))}{\mu_{t,0}(\exp(\gamma_t s_t(Y)); X)} \\
    &\hspace{6.25cm}  \times \left(Y  - \frac{\mu_{t,0}(Y\exp(\gamma_t s_t(Y)); X)}{\mu_{t,0}(\exp(\gamma_t s_t(Y)); X)} \right) \Bigg\} \\
    &\hspace{0.25cm} + \mathbb{I}(R=0, T=1-t) \times \frac{\mu_{t,0}(Y\exp(\gamma_t s_t(Y)); X)}{\mu_{t,0}(\exp(\gamma_t s_t(Y)); X)} \\
    &\hspace{0.25cm} + \mathbb{I}(R = 1) \times  \bigg\{  \frac{\mathbb{I}(T=t)}{\pi_{t,1}(X)}  \times \big( Y - \mu_{t,1}(Y;X) \big) + \mu_{t,1}(Y;X) \bigg\}  - \psi_t(P; \gamma_t) \ . 
\end{aligned}
\end{equation}

Now we will show that after replacing $\pi_{t,1}(X)$ with $\pi_{t,1}$, $\phi^c_t(P; \gamma_t)$ is the efficient influence function for $\psi_t(P; \gamma_t)$ under Assumptions (A1-A4) with no outcome missingness. Under Assumptions (A1-A2), (A3') and (A4), the full-data likelihood is
\begin{align*}
    p(Y, R, A, X)=&p(X; \zeta_X)p(R|X; \zeta_R)p(A|X, R=1; \zeta_{A, 1})^Rp(A|X, R=0; \zeta_{A, 0})^{1-R}\\
    &\times\prod_{a, r\in\{0, 1\}} p(Y|X, A=a, R=r; \zeta_{Y, a, r})^{\I(A=a, R=r)}. 
\end{align*}
And the tangent space is the direct sum of the tangent spaces spanned by the scores of $\zeta_X$, $\zeta_{A, 1}$, $\zeta_{A, 0}$, $\zeta_{Y, 1, 1}$, $\zeta_{Y, 0, 1}$, $\zeta_{Y, 1, 0}$ and $\zeta_{Y, 0, 0}$:
$$\mathcal{F}=\mathcal{F}_{\zeta_X}\oplus\mathcal{F}_{\zeta_{A, 1}}\oplus\mathcal{F}_{\zeta_{A, 0}}\oplus\mathcal{F}_{\zeta_{Y, 1, 1}}\oplus\mathcal{F}_{\zeta_{Y, 0, 1}}\oplus\mathcal{F}_{\zeta_{Y, 1, 0}}\oplus\mathcal{F}_{\zeta_{Y, 0, 0}}.$$
In comparison, under Assumptions (A1-A4), the full-data likelihood is
\begin{align*}
    p(Y, R, A, X)=&p(X; \zeta_X)p(R|X; \zeta_R)p(A|R=1)^Rp(A|X, R=0; \zeta_{A, 0})^{1-R}\\
    &\prod_{a, r\in\{0, 1\}} p(Y|X, A=a, R=r; \zeta_{Y, a, r})^{\I(A=a, R=r)}. 
\end{align*}
And the tangent space is the same as before minus $\mathcal{F}_{\zeta_{A, 1}}$ (since $p(A|R=1)$ is known, there is no nuisance tangent space): 
$$\mathcal{F}'=\mathcal{F}_{\zeta_X}\oplus\mathcal{F}_{\zeta_{A, 0}}\oplus\mathcal{F}_{\zeta_{Y, 1, 1}}\oplus\mathcal{F}_{\zeta_{Y, 0, 1}}\oplus\mathcal{F}_{\zeta_{Y, 1, 0}}\oplus\mathcal{F}_{\zeta_{Y, 0, 0}}.$$
Thus $\mathcal{F}=\mathcal{F}'\oplus\mathcal{F}_{\zeta_{A, 1}}$. Since Assumption (A3) implies Assumption (A3'), we can treat $\phi^c_t(P; \gamma_t)$ as a ``naive'' influence function under Assumptions (A1-A4) and no outcome missingness. Using Theorem 3.5 in \citet{tsiatis2006semiparametric}, the efficient influence function under Assumptions (A1-A4) with fully observed outcome is
\begin{align*}
    \phi^{c, \text{eff}}_t(P; \gamma_t)&=\Pi(\phi^c_t(P; \gamma_t)|\mathcal{F}')\\
    &=\Pi(\phi^c_t(P; \gamma_t)|\mathcal{F})-\Pi(\phi^c_t(P; \gamma_t)|\mathcal{F}_{\zeta_{A, 1}})\\
    &=\Pi(\phi^c_t(P; \gamma_t)|\mathcal{F})\\
    &=\phi^c_t(P; \gamma_t),
\end{align*}
as $\Pi(\phi^c_t(P; \gamma_t)|\mathcal{F}_{\zeta_{A, 1}})=0$ (since $\frac{\partial \psi_t(P^\ast_{\theta}; \gamma_t)}{\partial\beta}\bigg|_{\theta=0}=0$) and $\phi^c_t(P; \gamma_t)\in\mathcal{F}$. 

Note in the paper, we choose to adopt the exact form of influence function in (\ref{IF_proof}) with covariate adjusted $\pi_{t, 1}(X)$. This is to adjust for chance imbalance in baseline covariates in the  RCT. 
 
Now we will derive the observed data influence function whenoutcomes are  missing at random. We need to derive the space orthogonal to the observed nuisance tangent space, denoted as $\Lambda^{\widetilde{O}, \perp}$, since the observed data influence function belongs to $\Lambda^{\widetilde{O}, \perp}$. We have $\Lambda^{\widetilde{O}, \perp} = \Lambda^{\widetilde{O}, \perp}_{O} \cap \Lambda^{\widetilde{O}, \perp}_{M | O}$: elements of $\Lambda^{\widetilde{O}, \perp}$ live in $\Lambda^{\widetilde{O}, \perp}_{O} $ and are orthogonal to $\Lambda^{\widetilde{O}}_{M | O}$. To derive $\Lambda^{\widetilde{O}, \perp}$, we follow two steps: 

\textbf{STEP 1.} Derive $\Lambda^{\widetilde{O}, \perp}_{O}$. \\ \\
By Theorem 7.2 in \citet{tsiatis2006semiparametric}, 
\begin{align*}
    \Lambda^{\widetilde{O}, \perp}_{O} 
    &= \bigg\{ \frac{M}{\eta_{1}(X, R, T)}  \left\{\phi^c_t(\widetilde{P}; \gamma_t)(O)\oplus\mathcal{F}^{'\perp}\right\} + b(\widetilde{O}) \ : \ \E[b(\widetilde{O})\mid O]=0 \bigg\}. 
\end{align*}
Note that by Assumption (A5), $\phi^c_t(\widetilde{P}; \gamma_t)=\phi^c_t(P; \gamma_t)$. 

Any function of observed data $b(\widetilde{O})$ can be written as:
\begin{align*}
    b(\widetilde{O}) = M \ c(O) + (1-M) \ a_t(\widetilde{P})(X, R, T). 
\end{align*}

Applying the restriction on $b(\widetilde{O})$ to be mean zero given ${X, R, T, Y}$ implies: 
\begin{align*}
    \E[b(\widetilde{O}) \mid X, R, T, Y] 
    &= c(O) \eta_{1}(X, R, T) + a_t(\widetilde{P})(X, R, T)\eta_{0}(X, R, T) = 0. 
\end{align*}
Thus, it must be the case that $c(O) = - \frac{\eta_{0}(X, R, T)}{\eta_{1}(X, R, T)} a_t(\widetilde{P})(X, R, T)$. So, 
\begin{align*}
    b(\widetilde{O}) 
    &= a_t(\widetilde{P})(X, R, T) \Big\{ (1-M) - M \ \frac{\eta_{0}(X, R, T)}{\eta_{1}(X, R, T)}  \Big\} \\
    &= a_t(\widetilde{P})(X, R, T) \Big( 1 - \frac{M}{\eta_{1}(X, R, T)}  \Big).
\end{align*}
Thus, 
\begin{align*}
    \Lambda^{\widetilde{O}, \perp}_{O} 
    &= \bigg\{ \frac{M}{\eta_{1}(X, R, T)}  \left\{\phi^c_t(\widetilde{P}; \gamma_t)(O)\oplus\mathcal{F}^{'\perp}\right\} + a_t(\widetilde{P})(X, R, T) \Big( 1 - \frac{M}{\eta_{1}(X, R, T)}  \Big) \ : \ a_t(\widetilde{P})(X, R, T) \bigg\}\\
    &= \bigg\{ \frac{M}{\eta_{1}(X, R, T)}  \phi^c_t(\widetilde{P}; \gamma_t)(O) + a_t(\widetilde{P})(X, R, T) \Big( 1 - \frac{M}{\eta_{1}(X, R, T)}  \Big) \ : \ a_t(\widetilde{P})(X, R, T) \bigg\}\oplus\frac{M}{\eta_{1}(X, R, T)}\mathcal{F}^{'\perp}
\end{align*}
where $\phi^c_t(\widetilde{P}; \gamma_t)(O)$ is 
\begin{align*}
    \phi^c_t(\widetilde{P}; \gamma_t)(O) 
    &=  \mathbb{I}(R = 0, T = t) \times \Bigg\{ Y +  \frac{\pi_{1-t,0}(X)}{\pi_{t,0}(X)} \times \frac{\exp(\gamma_t s_t(Y))}{\widetilde{\mu}_{t,0}(\exp(\gamma_t s_t(Y)); X)} \\
    &\hspace{6.25cm}  \times \left(Y  - \frac{\widetilde{\mu}_{t,0}(Y\exp(\gamma_t s_t(Y)); X)}{\widetilde{\mu}_{t,0}(\exp(\gamma_t s_t(Y)); X)} \right) \Bigg\} \\
    &\hspace{0.25cm} + \mathbb{I}(R=0, T=1-t) \times \frac{\widetilde{\mu}_{t,0}(Y\exp(\gamma_t s_t(Y)); X)}{\widetilde{\mu}_{t,0}(\exp(\gamma_t s_t(Y)); X)} \\
    &\hspace{0.25cm} + \mathbb{I}(R = 1) \times  \bigg\{  \frac{\mathbb{I}(T=t)}{\pi_{t,1}(X)}  \times \big( Y - \widetilde{\mu}_{t,1}(Y;X) \big) + \widetilde{\mu}_{t,1}(Y;X) \bigg\}  - \psi_t(\widetilde{P}; \gamma_t) \ . 
\end{align*}

\textbf{STEP 2. Ensuring orthogonality of elements in $\Lambda^{\widetilde{O}, \perp}_{O} $ to $\Lambda^{\widetilde{O}}_{M|O}$. }\\ \\
By Assumption (A5), 
\begin{align*}
    \Lambda^{\widetilde{O}}_{M|O}=\Lambda_{M | O} = \Big\{ (M - \eta_{1}(X, R, T)) \  h(X, R, T) \ : \ h(X, R, T) \Big\}
\end{align*}
Now we find $c_t(X, R, T)$ that ensures orthogonality to all elements of $\Lambda_{M|O}$, i.e.,
\begin{align*}
    &\E\bigg[\frac{M}{\eta_{1}(X, R, T)}  \phi^c_t(\widetilde{P}; \gamma_t)(O) (M - \eta_{1}(X, R, T)) 
    \\
    &\hspace{0.5cm}+ a_t(\widetilde{P})(X, R, T) \Big( 1 - \frac{M}{\eta_{1}(X, R, T)}  \Big) (M - \eta_{1}(X, R, T)) \ \bigg| \ X, R, T \bigg] = 0 \ . 
\end{align*}
Therefore, we have: 
\begin{align*}
    &\E\bigg[\frac{M}{\eta_{1}(X, R, T)} \ (M - \eta_{1}(X, R, T)) \ \mathbb{I}(R = 0, T = t) \times Y
    \\
    &\hspace{0.5cm} + \frac{M}{\eta_{1}(X, R, T)} \ (M - \eta_{1}(X, R, T)) \ \mathbb{I}(R = 0, T = t) \times   \frac{\pi_{1-t,0}(X)}{\pi_{t,0}(X)} \times  \frac{\exp(\gamma_t s_t(Y))}{\widetilde{\mu}_{0,t}(\exp(\gamma_t s_t(Y)); X)} \\ 
    &\hspace{4cm} \times\Big( Y - \frac{ \widetilde{\mu}_{t,0}(Y\exp(\gamma_t s_t(Y)); X)}{\widetilde{\mu}_{t,0}(\exp(\gamma_t s_t(Y)); X)} \Big) 
    \\
    &\hspace{0.5cm} + \frac{M}{\eta_{1}(X, R, T)} \ (M - \eta_{1}(X, R, T)) \ \mathbb{I}(R=0, T=1-t) \times \frac{\widetilde{\mu}_{t,0}(Y\exp(\gamma_t s_t(Y)); X)}{\widetilde{\mu}_{t,0}(\exp(\gamma_t s_t(Y)); X)} 
    \\
    &\hspace{0.5cm} + \frac{M}{\eta_{1}(X, R, T)} \ (M - \eta_{1}(X, R, T)) \ \mathbb{I}(R = 1) \times \bigg\{ \frac{\mathbb{I}(T = t)}{\pi_{t, 1}(X)}  \times \big( Y - \widetilde{\mu}_{t,1}(Y;X) \big) + \widetilde{\mu}_{t,1}(Y;X) \bigg\} 
    \\
     &\hspace{0.5cm} - \frac{M}{\eta_{1}(X, R, T)} \ (M - \eta_{1}(X, R, T)) \ \psi_t(\widetilde{P}; \gamma_t) 
    \\
    &\hspace{0.5cm}+ a_t(\widetilde{P})(X, R, T) \Big( 1 - \frac{M}{\eta_{1}(X, R, T)}  \Big) (M - \eta_{1}(X, R, T)) \ \bigg| \ X, R, T \bigg] = 0 \ . 
\end{align*}

Since 
\begin{align*}
    &\E\Bigg[\frac{M}{\eta_{1}(X, R, T)} \ (M - \eta_{1}(X, R, T)) \ \mathbb{I}(R = 0, T = t) \times \frac{\pi_{1-t,0}(X)}{\pi_{t,0}(X)} \times  \frac{\exp(\gamma_t s_t(Y))}{\widetilde{\mu}_{0,t}(\exp(\gamma_t s_t(Y)); X)} \\ 
    &\hspace{4cm} \times\Big( Y - \frac{ \widetilde{\mu}_{t,0}(Y\exp(\gamma_t s_t(Y)); X)}{\widetilde{\mu}_{t,0}(\exp(\gamma_t s_t(Y)); X)} \Big) \bigg|X, R, T\Bigg]\\
    =&\E\Bigg[\eta_{0}(X, R, T) \mathbb{I}(R = 0, T = t) \times \frac{\pi_{1-t,0}(X)}{\pi_{t,0}(X)} \times  \frac{\exp(\gamma_t s_t(Y))}{\widetilde{\mu}_{0,t}(\exp(\gamma_t s_t(Y)); X)} \\ 
    &\hspace{4cm} \times\Big( Y - \frac{ \widetilde{\mu}_{t,0}(Y\exp(\gamma_t s_t(Y)); X)}{\widetilde{\mu}_{t,0}(\exp(\gamma_t s_t(Y)); X)} \Big) \bigg|X, R, T, M=1\Bigg]\\
    =& \eta_{0}(X, R, T) \mathbb{I}(R = 0, T = t) \times \frac{\pi_{1-t,0}(X)}{\pi_{t,0}(X)}\\
    &\hspace{1cm} \times\E\Bigg[\frac{\exp(\gamma_t s_t(Y))}{\widetilde{\mu}_{0,t}(\exp(\gamma_t s_t(Y)); X)}\Big( Y - \frac{ \widetilde{\mu}_{t,0}(Y\exp(\gamma_t s_t(Y)); X)}{\widetilde{\mu}_{t,0}(\exp(\gamma_t s_t(Y)); X)} \Big) \bigg|X, R, T, M=1\Bigg]\\
    =& \eta_{0}(X, R, T) \mathbb{I}(R = 0, T = t) \times \frac{\pi_{1-t,0}(X)}{\pi_{t,0}(X)}\times\Big\{\frac{\widetilde{\mu}_{0,t}(Y\exp(\gamma_t s_t(Y)); X)}{\widetilde{\mu}_{0,t}(\exp(\gamma_t s_t(Y)); X)}-\frac{\widetilde{\mu}_{0,t}(Y\exp(\gamma_t s_t(Y)); X)}{\widetilde{\mu}_{0,t}(\exp(\gamma_t s_t(Y)); X)}\Big\}\\
    =&0,
\end{align*}
we have
\begin{align*}
    0 &=\I(R=0, T=t) \ \eta_{0}(X, R, T) \ \widetilde{\mu}_{t,0}(Y; X) \\
    &\hspace{0.25cm}+ \I(R=0, T=1-t) \ \eta_{0}(X, R, T) \  \frac{\widetilde{\mu}_{t,0}(Y\exp(\gamma_t s_t(Y)); X)}{\widetilde{\mu}_{t,0}(\exp(\gamma_t s_t(Y)); X)} \\
    &\hspace{0.25cm}+ \I(R=1) \ \eta_{0}(X, R, T) \ \widetilde{\mu}_{t,1}(Y; X) \\
    &\hspace{0.25cm} - \eta_{0}(X, R, T) \ \psi_t(\widetilde{P}; \gamma_t) \\
    &\hspace{0.25cm} - a_t(\widetilde{P})(X, R, T) \eta_{0}(X, R, T).
\end{align*}
So, 
\begin{equation*} 
\begin{aligned}
   a_t(\widetilde{P})(X, R, T) &= \frac{1}{\eta_{0}(X, R, T)} \bigg\{  \I(R=0, T=t) \ \eta_{0}(X, R, T) \ \widetilde{\mu}_{t,0}(Y; X) \\
    &\hspace{1cm}+ \I(R=0, T=1-t) \ \eta_{0}(X, R, T) \  \frac{\widetilde{\mu}_{t,0}(Y\exp(\gamma_t s_t(Y)); X)}{\widetilde{\mu}_{t,0}(\exp(\gamma_t s_t(Y)); X)}  \\
    &\hspace{1cm}+ \I(R=1) \ \eta_{0}(X, R, T) \ \widetilde{\mu}_{t,1}(Y; X)  \bigg\} - \psi_t(\widetilde{P}; \gamma_t)\\
    &= \I(R=0, T=t) \ \widetilde{\mu}_{t,0}(Y; X) + \I(R=0, T=1-t) \  \frac{\widetilde{\mu}_{t,0}(Y\exp(\gamma_t s_t(Y)); X)}{\widetilde{\mu}_{t,0}(\exp(\gamma_t s_t(Y)); X)}  \\
    &\hspace{1cm}+ \I(R=1) \ \widetilde{\mu}_{t,1}(Y; X) - \psi_t(\widetilde{P}; \gamma_t)
\end{aligned}
\end{equation*}
In addition, 
$$\Pi\left(\frac{M}{\eta_{1}(X, R, T)}\mathcal{F}^{'\perp}\bigg|\Lambda^{\widetilde{O}}_{M|O}\right)=\mathcal{F}^{'\perp}.\,$$
Putting all these together, we get the class of influence functions for $\psi_t(\widetilde{P}; \gamma_t)$ as \\
$\left(\frac{M}{\eta_{1}(X, R, T)}  \phi^c_t(\widetilde{P}; \gamma_t)(O) + a_t(\widetilde{P})(X, R, T) \Big( 1 - \frac{M}{\eta_{1}(X, R, T)}  \Big)\right)\oplus\mathcal{F}^{'\perp}$. 

Since the observed tangent space $\mathcal{F}^{'\widetilde{O}}=\E[\mathcal{F}^{'}|\widetilde{O}]=\mathcal{F}'$, the observed influence function for $\psi_t(\widetilde{P}; \gamma_t)$, denoted as $\phi_t(\widetilde{P}; \gamma_t)(\widetilde{O})$, is
\begin{align*}
\phi_t(\widetilde{P}; \gamma_t)(\widetilde{O})=&\Pi\left(\left(\frac{M}{\eta_{1}(X, R, T)}  \phi^c_t(\widetilde{P}; \gamma_t)(O) + a_t(\widetilde{P})(X, R, T) \Big( 1 - \frac{M}{\eta_{1}(X, R, T)}  \Big)\right)\oplus\mathcal{F}^{'\perp}\Bigg|\mathcal{F}^{'\widetilde{O}}\right)\\
=&\Pi\left(\left(\frac{M}{\eta_{1}(X, R, T)}  \phi^c_t(\widetilde{P}; \gamma_t)(O) + a_t(\widetilde{P})(X, R, T) \Big( 1 - \frac{M}{\eta_{1}(X, R, T)}  \Big)\right)\oplus\mathcal{F}^{'\perp}\Bigg|\mathcal{F}'\right)\\
=&\frac{M}{\eta_{1}(X, R, T)}  \phi^c_t(\widetilde{P}; \gamma_t)(O) + a_t(\widetilde{P})(X, R, T) \Big( 1 - \frac{M}{\eta_{1}(X, R, T)}  \Big)
\end{align*}

We can re-write $\phi_t(\widetilde{P}; \gamma_t)(\widetilde{O})$ as
\begin{align*}
    \phi_t(\widetilde{P}; \gamma_t)(\widetilde{O})
    &=  \underbrace{\upsilon_{t, 1}(\widetilde{P})(\widetilde{O})\times P(R=1)+\upsilon_{t, 0}(\widetilde{P}; \gamma_t)(\widetilde{O})\times P(R=0)}_{\upsilon_{t}(\widetilde{P}; \gamma_t)(\widetilde{O})} - \psi_t(\widetilde{P}; \gamma_t) \ , 
\end{align*}
where $\upsilon_{t, 1}(\widetilde{P})(\widetilde{O})$ is given by
\begin{equation*}
    \begin{aligned}
    \upsilon_{t, 1}(\widetilde{P})(\widetilde{O})
    =& \frac{M}{\eta_1(X, R, T)}\times\frac{\I(R=1)}{P(R=1)}\left\{\frac{\I(T=t)}{\pi_{t, 1}(X)}\left(Y-\widetilde{\mu}_{t, 1}(Y; X)\right)+\widetilde{\mu}_{t, 1}(Y; X))\right\} \\
    &\hspace{0.25cm}+\Big\{ 1 - \frac{M}{\eta_{1}(X, R, T)}  \Big\}\times \frac{\I(R=1)}{P(R=1)}\widetilde{\mu}_{t, 1}(Y; X )\ ,
\end{aligned}
\end{equation*}
and $\upsilon_{t, 0}(\widetilde{P}; \gamma_t)(\widetilde{O})$ is given by
\begin{equation*}
    \begin{aligned}
   \upsilon_{t, 0}(\widetilde{P}; \gamma_t)(\widetilde{O})
    =& \frac{M}{\eta_1(X, R, T)}\Bigg[\frac{\mathbb{I}(R = 0, T = t)}{P(R=0)}\Bigg\{ Y +  \frac{\pi_{1-t,0}(X)}{\pi_{t,0}(X)} \times \frac{\exp(\gamma_t s_t(Y))}{\widetilde{\mu}_{t,0}(\exp(\gamma_t s_t(Y)); X)}  \\
    &\hspace{6cm}\times\left\{Y  - \frac{\widetilde{\mu}_{t,0}(Y\exp(\gamma_t s_t(Y)); X)}{\widetilde{\mu}_{t,0}(\exp(\gamma_t s_t(Y)); X)} \right\} \Bigg\}\\
    &\hspace{2.5cm}+\frac{\mathbb{I}(R=0, T=1-t)}{P(R=0)} \times \frac{\widetilde{\mu}_{t,0}(Y\exp(\gamma_t s_t(Y)); X)}{\widetilde{\mu}_{t,0}(\exp(\gamma_t s_t(Y)); X)}\Bigg]\\
    &\hspace{0.25cm}+\Big\{ 1 - \frac{M}{\eta_{1}(X, R, T)}  \Big\}\Bigg[ \frac{\I(R=0, T=t)}{P(R=0)} \ \widetilde{\mu}_{t,0}(Y; X)\\
    &\hspace{2.5cm}+\frac{\I(R=0, T=1-t)}{P(R=0)}\ \frac{\widetilde{\mu}_{t,0}(Y\exp(\gamma_t s_t(Y)); X)}{\widetilde{\mu}_{t,0}(\exp(\gamma_t s_t(Y)); X)}\Bigg]\ .
\end{aligned}
\end{equation*}

\subsection{Second-order remainder term}\label{sec:remainder} 

\begin{lemma}\label{remainder}
The remainder term, $Rem_t(\widetilde{P}^\ast, \widetilde{P}) \equiv \psi_t(\widetilde{P}^\ast; \gamma_t) - \psi_t(\widetilde{P}; \gamma_t) + \E[\phi_t(\widetilde{P}^\ast; \gamma_t)(\widetilde{O})]$, where $\widetilde{P}^\ast$ is any distribution of $\widetilde{O}$, takes the form:
{\small
\begin{align*}
&\E\Bigg[\frac{1}{\eta^\ast_{1}(X, R=0, T=t)}g_0(X)\pi_{t, 0}(X)\bigg(\widetilde{\mu}_{t, 0}(Y; X)-\widetilde{\mu}^\ast_{t,0}(Y; X)\bigg)\bigg(\eta_{1}(X, R=0, T=t)-\eta^\ast_{1}(X, R=0, T=t)\bigg)\\
    &\hspace{1cm} + \frac{g_0(X)\bigg(\widetilde{\mu}_{t,0}(Y\exp(\gamma_t s_t(Y)); X)-\widetilde{\mu}^\ast_{t,0}(Y\exp(\gamma_t s_t(Y)); X)\bigg)}{\eta^\ast_1(X, R=0, T=t)\pi^\ast_{t, 0}(X)\widetilde{\mu}_{t, 0}(\exp(\gamma_t s_t(Y)); X)\widetilde{\mu}^\ast_{t, 0}(\exp(\gamma_t s_t(Y)); X)}\\
    &\hspace{1.5cm}\times\bigg\{\bigg(\eta_1(X, R=0, T=t)-\eta^\ast_1(X, R=0, T=t)\bigg)\pi_{t, 0}(X)\widetilde{\mu}_{t, 0}(\exp(\gamma_t s_t(Y)); X)\big(1-\pi^\ast_{t, 0}(X)\big)\\
    &\hspace{2cm}+\bigg(\widetilde{\mu}_{t, 0}(\exp(\gamma_t s_t(Y)); X)-\widetilde{\mu}^\ast_{t, 0}(\exp(\gamma_t s_t(Y)); X)\bigg)\pi^\ast_{t, 0}(X)\eta^\ast_{1}(X, R=0, T=t)\big(1-\pi_{t, 0}(X)\big)\\
    &\hspace{2cm}+\bigg(\pi_{t, 0}(X)-\pi^\ast_{t, 0}(X)\bigg)\eta^\ast_{1}(X, R=0, T=t)\widetilde{\mu}_{t, 0}(\exp(\gamma_t s_t(Y)); X)\bigg\}\\
    &\hspace{1cm}-\frac{\widetilde{\mu}^\ast_{t, 0}(Y\exp(\gamma_t s_t(Y)); X)}{\widetilde{\mu}^\ast_{t, 0}(\exp(\gamma_t s_t(Y)); X)}\frac{g_0(X)\bigg(\widetilde{\mu}_{t, 0}(\exp(\gamma_t s_t(Y)); X)-\widetilde{\mu}^\ast_{t, 0}(\exp(\gamma_t s_t(Y)); X)\bigg)}{\eta^\ast_1(X, R=0, T=t)\pi^\ast_{t, 0}(X)\widetilde{\mu}_{t, 0}(\exp(\gamma_t s_t(Y)); X)\widetilde{\mu}^\ast_{t, 0}(\exp(\gamma_t s_t(Y)); X)}\\
    &\hspace{1.5cm}\times\bigg\{\bigg(\eta_1(X, R=0, T=t)-\eta^\ast_1(X, R=0, T=t)\bigg)\pi_{t, 0}(X)\widetilde{\mu}_{t, 0}(\exp(\gamma_t s_t(Y)); X)\big(1-\pi^\ast_{t, 0}(X)\big)\\
    &\hspace{2cm}+\bigg(\mu_{t, 0}(\exp(\gamma_t s_t(Y)); X)-\widetilde{\mu}^\ast_{t, 0}(\exp(\gamma_t s_t(Y)); X)\bigg)\pi^\ast_{t, 0}(X)\eta^\ast_{1}(X, R=0, T=t)\big(1-\pi_{t, 0}(X)\big)\\
    &\hspace{2cm}+\bigg(\pi_{t, 0}(X)-\pi^\ast_{t, 0}(X)\bigg)\eta^\ast_{1}(X, R=0, T=t)\widetilde{\mu}_{t, 0}(\exp(\gamma_t s_t(Y)); X)\bigg\}\\
    &\hspace{1cm}+\frac{g_1(X)}{\eta^\ast_1(X, R=1, T=t)\pi^\ast_{t, 1}(X)}\bigg( \widetilde{\mu}_{t,1}(Y;X) - \widetilde{\mu}^\ast_{t,1}(Y;X) \bigg)\\
    &\hspace{1.5cm}\times\bigg\{\eta_1(X, R=1, T=t)\bigg(\pi_{t, 1}(X)-\pi^\ast_{t, 1}(X)\bigg)+\pi^\ast_{t, 1}(X)\bigg(\eta_1(X, R=1, T=t)-\eta^\ast_1(X, R=1, T=t)\bigg)\bigg\}\Bigg]
\end{align*}
}
\end{lemma}

\textbf{Proof of Lemma \ref{remainder}}
{\small
\begin{align*}
    &Rem_t(\widetilde{P}^\ast, \widetilde{P}) \equiv \psi_t(\widetilde{P}^\ast; \gamma_t) - \psi_t(\widetilde{P}; \gamma_t) + \E[\phi_t(\widetilde{P}^\ast; \gamma_t)(\widetilde{O})]
    \\ 
    &= \psi_t(\widetilde{P}^\ast; \gamma_t) - \psi_t(\widetilde{P}; \gamma_t) +\E\left[\frac{M}{\eta^\ast_{1}(X, R, T)}  \phi^c_t(\widetilde{P}^\ast)(O) +  \Big\{ 1 - \frac{M}{\eta^\ast_{1}(X, R, T)}  \Big\} c(X, R, T; \widetilde{P}^\ast)\right]\\
    &= \E\Bigg[\frac{\mathbb{I}(M=1, R = 0, T = t)}{\eta^\ast_{1}(X, R, T)} \times \Bigg\{Y + \frac{\pi^\ast_{1-t, 0}(X)}{\pi^\ast_{t, 0}(X)} \times \left( \frac{Y\exp(\gamma_t s_t(Y))}{\widetilde{\mu}^\ast_{t, 0}(\exp(\gamma_t s_t(Y)); X)} - \frac{\exp(\gamma_t s_t(Y))\widetilde{\mu}^\ast_{t, 0}(Y\exp(\gamma_t s_t(Y)); X)}{\widetilde{\mu}^\ast_{t, 0}(\exp(\gamma_t s_t(Y)); X)^2} \right) \Bigg\} \\
    &\hspace{1cm} + \frac{\mathbb{I}(M=1, R=0, T=1-t)}{\eta^\ast_{1}(X, R, T)} \times \frac{\widetilde{\mu}^\ast_{t, 0}(Y\exp(\gamma_t s_t(Y)); X)}{\widetilde{\mu}^\ast_{t, 0}(\exp(\gamma_t s_t(Y)); X)} \\
    &\hspace{1cm} + \frac{\mathbb{I}(M=1, R=1)}{\eta^\ast_{1}(X, R, T)} \times \bigg\{ \frac{\mathbb{I}(T = t)}{\pi^\ast_{t, 1}(X)}  \times \big( Y - \widetilde{\mu}^\ast_{t,1}(Y;X) \big) + \widetilde{\mu}^\ast_{t,1}(Y;X) \bigg\} \\
    &\hspace{1cm} + \left(1-\frac{M}{\eta^\ast_{1}(X, R, T)}\right)\I(R=0, T=t) \ \widetilde{\mu}^\ast_{t,0}(Y; X)+ \left(1-\frac{M}{\eta^\ast_{1}(X, R, T)}\right)\I(R=0, T=1-t) \  \frac{\widetilde{\mu}^\ast_{t,0}(Y\exp(\gamma_t s_t(Y)); X)}{\widetilde{\mu}^\ast_{t,0}(\exp(\gamma_t s_t(Y)); X)}\\
    &\hspace{1cm} + \left(1-\frac{M}{\eta^\ast_{1}(X, R, T)}\right)\I(R=1) \ \widetilde{\mu}^\ast_{t,1}(Y; X)\Bigg]- \psi_t(\widetilde{P}; \gamma_t)\\
    &= \E\Bigg[\frac{\eta_{1}(X, R=0, T=t)}{\eta^\ast_{1}(X, R=0, T=t)}g_0(X)\pi_{t, 0}(X) \Bigg\{\widetilde{\mu}_{t, 0}(Y; X) + \frac{\pi^\ast_{1-t, 0}(X)}{\pi^\ast_{t, 0}(X)} \\
    &\hspace{6cm}\times \left( \frac{\widetilde{\mu}_{t, 0}(Y\exp(\gamma_t s_t(Y)); X)}{\widetilde{\mu}^\ast_{t, 0}(\exp(\gamma_t s_t(Y)); X)} - \frac{\widetilde{\mu}_{t, 0}(\exp(\gamma_t s_t(Y)); X)\widetilde{\mu}^\ast_{t, 0}(Y\exp(\gamma_t s_t(Y)); X)}{\widetilde{\mu}^\ast_{t, 0}(\exp(\gamma_t s_t(Y)); X)^2} \right) \Bigg\} \\
    &\hspace{1cm}+\frac{\eta_{1}(X, R=0, T=1-t)}{\eta^\ast_{1}(X, R=0, T=1-t)}g_0(X)\pi_{1-t, 0}(X) \times \frac{\widetilde{\mu}^\ast_{t, 0}(Y\exp(\gamma_t s_t(Y)); X)}{\widetilde{\mu}^\ast_{t, 0}(\exp(\gamma_t s_t(Y)); X)}\\
    &\hspace{1cm} + \frac{\eta_{1}(X, R=1, T=t)}{\eta^\ast_{1}(X, R=1, T=t)}\times\frac{\pi_{t, 1}(X)}{\pi^\ast_{t, 1}(X)} g_1(X)\big( \widetilde{\mu}_{t,1}(Y;X) - \widetilde{\mu}^\ast_{t,1}(Y;X) \big) +\frac{\eta_{1}(X, R=1, T)}{\eta^\ast_{1}(X, R=1, T)}g_1(X)\widetilde{\mu}^\ast_{t,1}(Y;X) 
    \\
    &\hspace{1cm} + \left(1-\frac{\eta_{1}(X, R=0, T=t)}{\eta^\ast_{1}(X, R=0, T=t)}\right)g_0(X)\pi_{t, 0}(X) \ \widetilde{\mu}^\ast_{t,0}(Y; X)\\
    &\hspace{1cm} + \left(1-\frac{\eta_{1}(X, R=0, T=1-t)}{\eta^\ast_{1}(X, R=0, T=1-t)}\right)g_0(X)\pi_{1-t, 0}(X) \  \frac{\widetilde{\mu}^\ast_{t,0}(Y\exp(\gamma_t s_t(Y)); X)}{\widetilde{\mu}^\ast_{t,0}(\exp(\gamma_t s_t(Y)); X)}\\
    &\hspace{1cm} + \left(1-\frac{\eta_{1}(X, R=1, T)}{\eta^\ast_{1}(X, R=1, T)}\right)g_1(X)  \ \widetilde{\mu}^\ast_{t,1}(Y; X)\Bigg]- \psi_t(\widetilde{P}; \gamma_t)\\
    &=\E\Bigg[\frac{\eta_{1}(X, R=0, T=t)}{\eta^\ast_{1}(X, R=0, T=t)}g_0(X)\pi_{t, 0}(X)\big(\widetilde{\mu}_{t, 0}(Y; X)-\widetilde{\mu}^\ast_{t,0}(Y; X)\big)\\
    &\hspace{1cm}+\frac{\eta_{1}(X, R=0, T=t)}{\eta^\ast_{1}(X, R=0, T=t)}g_0(X)\pi_{t, 0}(X)\times\frac{\pi^\ast_{1-t, 0}(X)}{\pi^\ast_{t, 0}(X)}\\
    &\hspace{3cm}\times\frac{\big(\widetilde{\mu}_{t, 0}(Y\exp(\gamma_t s_t(Y)); X)-\widetilde{\mu}^\ast_{t, 0}(Y\exp(\gamma_t s_t(Y)); X)\big)}{\widetilde{\mu}^\ast_{t, 0}(\exp(\gamma_t s_t(Y)); X)}\\
    &\hspace{1cm}-\frac{\eta_{1}(X, R=0, T=t)}{\eta^\ast_{1}(X, R=0, T=t)}g_0(X)\pi_{t, 0}(X)\times\frac{\pi^\ast_{1-t, 0}(X)}{\pi^\ast_{t, 0}(X)}\\
    &\hspace{3cm}\times\frac{\widetilde{\mu}^\ast_{t, 0}(Y\exp(\gamma_t s_t(Y)); X)\big(\widetilde{\mu}_{t, 0}(\exp(\gamma_t s_t(Y)); X)-\widetilde{\mu}^\ast_{t, 0}(\exp(\gamma_t s_t(Y)); X)\big)}{\widetilde{\mu}^\ast_{t, 0}(\exp(\gamma_t s_t(Y)); X)^2}\\
    &\hspace{1cm} + \frac{\eta_{1}(X, R=1, T=t)}{\eta^\ast_{1}(X, R=1, T=t)}\times\frac{\pi_{t, 1}(X)}{\pi^\ast_{t, 1}(X)} g_1(X)\big( \widetilde{\mu}_{t,1}(Y;X) - \widetilde{\mu}^\ast_{t,1}(Y;X) \big)\\
    &\hspace{1cm} -g_0(X)\pi_{t, 0}(X) \big(\widetilde{\mu}_{t,0}(Y; X)-\widetilde{\mu}^\ast_{t,0}(Y; X)\big)\\
    &\hspace{1cm} +g_0(X)\pi_{1-t, 0}(X) \  \frac{\widetilde{\mu}^\ast_{t,0}(Y\exp(\gamma_t s_t(Y)); X)\big(\widetilde{\mu}_{t,0}(\exp(\gamma_t s_t(Y)); X)-\widetilde{\mu}^\ast_{t,0}(\exp(\gamma_t s_t(Y)); X)\big)}{\widetilde{\mu}_{t,0}(\exp(\gamma_t s_t(Y)); X)\widetilde{\mu}^\ast_{t,0}(\exp(\gamma_t s_t(Y)); X)}\\
    &\hspace{1cm} -g_0(X)\pi_{1-t, 0}(X) \  \frac{\big(\widetilde{\mu}_{t,0}(Y\exp(\gamma_t s_t(Y)); X)-\widetilde{\mu}^\ast_{t,0}(Y\exp(\gamma_t s_t(Y)); X)\big)}{\widetilde{\mu}_{t,0}(\exp(\gamma_t s_t(Y)); X)}\\
    &\hspace{1cm} -g_1(X)\bigg(\widetilde{\mu}_{t,1}(Y; X)- \widetilde{\mu}^\ast_{t,1}(Y; X)\bigg)\bigg]\\
    &=\E\Bigg[\frac{1}{\eta^\ast_{1}(X, R=0, T=t)}g_0(X)\pi_{t, 0}(X)\bigg(\widetilde{\mu}_{t, 0}(Y; X)-\widetilde{\mu}^\ast_{t,0}(Y; X)\bigg)\bigg(\eta_{1}(X, R=0, T=t)-\eta^\ast_{1}(X, R=0, T=t)\bigg)\\
    &\hspace{1cm} + \frac{g_0(X)\bigg(\widetilde{\mu}_{t,0}(Y\exp(\gamma_t s_t(Y)); X)-\widetilde{\mu}^\ast_{t,0}(Y\exp(\gamma_t s_t(Y)); X)\bigg)}{\eta^\ast_1(X, R=0, T=t)\pi^\ast_{t, 0}(X)\widetilde{\mu}_{t, 0}(\exp(\gamma_t s_t(Y)); X)\widetilde{\mu}^\ast_{t, 0}(\exp(\gamma_t s_t(Y)); X)}\\
    &\hspace{1.5cm}\times\bigg\{\bigg(\eta_1(X, R=0, T=t)-\eta^\ast_1(X, R=0, T=t)\bigg)\pi_{t, 0}(X)\widetilde{\mu}_{t, 0}(\exp(\gamma_t s_t(Y)); X)\big(1-\pi^\ast_{t, 0}(X)\big)\\
    &\hspace{2cm}+\bigg(\widetilde{\mu}_{t, 0}(\exp(\gamma_t s_t(Y)); X)-\widetilde{\mu}^\ast_{t, 0}(\exp(\gamma_t s_t(Y)); X)\bigg)\pi^\ast_{t, 0}(X)\eta^\ast_{1}(X, R=0, T=t)\big(1-\pi_{t, 0}(X)\big)\\
    &\hspace{2cm}+\bigg(\pi_{t, 0}(X)-\pi^\ast_{t, 0}(X)\bigg)\eta^\ast_{1}(X, R=0, T=t)\widetilde{\mu}_{t, 0}(\exp(\gamma_t s_t(Y)); X)\bigg\}\\
    &\hspace{1cm}-\frac{\widetilde{\mu}^\ast_{t, 0}(Y\exp(\gamma_t s_t(Y)); X)}{\widetilde{\mu}^\ast_{t, 0}(\exp(\gamma_t s_t(Y)); X)}\frac{g_0(X)\bigg(\widetilde{\mu}_{t, 0}(\exp(\gamma_t s_t(Y)); X)-\widetilde{\mu}^\ast_{t, 0}(\exp(\gamma_t s_t(Y)); X)\bigg)}{\eta^\ast_1(X, R=0, T=t)\pi^\ast_{t, 0}(X)\widetilde{\mu}_{t, 0}(\exp(\gamma_t s_t(Y)); X)\widetilde{\mu}^\ast_{t, 0}(\exp(\gamma_t s_t(Y)); X)}\\
    &\hspace{1.5cm}\times\bigg\{\bigg(\eta_1(X, R=0, T=t)-\eta^\ast_1(X, R=0, T=t)\bigg)\pi_{t, 0}(X)\widetilde{\mu}_{t, 0}(\exp(\gamma_t s_t(Y)); X)\big(1-\pi^\ast_{t, 0}(X)\big)\\
    &\hspace{2cm}+\bigg(\widetilde{\mu}_{t, 0}(\exp(\gamma_t s_t(Y)); X)-\widetilde{\mu}^\ast_{t, 0}(\exp(\gamma_t s_t(Y)); X)\bigg)\pi^\ast_{t, 0}(X)\eta^\ast_{1}(X, R=0, T=t)\big(1-\pi_{t, 0}(X)\big)\\
    &\hspace{2cm}+\bigg(\pi_{t, 0}(X)-\pi^\ast_{t, 0}(X)\bigg)\eta^\ast_{1}(X, R=0, T=t)\widetilde{\mu}_{t, 0}(\exp(\gamma_t s_t(Y)); X)\bigg\}\\
    &\hspace{1cm}+\frac{g_1(X)}{\eta^\ast_1(X, R=1, T=t)\pi^\ast_{t, 1}(X)}\bigg( \widetilde{\mu}_{t,1}(Y;X) - \widetilde{\mu}^\ast_{t,1}(Y;X) \bigg)\\
    &\hspace{1.5cm}\times\bigg\{\eta_1(X, R=1, T=t)\bigg(\pi_{t, 1}(X)-\pi^\ast_{t, 1}(X)\bigg)+\pi^\ast_{t, 1}(X)\bigg(\eta_1(X, R=1, T=t)-\eta^\ast_1(X, R=1, T=t)\bigg)\bigg\}\Bigg]
\end{align*} 
}

\subsection{Proof of robustness property}\label{robustness}

\begin{theorem}\label{theorem_robustness}
The one-step, split sample truncated estimator $\widehat{\psi}_t^\dagger(\gamma_t)$ is a consistent estimator of $\psi_t(\widetilde{P}; \gamma_t)$ if $\widetilde{F}_{t, r}(y|X), r=0, 1$ are correctly specified. 
\end{theorem}
\textbf{Proof of Theorem \ref{theorem_robustness}}
Let $\widehat{\pi}_{t, r}(X)\xrightarrow{P}\pi_{t, r}^\ast(X)$ and $\widehat{\eta}_1(X, r, t)\xrightarrow{P}\eta^\ast_1(X, r, t)$, $r=0, 1$, for any $\pi_{1, r}^\ast(X)$ and $\eta^\ast_1(X, r, t)$. Let $\widehat{\widetilde{F}}_{t, r}(y|X)\xrightarrow{P}\widetilde{F}_{t, r}(y|X), r=0, 1$ where $\widetilde{F}_{t, r}(y|X)$ is the true cumulative conditional distribution of outcome $Y$ given $T=t$, $R=r$ and $X$. Since $\widehat{\widetilde{\mu}}_{t, 0}\big(Y\exp \{ \gamma_t s_t( Y)\} ; X\big) = \int y\exp\{\gamma_t s_t(y)\} d \widehat{\widetilde{F}}_{t, 0}(y|X)$, $\int \exp\{\gamma_t s_t(y)\} d \widehat{\widetilde{F}}_{t, 0}(y|X)= \int \exp\{\gamma_t s_t(y)\} d \widehat{\widetilde{F}}_{t, 0}(y|X)$ and $\widehat{\widetilde{\mu}}_{t, r}\big(Y; X\big) = \int y d \widehat{\widetilde{F}}_{t, r}(y|X)$, we also have $\widehat{\widetilde{\mu}}_{t, 0}\big(Y\exp \{ \gamma_t s_t( Y)\} ; X\big)\xrightarrow{P}\widetilde{\mu}_{t, 0}\big(Y\exp \{ \gamma_t s_t( Y)\} ; X\big)$, $\widehat{\widetilde{\mu}}_{t, 0}\big(\exp \{ \gamma_t s_t( Y)\} ; X\big)\xrightarrow{P}\widetilde{\mu}_{t, 0}\big(\exp \{ \gamma_t s_t( Y)\} ; X\big)$ and $\widehat{\widetilde{\mu}}_{t, r}\big(Y; X\big)\xrightarrow{P}\widetilde{\mu}_{t, r}\big(Y; X\big)$. Let $\widetilde{P}^\star=\widetilde{P}\backslash\{\pi_{t, r}(X), \eta_1(X, r, t): t, r\in\{0, 1\}\}\cup\{\pi_{t, r}^\ast(X), \eta^\ast_1(X, r, t): t, r\in\{0, 1\}\}$. 

\citet{nabi2024semiparametric} has proven consistency of the truncation procedure: $\widehat{\psi}_t^\dagger(\gamma_t)\xrightarrow{P}\widehat{\psi}_t(\gamma_t)$. By the weak law of large numbers and the continuous mapping theorem, 
$$\widehat{\psi}_t(\gamma_t)\xrightarrow{P}\E\left[\upsilon_{t}(\widetilde{P}^\star; \gamma_t)(\widetilde{O})\right]$$
By simply plugging in $\widetilde{P}^\ast=\widetilde{P}^\star$ to the remainder term in Lemma~\ref{remainder} we have
\[
    \E\left[\upsilon_{t}(\widetilde{P}^\star; \gamma_t)(\widetilde{O})\right]-\psi_t(\widetilde{P}; \gamma_t)=Rem_t(\widetilde{P}^\star, \widetilde{P})=0
\]
Thus, 
$$\widehat{\psi}_t(\gamma_t)\xrightarrow{P}\psi_t(\widetilde{P}; \gamma_t)$$
Now we have proven $\widehat{\psi}_t^\dagger(\gamma_t)\xrightarrow{P}\psi_t(\widetilde{P}; \gamma_t)$ if $\widetilde{F}_{t, 0}(y|X)$ and $\widetilde{F}_{t, 1}(y|X)$ are correctly specified. 

\subsection{Modeling and estimation of $\widetilde{P}$}\label{nuisance}

There are different implementations and constraints that we can use to fit the single index models \citep{redd2025sensiatrpackageconducting}. In this paper, we apply the Gaussian kernel, and estimate $\beta_t$ and $h$ under the restriction that $\beta_t$ has norm 1. By default, initial values for $\beta_t$ are estimated using the Minimum Average Variance Estimation (MAVE) method \citep{mave_xia_2002, mave_wang_2008}. If returned NA, then we estimate the initial values of $\beta_t$ using the cumulative sliced inverse regression method \citep{slice_zhu_2010}. 

\citet{nabi2024semiparametric} has proven that the regression model for fixed bounded outcome $\rho(Y; X; \gamma_t)$, estimated using single index models, satisfy the following rate conditions:
$$\bigg|\bigg|\widehat{\widetilde{\mu}}(\rho(Y; X; \gamma_t); X)-\widetilde{\mu}(\rho(Y; X; \gamma_t); X)\bigg|\bigg|_{L_2}=O_p\left(\left(\frac{\log n}{n}\right)^{\frac{2}{2r+1}}\right)$$
when $\widetilde{F}_{t, r}(y, X\beta_{t, r}; \beta_{t, r})$ has Lipschitz (r + 1)th-order derivatives. With $r\geq1$, we have
$$\bigg|\bigg| \widehat{\widetilde{\mu}}_{t, 0}\big(\exp(\gamma_t s_t(Y)); X\big) - \widetilde{\mu}_{t, 0}\big(\exp(\gamma_t s_t(Y)); X\big) \bigg|\bigg|_{L_2}=O_P(n^{-2/5})$$ 
and 
$$\bigg|\bigg| \widehat{\widetilde{\mu}}_{t, 0}\big(Y\exp(\gamma_t s_t(Y)); X\big) - \widetilde{\mu}_{t, 0}\big(Y\exp(\gamma_t s_t(Y)); X\big) \bigg|\bigg|_{L_2}=O_P(n^{-2/5})$$
and 
$$\bigg|\bigg| \widehat{\widetilde{\mu}}_{t, r}\big(Y; X\big) - \widetilde{\mu}_{t, r}\big(Y; X\big) \bigg|\bigg|_{L_2}=O_P(n^{-2/5})$$

\subsection{Asymptotic properties of $\widehat{\psi}_t^\dagger(\gamma_t)$}\label{asymptotic}

\textbf{Proof of Theorem \ref{asymptotic_theorem}}
Since $\E[\phi_t(\widetilde{P}; \gamma_t)]=0$, we can decompose $\sqrt{n}(\widehat{\psi}_t^\dagger(\gamma_t)-\psi_t(\widetilde{P}; \gamma_t))$ into
\begin{align*}
    \sqrt{n}(\widehat{\psi}_t^\dagger(\gamma_t)-\psi_t(\widetilde{P}; \gamma_t)) &= \sqrt{n}(\widehat{\psi}_t^\dagger(\gamma_t)-\widehat{\psi}_t(\gamma_t)+\widehat{\psi}_t(\gamma_t)-\psi_t(\widetilde{P}; \gamma_t))\\
    &= \sqrt{n}(\widehat{\psi}_t^\dagger(\gamma_t)-\widehat{\psi}_t(\gamma_t))+\sqrt{n}\times\frac{1}{n}\sum_i\phi_t(\widetilde{P}; \gamma_t)(\widetilde{O}_i)+R_{n, 1}+R_{n, 2}
\end{align*}
where 
\begin{align*}
R_{n,1} &= \sqrt{n}\sum_{k=1}^K\left(\frac{n_k}{n}\right)(\frac{1}{n_k}\sum_{i: S_i=k}-\E) \left[ \phi_t(\widehat{\widetilde{P}}^{(-k)}; \gamma_t)(\widetilde{O})-\phi_t(\widetilde{P}; \gamma_t)(\widetilde{O})\right]\\
R_{n,2} &=\sqrt{n}\sum_{k=1}^K\left(\frac{n_k}{n}\right) Rem_t(\widehat{\widetilde{P}}^{(-k)}, \widetilde{P})
\end{align*}

\citet{nabi2024semiparametric} has proven the asymptotic negligibility of truncation procedure \citep{tuning_wang_2021}, that is, $\sqrt{n}(\widehat{\psi}_t^\dagger(\gamma_t)-\widehat{\psi}_t(\gamma_t))=o_p(1)$. By the sample splitting proposition in~\citet{kennedy2023semiparametricdoublyrobusttargeted}, to prove $R_{n,1}=o_P(1)$, it suffices to show $\left\lVert \phi_t(\widehat{\widetilde{P}}^{(-k)}; \gamma_t)(\widetilde{O})-\phi_t(\widetilde{P}; \gamma_t)(\widetilde{O})\right\rVert_{L_2}=o_p(1)$. Since we assume $|Y|$ and $|\exp(\gamma_ts_t(Y))|$ are bounded in probability, by the triangle and Cauchy-Schwarz inequalities, it is easy to show $\left\lVert \phi_t(\widehat{\widetilde{P}}^{(-k)}; \gamma_t)(\widetilde{O})-\phi_t(\widetilde{P}; \gamma_t)(\widetilde{O})\right\rVert_{L_2}=o_p(1)$. Thus, $R_{n,1}=o_P(1)$. 

Applying Jensen's inequality and Cauchy-Schwarz inequality to the results in Lemma~\ref{remainder}, we can have a upper bound for $\left|Rem_t(\widehat{\widetilde{P}}^{(-k)}, \widetilde{P})\right|$ as
{\small
    \begin{align*}
        &\lVert \widehat{\widetilde{\mu}}_{t, 0}^{(-k)}(Y; X) - \widetilde{\mu}_{t, 0}(Y; X) \rVert_{L_2}\lVert \widehat{\eta}_{1}^{(-k)}(X, R=0, T=t) - \eta_{1}(X, R=0, T=t) \rVert_{L_2} \\
        &\hspace{0.5cm}+ \bigg(\lVert \widehat{\widetilde{\mu}}^{(-k)}_{t, 0}(Y\exp(\gamma_t s_t(Y)); X) - \widetilde{\mu}_{t, 0}(Y\exp(\gamma_t s_t(Y)); X)\rVert_{L_2}+\lVert \widehat{\widetilde{\mu}}^{(-k)}_{t, 0}(\exp(\gamma_t s_t(Y)); X) - \widetilde{\mu}_{t, 0}(\exp(\gamma_t s_t(Y)); X)\rVert_{L_2}\bigg)\\
        &\hspace{1.5cm}\times\bigg(\lVert \widehat{\eta}_1^{(-k)}(X, R=0, T=t) - \eta_1(X, R=0, T=t) \rVert_{L_2}+\lVert \widehat{\pi}^{(-k)}_{t, 0}(X) - \pi_{t, 0}(X) \rVert_{L_2}\\
        &\hspace{3cm}+\lVert \widehat{\widetilde{\mu}}^{(-k)}_{t, 0}(\exp(\gamma_t s_t(Y)); X) - \widetilde{\mu}_{t, 0}(\exp(\gamma_t s_t(Y)); X)\rVert_{L_2}\bigg)\\
        &\hspace{0.5cm}+ \lVert \widehat{\widetilde{\mu}}^{(-k)}_{t, 1}(Y; X) - \widetilde{\mu}_{t, 1}(Y; X) \rVert_{L_2}\times\bigg(\lVert \widehat{\pi}^{(-k)}_{t, 1}(X) - \pi_{t, 1}(X) \rVert_{L_2}+\lVert \widehat{\eta}_{1}^{(-k)}(X, R=1, T=t) - \eta_{1}(X, R=1, T=t) \rVert_{L_2}\bigg)
    \end{align*}
    }
Since we apply the same nuisance estimation methods for each fold $k$, the upper bound for $\left|Rem_t(\widehat{\widetilde{P}}^{(-k)}, \widetilde{P})\right|$ is $o_P(n^{-1/2})$ for $k=1, ..., K$ under the estimation methods for $\widetilde{P}$ in Section~\ref{estimateP} and Section~\ref{sec:data}. Since $n_k=O(n)$, we have $R_{n,2}=o_P(1)$. 

Thus, 
\begin{eqnarray*}  
\sqrt{n}(\widehat{\psi}_t^\dagger(\gamma_t)-\psi_t(\widetilde{P}; \gamma_t))=\sqrt{n}\times\frac{1}{n}\sum_i\phi_t(\widetilde{P}; \gamma_t)(\widetilde{O}_i)+o_p(1)
\end{eqnarray*}
By the central limit theorem and $\E[\phi_t(\widetilde{P}; \gamma_t)(\widetilde{O})]=0$, 
\[
\sqrt{n}(\widehat{\psi}_t^\dagger(\gamma_t)-\psi_t(\widetilde{P}; \gamma_t))\stackrel{D(\widetilde{P})}{\rightarrow} \mathcal{N}(0, \E[\phi_t(\widetilde{P}; \gamma_t)(\widetilde{O})^2])
\]

\begin{lemma}\label{mis_asymp}
Assume $\widehat{\pi}_{t, r}(X)\xrightarrow{P}\pi_{t, r}^\ddagger(X)$, $\widehat{\eta}_1(X, r, t)\xrightarrow{P}\eta^\ddagger_1(X, r, t)$, and $\widehat{\widetilde{F}}_{t, r}(y|X)\xrightarrow{P}\widetilde{F}_{t, r}(y|X), r=0, 1$ where $\pi_{t, r}^\ddagger(X)\neq\pi_{t, r}(X)$ and $\eta^\ddagger_1(X, r, t)\neq\eta_1(X, r, t)$. Additionally assume (1) $|Y|$ and $|\exp(\gamma_ts_t(Y))|$ are bounded in probability, (2) $\pi_{t, r}^\ddagger(X)$, $\eta^\ddagger_1(X, r, t)$, $\mu_{t, 0}(\exp(\gamma_t s_t(Y);X)$ are bounded away from zero with probability one, (3) for any bounded functions of outcome $h(Y)$, $\lVert \widehat{\widetilde{\mu}}_{t, r}(h(Y); X) - \widetilde{\mu}_{t, r}(h(Y); X) \rVert_{L_2}=o_p(n^{-1/2})$. 

Let $\widetilde{P}^\ddagger=\widetilde{P}\backslash\{\pi_{t, r}(X), \eta_1(X, r, t): t, r\in\{0, 1\}\}\cup\{\pi_{t, r}^\ddagger(X), \eta^\ddagger_1(X, r, t): t, r\in\{0, 1\}\}$. Under Assumptions (A1-A5),
\[
\sqrt{n}(\widehat{\psi}_t^\dagger(\gamma_t)-\psi_t(\widetilde{P}; \gamma_t))\stackrel{D(\widetilde{P}^\ddagger)}{\rightarrow} \mathcal{N}(0, \E[\phi_t(\widetilde{P}^\ddagger; \gamma_t)(\widetilde{O})^2])
\]
\end{lemma}

\textbf{Proof of Lemma \ref{mis_asymp}}

\begin{align*}
    \sqrt{n}(\widehat{\psi}_t^\dagger(\gamma_t)-\psi_t(\widetilde{P}; \gamma_t)) &=\sqrt{n}(\widehat{\psi}_t^\dagger(\gamma_t)-\widehat{\psi}_t(\gamma_t)+\widehat{\psi}_t(\gamma_t)-\psi_t(\widetilde{P}; \gamma_t))\\
    &= \sqrt{n}(\widehat{\psi}_t^\dagger(\gamma_t)-\widehat{\psi}_t(\gamma_t))+\sqrt{n}\times\frac{1}{n}\sum_i\phi_t(\widetilde{P}; \gamma_t)(\widetilde{O}_i)\\
    &\hspace{1cm}+\sqrt{n}\sum_{k=1}^K\left(\frac{n_k}{n}\right)(\frac{1}{n_k}\sum_{i: S_i=k}-\E) \left[ \phi_t(\widehat{\widetilde{P}}^{(-k)}; \gamma_t)(\widetilde{O})-\phi_t(\widetilde{P}; \gamma_t)(\widetilde{O})\right]\\
    &\hspace{1cm}+\sqrt{n}\sum_{k=1}^K\left(\frac{n_k}{n}\right) Rem_t(\widehat{\widetilde{P}}^{(-k)}, \widetilde{P})\\
    &= \sqrt{n}(\widehat{\psi}_t^\dagger(\gamma_t)-\widehat{\psi}_t(\gamma_t))+\sqrt{n}\times\frac{1}{n}\sum_i\phi_t(\widetilde{P}^\ddagger; \gamma_t)(\widetilde{O}_i)\\
    &\hspace{1cm}-\sqrt{n}\E\left[\phi_t(\widetilde{P}^\ddagger; \gamma_t)(\widetilde{O})\right]+\sqrt{n}\E\left[\phi_t(\widetilde{P}; \gamma_t)(\widetilde{O})\right]\\
    &\hspace{1cm}+\sqrt{n}\sum_{k=1}^K\left(\frac{n_k}{n}\right)(\frac{1}{n_k}\sum_{i: S_i=k}-\E) \left[ \phi_t(\widehat{\widetilde{P}}^{(-k)}; \gamma_t)(\widetilde{O})-\phi_t(\widetilde{P}^\ddagger; \gamma_t)(\widetilde{O})\right]\\
    &\hspace{1cm}+\sqrt{n}\sum_{k=1}^K\left(\frac{n_k}{n}\right) Rem_t(\widehat{\widetilde{P}}^{(-k)}, \widetilde{P})
\end{align*}
Using Theorem~\ref{theorem_robustness}, we know $\E\left[\phi_t(\widetilde{P}^\ddagger; \gamma_t)(\widetilde{O})\right]=\E\left[\phi_t(\widetilde{P}; \gamma_t)(\widetilde{O})\right]=0$. As mentioned in Theorem~\ref{theorem_robustness}, we have $\sqrt{n}(\widehat{\psi}_t^\dagger(\gamma_t)-\widehat{\psi}_t(\gamma_t))=o_p(1)$. By sample splitting and $\left\lVert \phi_t(\widehat{\widetilde{P}}^{(-k)}; \gamma_t)(\widetilde{O})-\phi_t(\widetilde{P}^\ddagger; \gamma_t)(\widetilde{O})\right\rVert_{L_2}\allowbreak=o_p(1)$, we have $\sqrt{n}\sum_{k=1}^K\left(\frac{n_k}{n}\right)(\frac{1}{n_k}\sum_{i: S_i=k}-\E) \left[ \phi_t(\widehat{\widetilde{P}}^{(-k)}; \gamma_t)(\widetilde{O})-\phi_t(\widetilde{P}^\ddagger; \gamma_t)(\widetilde{O})\right]=o_p(1)$. \\
By the triangle inequality, for $r=0, 1$, 
\begin{align*}
    \lVert \widehat{\pi}^{(-k)}_{t, r}(X) - \pi_{t, r}(X) \rVert_{L_2}&=\lVert \widehat{\pi}^{(-k)}_{t, r}(X) - \pi_{t, r}^\ddagger(X)+ \pi_{t, r}^\ddagger(X)-\pi_{t, r}(X)\rVert_{L_2}\\
    &\leq \lVert \widehat{\pi}^{(-k)}_{t, r}(X) - \pi_{t, r}^\ddagger(X)\rVert_{L_2}+\lVert \pi_{t, r}^\ddagger(X)-\pi_{t, r}(X)\rVert_{L_2}
\end{align*}
and
\begin{align*}
    \lVert \widehat{\eta}_{1}^{(-k)}(X, r, t) - \eta_{1}(X, r, t) \rVert_{L_2}\leq\lVert \widehat{\eta}_{1}^{(-k)}(X, r, t) - \eta_{1}^\ddagger(X, r, t) \rVert_{L_2}+\lVert \eta_{1}^\ddagger(X, r, t)- \eta_{1}(X, r, t) \rVert_{L_2}.
\end{align*}
Under the assumption that $\lVert \widehat{\widetilde{\mu}}_{t, r}(h(Y); X) - \widetilde{\mu}_{t, r}(h(Y); X) \rVert_{L_2}=o_p(n^{-1/2})$, we have 
\begin{align*}
    &\lVert \widehat{\widetilde{\mu}}_{t, r}^{(-k)}(h(Y); X) - \widetilde{\mu}_{t, r}(h(Y); X) \rVert_{L_2}\lVert \widehat{\pi}^{(-k)}_{t, r}(X) - \pi_{t, r}(X) \rVert_{L_2}\\
    &\hspace{1cm}\leq \lVert \widehat{\widetilde{\mu}}_{t, r}^{(-k)}(h(Y); X) - \widetilde{\mu}_{t, r}(h(Y); X) \rVert_{L_2}\lVert \widehat{\pi}^{(-k)}_{t, r}(X) - \pi_{t, r}^\ddagger(X)\rVert_{L_2}\\
    &\hspace{1.5cm}+\lVert \widehat{\widetilde{\mu}}_{t, r}^{(-k)}(h(Y); X) - \widetilde{\mu}_{t, r}(h(Y); X) \rVert_{L_2}\lVert \pi_{t, r}^\ddagger(X)-\pi_{t, r}(X)\rVert_{L_2}\\
    &\hspace{1cm}=o_P(n^{-1/2})
\end{align*}
and 
\begin{align*}
    &\lVert \widehat{\widetilde{\mu}}_{t, r}^{(-k)}(h(Y); X) - \widetilde{\mu}_{t, r}(h(Y); X) \rVert_{L_2}\lVert \widehat{\eta}_{1}^{(-k)}(X, r, t) - \eta_{1}(X, r, t) \rVert_{L_2}\\
    &\hspace{1cm}\leq \lVert \widehat{\widetilde{\mu}}_{t, r}^{(-k)}(h(Y); X) - \widetilde{\mu}_{t, r}(h(Y); X) \rVert_{L_2}\lVert \widehat{\eta}_{1}^{(-k)}(X, r, t) - \eta_{1}^\ddagger(X, r, t) \rVert_{L_2}\\
    &\hspace{1.5cm}+\lVert \widehat{\widetilde{\mu}}_{t, r}^{(-k)}(h(Y); X) - \widetilde{\mu}_{t, r}(h(Y); X) \rVert_{L_2}\lVert \eta_{1}^\ddagger(X, r, t)- \eta_{1}(X, r, t) \rVert_{L_2}\\
    &\hspace{1cm}=o_P(n^{-1/2}).
\end{align*}
Since the upper bound for $\left|Rem_t(\widehat{\widetilde{P}}^{(-k)}, \widetilde{P})\right|$ is the summation of product terms between the outcome models and propensity scores/missing outcome probabilities, we know that the upper bound for $\left|Rem_t(\widehat{\widetilde{P}}^{(-k)}, \widetilde{P})\right|$ is $o_P(n^{-1/2})$. Thus, $\sqrt{n}\sum_{k=1}^K\left(\frac{n_k}{n}\right) Rem_t(\widehat{\widetilde{P}}^{(-k)}, \widetilde{P})=o_p(1)$. And we have
\begin{eqnarray*}  
\sqrt{n}(\widehat{\psi}_t^\dagger(\gamma_t)-\psi_t(\widetilde{P}; \gamma_t))=\sqrt{n}\times\frac{1}{n}\sum_i\phi_t(\widetilde{P}^\ddagger; \gamma_t)(\widetilde{O}_i)+o_p(1).
\end{eqnarray*}
By the central limit theorem and $\E[\phi_t(\widetilde{P}^\ddagger; \gamma_t)(\widetilde{O})]=0$, 
\[
\sqrt{n}(\widehat{\psi}_t^\dagger(\gamma_t)-\psi_t(\widetilde{P}; \gamma_t))\stackrel{D(\widetilde{P}^\ddagger)}{\rightarrow} \mathcal{N}(0, \E[\phi_t(\widetilde{P}^\ddagger; \gamma_t)(\widetilde{O})^2]).
\]

In conclusion, when $\pi_{t, r}(X)$ and $\eta_1(X, r, t), r=0, 1$ are mis-specified, we can still obtain a $\sqrt{n}$-consistent estimator of $\psi(\widetilde{P};\gamma_t)$ when $\lVert \widehat{\widetilde{\mu}}_{t, r}^{(-k)}(h(Y); X) - \widetilde{\mu}_{t, r}(h(Y); X) \rVert_{L_2}=o_P(n^{-1/2})$.

\section{Proofs for $\E[Y(t)|R=r]$ under Assumptions (A1-A5)}
\label{app:proofs_R0R1}

We can similarly derive the influence function for $\psi_{t, 1}(\widetilde{P})$ and $\psi_{t, 0}(\widetilde{P}; \gamma_t)$. 
\begin{theorem}
    An influence function for $\psi_{t, 1}(\widetilde{P})$ in \eqref{eq:parameter_id1}, denoted by $\phi_{t, 1}(\widetilde{P})(\widetilde{O})$, is:
    \begin{equation*}
        \phi_{t, 1}(\widetilde{P})(\widetilde{O})=\underbrace{\frac{M}{\eta_1(X, R, T)}\upsilon^c_{t, 1}(\widetilde{P})(O)+\Big\{ 1 - \frac{M}{\eta_{1}(X, R, T)}  \Big\} a_{t, 1}(\widetilde{P})(X, R, T)}_{\upsilon_{t, 1}(\widetilde{P})(\widetilde{O})} -\frac{\I(R=1)}{P(R=1)}\psi_{t, 1}(\widetilde{P})
    \end{equation*}
    where $\upsilon^c_{t, 1}(\widetilde{P})(O)$ is given by
    \begin{equation*}
        \begin{aligned}
            \upsilon^c_{t, 1}(\widetilde{P})(O)= \frac{\I(R=1)}{P(R=1)}\left\{\frac{\I(T=t)}{\pi_{t, 1}(X)}\left(Y-\widetilde{\mu}_{t, 1}(Y; X)\right)+\widetilde{\mu}_{t, 1}(Y; X)\right\}
        \end{aligned}
    \end{equation*}
    and $a_{t, 1}(\widetilde{P})(X, R, T)$ is given by
    \begin{equation*}
        \begin{aligned}
            a_{t, 1}(\widetilde{P})(X, R, T)= \frac{\I(R=1)}{P(R=1)}\widetilde{\mu}_{t, 1}(Y; X)
        \end{aligned}
    \end{equation*}
\end{theorem}
\begin{theorem}
    The non-parametric influence function for $\psi_{t, 0}(\widetilde{P}; \gamma_t)$ in \eqref{eq:parameter_id1}, denoted by $\phi_{t, 0}(\widetilde{P}; \gamma_t)(\widetilde{O})$, is:
    \begin{align*}
    \phi_{t, 0}(\widetilde{P}; \gamma_t)(\widetilde{O})
    &=  \underbrace{\frac{M}{\eta_{1}(X, R, T)}  \upsilon^c_{t, 0}(\widetilde{P}; \gamma_t)(O) +  \Big\{ 1 - \frac{M}{\eta_{1}(X, R, T)}  \Big\} a_{t, 0}(\widetilde{P}; \gamma_t)(X, R, T)}_{\upsilon_{t, 0}(\widetilde{P}; \gamma_t)(\widetilde{O})}-\frac{\mathbb{I}(R=0)}{P(R=0)}\psi_{t, 0}(\widetilde{P}; \gamma_t)  \ , 
\end{align*}
where $\upsilon^c_{t, 0}(P; \gamma_t)(O)$ is given by
{\small
\begin{equation*} 
\begin{aligned}
    \upsilon^c_{t, 0}(P; \gamma_t)(O) 
    =& \frac{\mathbb{I}(R = 0, T = t)}{P(R=0)}\times \Bigg\{ Y +  \frac{\pi_{1-t,0}(X)}{\pi_{t,0}(X)} \times \frac{\exp(\gamma_t s_t(Y))}{\mu_{t,0}(\exp(\gamma_t s_t(Y)); X)}  \left\{Y  - \frac{\mu_{t,0}(Y\exp(\gamma_t s_t(Y)); X)}{\mu_{t,0}(\exp(\gamma_t s_t(Y)); X)} \right\} \Bigg\} \\
    &\hspace{0.25cm} + \frac{\mathbb{I}(R=0, T=1-t)}{P(R=0)} \times \frac{\mu_{t,0}(Y\exp(\gamma_t s_t(Y)); X)}{\mu_{t,0}(\exp(\gamma_t s_t(Y)); X)}, 
\end{aligned}
\end{equation*}
and $a_{t, 0}(\widetilde{P}; \gamma_t)(X, R, T)$ is given by
\begin{equation*}
\begin{aligned}
    a_{t, 0}(\widetilde{P}; \gamma_t)(X, R, T)
    =& \frac{\I(R=0, T=t)}{P(R=0)} \ \mu_{t,0}(Y; X) + \frac{\I(R=0, T=1-t)}{P(R=0)}\ \frac{\mu_{t,0}(Y\exp(\gamma_t s_t(Y)); X)}{\mu_{t,0}(\exp(\gamma_t s_t(Y)); X)}
\end{aligned}
\end{equation*}}
\end{theorem}

We note that the influence function for $\psi_{t}(\widetilde{P}; \gamma_t)$ is a weighted combination of the influence functions for $\psi_{t, 1}(\widetilde{P})$ and $\psi_{t, 0}(\widetilde{P})$:
\begin{align*}
    \phi_t(\widetilde{P}; \gamma_t)(\widetilde{O})=&\phi_{t, 1}(\widetilde{P})(\widetilde{O})P(R=1)+\psi_{t, 1}(\widetilde{P})\left(\I(R=1)-P(R=1)\right)\\
    &\hspace{0.25cm}+\phi_{t, 0}(\widetilde{P}; \gamma_t)(\widetilde{O})P(R=0)+\psi_{t, 0}(\widetilde{P}; \gamma_t)\left(\I(R=0)-P(R=0)\right),
\end{align*}
where $\phi_{t, 1}(\widetilde{P})(\widetilde{O})$ and $\phi_{t, 0}(\widetilde{P}; \gamma_t)(\widetilde{O})$ are influence functions of $\psi_{t, 1}(\widetilde{P})$ and $\psi_{t, 0}(\widetilde{P}; \gamma_t)$, respectively. This complies with the product rule of influence functions in \citet{kennedy2023semiparametricdoublyrobusttargeted}.

Our one-step, split sample estimators for $\psi_{t, 1}(\widetilde{P})$ and $\psi_{t, 0}(\widetilde{P}; \gamma_t)$ are:
\begin{equation*}
    	\widehat{\psi}_{t, 1} = \frac{1}{K} \sum_{k=1}^K \underbrace{\left\{ \psi_{t, 1} \left( \widehat{\widetilde{P}}^{(-k)}\right) + \frac{1}{n_k} \sum_{i: S_i=k} \phi_{t, 1} \left(\widehat{\widetilde{P}}^{(-k)}\right) (\widetilde{O}_i) \right\}}_{\mbox{$k$th Split One-Step Estimator } \left(\widehat{\psi}_{t, 1}^{(k)}\right)}\,
\end{equation*}
and 
\begin{equation*}
    	\widehat{\psi}_{t, 0}(\gamma_t) = \frac{1}{K} \sum_{k=1}^K \underbrace{\left\{ \psi_{t, 0} \left( \widehat{\widetilde{P}}^{(-k)}; \gamma_t \right) + \frac{1}{n_k} \sum_{i: S_i=k} \phi_{t, 0} \left(\widehat{\widetilde{P}}^{(-k)} ; \gamma_t\right) (\widetilde{O}_i) \right\}}_{\mbox{$k$th Split One-Step Estimator } \left(\widehat{\psi}_{t, 0}^{(k)}(\gamma_t)\right)}. 
\end{equation*}

We apply the same truncation procedure as in (\ref{truncate_estimator}), (\ref{trunc1}) and ($\ref{trunc2}$). We additionally truncate the plug-in estimators $\psi_{t, 1} \left( \widehat{\widetilde{P}}^{(-k)}\right)$ and $\psi_{t, 0} \left( \widehat{\widetilde{P}}^{(-k)}; \gamma_t\right)$ to $\psi_{t, 1}^\dagger \left( \widehat{\widetilde{P}}^{(-k)}\right)$ and $\psi_{t, 0}^\dagger \left( \widehat{\widetilde{P}}^{(-k)}; \gamma_t\right)$ in the same fashion. Finally, we obtain the $k$th split, one-step, truncated estimators $\widehat{\psi}_{t, 1}^{(k)\dagger}$ and $\widehat{\psi}_{t, 0}^{(k)\dagger}(\gamma_t)$:
\begin{equation*}
    	\widehat{\psi}_{t, 1}^{(k)\dagger} = \psi_{t, 1}^\dagger \left( \widehat{\widetilde{P}}^{(-k)}\right)+\frac{1}{n_k} \sum_{i: S_i=k}\upsilon_{t, 1}^\dagger\left(\widehat{\widetilde{P}}^{(-k)}\right)(\widetilde{O}_i)-\frac{\I(R_i=1)}{P^{(-k)}_{n/n_k}(R=1)}\psi_{t, 1}^\dagger \left( \widehat{\widetilde{P}}^{(-k)}\right)
\end{equation*}
\begin{equation*}
    	\widehat{\psi}_{t, 0}^{(k)\dagger}(\gamma_t) = \psi_{t, 0}^\dagger \left( \widehat{\widetilde{P}}^{(-k)}; \gamma_t\right)+\frac{1}{n_k} \sum_{i: S_i=k}\upsilon_{t, 0}^\dagger\left(\widehat{\widetilde{P}}^{(-k)}; \gamma_t\right)(\widetilde{O}_i)-\frac{\I(R_i=0)}{P^{(-k)}_{n/n_k}(R=0)}\psi_{t, 0}^\dagger \left( \widehat{\widetilde{P}}^{(-k)}; \gamma_t\right).
\end{equation*}

We estimate the variance of $\widehat{\psi}^\dagger_{t, 1}$ by
$$\frac{1}{nK}\sum_{k=1}^K\left\{\frac{1}{n_k-1}\sum_{i: S_i=k}\left\{\upsilon_{t, 1}^\dagger \left(\widehat{\widetilde{P}}^{(-k)}\right) (\widetilde{O}_i)-\frac{\I(R_i=1)}{P^{(-k)}_{n/n_k}(R=1)}\widehat{\psi}_{t, 1}^{(k)\dagger}\right\}^2\right\},$$
and the variance of $\widehat{\psi}^\dagger_{t, 0}(\gamma_t)$ by
$$\frac{1}{nK}\sum_{k=1}^K\left\{\frac{1}{n_k-1}\sum_{i: S_i=k}\left\{\upsilon^\dagger_{t, 0} \left(\widehat{\widetilde{P}}^{(-k)} ; \gamma_t\right) (\widetilde{O}_i)-\frac{\I(R_i=0)}{P^{(-k)}_{n/n_k}(R=0)}\widehat{\psi}_{t, 0}^{(k)\dagger}(\gamma_t)\right\}^2\right\}.$$

\section{Estimation of $\psi_t'(\widetilde{P}; \gamma_t')$ \& $\psi_{t, 0}'(\widetilde{P}; \gamma_t')$ under Assumptions (A1-A3), (A5), (\ref{exchangeability})}
\label{app:model_2}

\begin{theorem}\label{EIF_miss_ex}
An influence function for $\psi_t'(\widetilde{P}; \gamma_t')$ in \eqref{eq:parameter_id_ex}, denoted by $\phi_t'(\widetilde{P}; \gamma_t')(\widetilde{O})$, is:
    \begin{align*}
    \phi_t'(\widetilde{P}; \gamma_t')(\widetilde{O})
    &=  \underbrace{\upsilon_{t, 1}(\widetilde{P})(\widetilde{O})\times P(R=1)+\upsilon_{t, 0}'(\widetilde{P}; \gamma_t')(\widetilde{O})\times P(R=0)}_{\upsilon_{t}'(\widetilde{P}; \gamma_t')(\widetilde{O})} - \psi_t'(\widetilde{P}; \gamma_t') \ , 
\end{align*}
where $\upsilon_{t, 1}(\widetilde{P})(\widetilde{O})$ is the same as in Theorem \ref{EIF_miss}, and $\upsilon_{t, 0}'(\widetilde{P}; \gamma_t')(\widetilde{O})$ is given by
\begin{equation*}
    \begin{aligned}
   \upsilon_{t, 0}'(\widetilde{P}; \gamma_t')(\widetilde{O})
    =& \frac{M}{\eta_1(X, R, T)}\Bigg[ \frac{\mathbb{I}(T=t, R=1)}{\pi_{t, 1}(X)P(R=0)} \times \frac{g_{0}(X)}{g_{1}(X)} \times \frac{\exp(\gamma_t' s_t(Y))}{\mu_{t,1}(\exp(\gamma_t' s_t(Y)); X)}  \\
    &\hspace{6cm}\times\left\{Y  - \frac{\widetilde{\mu}_{t,1}(Y\exp(\gamma_t' s_t(Y)); X)}{\widetilde{\mu}_{t,1}(\exp(\gamma_t' s_t(Y)); X)} \right\}\\
    &\hspace{2.5cm}+\frac{\I(R=0)}{P(R=0)} \times \frac{\widetilde{\mu}_{t,1}(Y\exp(\gamma_t' s_t(Y)); X)}{\widetilde{\mu}_{t,1}(\exp(\gamma_t' s_t(Y)); X)}\Bigg]\\
    &\hspace{0.25cm}+\Big\{ 1 - \frac{M}{\eta_{1}(X, R, T)}  \Big\}\times \frac{\I(R=0)}{P(R=0)} \times\frac{\widetilde{\mu}_{t,1}(Y\exp(\gamma_t' s_t(Y)); X)}{\widetilde{\mu}_{t,1}(\exp(\gamma_t' s_t(Y)); X)} \ .
\end{aligned}
\end{equation*}
\label{thm:EIF_miss_ex}
\end{theorem}

\begin{theorem}\label{EIF_miss2_ex}
The non-parametric influence function for $\psi_{t, 0}'(\widetilde{P}; \gamma_t')$ in \eqref{eq:parameter_id_ex}, denoted by $\phi_{t, 0}'(\widetilde{P}; \gamma_t')(\widetilde{O})$, is:
\begin{align*}
    \phi_{t, 0}'(\widetilde{P}; \gamma_t')(\widetilde{O})
    &=  \underbrace{\frac{M}{\eta_{1}(X, R, T)}  \phi^{c'}_{t, 0}(\widetilde{P}; \gamma_t')(O) +  \Big\{ 1 - \frac{M}{\eta_{1}(X, R, T)}  \Big\} a_{t, 0}'(\widetilde{P}; \gamma_t')(X, R, T)}_{\upsilon_{t, 0}'(\widetilde{P}; \gamma_t')(\widetilde{O})}- \frac{\I(R=0)}{P(R=0)}\psi_{t, 0}'(\widetilde{P}; \gamma_t')  \ , 
\end{align*}
where $\phi^{c'}_{t, 0}(\widetilde{P}; \gamma_t')(O)$ is given by 
{\small 
\begin{equation*}
\begin{aligned}
    \phi^{c'}_{t, 0}(P; \gamma_t')(O) 
    =& \frac{\mathbb{I}(T=t, R=1)}{\pi_{t, 1}(X)P(R=0)} \times \frac{g_{0}(X)}{g_{1}(X)} \times \frac{\exp(\gamma_t' s_t(Y))}{\mu_{t,1}(\exp(\gamma_t' s_t(Y)); X)}  \left\{Y  - \frac{\mu_{t,1}(Y\exp(\gamma_t' s_t(Y)); X)}{\mu_{t,1}(\exp(\gamma_t' s_t(Y)); X)} \right\} \\
    &\hspace{0.25cm} + \frac{\I(R=0)}{P(R=0)} \times \frac{\mu_{t,1}(Y\exp(\gamma_t' s_t(Y)); X)}{\mu_{t,1}(\exp(\gamma_t' s_t(Y)); X)} ,
\end{aligned}
\end{equation*}
}
and $a_{t, 0}'(\widetilde{P}; \gamma_t')(X, R, T)$ is given by 
{\small 
\begin{equation*} 
\begin{aligned}
   a_{t, 0}'(\widetilde{P}; \gamma_t')(X, R, T) &= \frac{\I(R=0)}{P(R=0)}\times\frac{\mu_{t,1}(Y\exp(\gamma_t' s_t(Y)); X)}{\mu_{t,1}(\exp(\gamma_t' s_t(Y)); X)} . 
\end{aligned}
\end{equation*}
}
\end{theorem}

Our one-step, split sample estimators are:
\begin{equation*}
    	\widehat{\psi}_t'(\gamma_t') = \frac{1}{K} \sum_{k=1}^K \underbrace{\left\{\frac{1}{n_k} \sum_{i: S_i=k}\Bigg[ \upsilon_{t, 1} \left(\widehat{\widetilde{P}}^{(-k)}\right) (\widetilde{O}_i)\times P_{n/n_k}^{(-k)}(R=1) +\upsilon_{t, 0}' \left(\widehat{\widetilde{P}}^{(-k)} ; \gamma_t'\right) (\widetilde{O}_i)\times P_{n/n_k}^{(-k)}(R=0)\Bigg]\right\}}_{\mbox{$k$th Split One-Step Estimator } \left(\widehat{\psi}_t^{(k)'}(\gamma_t')\right)}\, 
\end{equation*}
and 
\begin{equation*}
    	\widehat{\psi}_{t, 0}'(\gamma_t') = \frac{1}{K} \sum_{k=1}^K \underbrace{\left\{ \psi_{t, 0}' \left( \widehat{\widetilde{P}}^{(-k)}; \gamma_t' \right) + \frac{1}{n_k} \sum_{i: S_i=k} \phi_{t, 0}' \left(\widehat{\widetilde{P}}^{(-k)} ; \gamma_t'\right) (\widetilde{O}_i) \right\}}_{\mbox{$k$th Split One-Step Estimator } \left(\widehat{\psi}_{t, 0}^{(k)'}(\gamma_t')\right)}. 
\end{equation*}

The $k$th split, one-step, truncated estimator $\widehat{\psi}_t^{(k)'\dagger}(\gamma_t')$ is 
\begin{equation*}
\begin{aligned}
    	\widehat{\psi}_t^{(k)'\dagger}(\gamma_t') = \frac{1}{n_k} \sum_{i: S_i=k}&\Bigg[\upsilon_{t, 1}^\dagger\left(\widehat{\widetilde{P}}^{(-k)}\right)(\widetilde{O}_i)\times P_{n/n_k}^{(-k)}(R=1)\\
        &+\underbrace{\min\left\{\left|\upsilon_{t, 0}'\left(\widehat{\widetilde{P}}^{(-k)} ; \gamma_t'\right)(\widetilde{O}_i)\right|, \widehat\tau_{t, 0}^{'(k)}\right\}\text{sign}\left\{\upsilon_{t, 0}'\left(\widehat{\widetilde{P}}^{(-k)} ; \gamma_t'\right)(\widetilde{O}_i)\right\}}_{\upsilon_{t, 0}^{'\dagger}\left(\widehat{\widetilde{P}}^{(-k)} ; \gamma_t'\right)(\widetilde{O}_i)}\times P_{n/n_k}^{(-k)}(R=0)\Bigg]\ ,
\end{aligned}
\end{equation*}
where $\upsilon_{t, 1}^\dagger\left(\widehat{\widetilde{P}}^{(-k)}\right)(\widetilde{O}_i)$ is the same as in (\ref{truncate_estimator}) and $\widehat\tau_{t, 0}^{'(k)}$ is the solution to 
\begin{equation*}
    	\frac{1}{n_k} \sum_{i: S_i=k}\frac{\min\left\{\upsilon_{t, 0}'\left(\widehat{\widetilde{P}}^{(-k)} ; \gamma_t'\right)(\widetilde{O}_i)^2, \left(\widehat\tau_{t, 0}^{'(k)}\right)^2\right\}}{\left(\widehat\tau_{t, 0}^{'(k)}\right)^2}=\frac{\log(n_k)}{n_k}.
\end{equation*}
And the $k$th split, one-step, truncated estimator $\widehat{\psi}_{t, 0}^{'(k)\dagger}$ is
\begin{equation*}
    	\widehat{\psi}_{t, 0}^{'(k)\dagger} = \psi_{t, 0}' \left( \widehat{\widetilde{P}}^{(-k)}; \gamma_t'\right)+\frac{1}{n_k} \sum_{i: S_i=k}\upsilon_{t, 0}^{'\dagger}\left(\widehat{\widetilde{P}}^{(-k)}; \gamma_t'\right)(\widetilde{O}_i)-\frac{\I(R_i=0)}{P^{(-k)}_{n/n_k}(R=0)}\psi_{t, 0}' \left( \widehat{\widetilde{P}}^{(-k)}; \gamma_t'\right).
\end{equation*}
We estimate the variance of $\widehat{\psi}_t^{'\dagger}(\gamma_t')$ by
$$\frac{1}{nK}\sum_{k=1}^K\left\{\frac{1}{n_k-1}\sum_{i: S_i=k}\left\{\upsilon_t^\dagger\left(\widehat{\widetilde{P}}^{(-k)} ; \gamma_t\right)(\widetilde{O}_i)-\widehat{\psi}_t^{(k)\dagger}(\gamma_t)\right\}^2\right\},$$
where 
$$\upsilon_t^\dagger\left(\widehat{\widetilde{P}}^{(-k)} ; \gamma_t\right)(\widetilde{O}_i)=\upsilon_{t, 1}^\dagger\left(\widehat{\widetilde{P}}^{(-k)}\right)(\widetilde{O}_i)\times P_{n/n_k}^{(-k)}(R=1)+\upsilon_{t, 0}^\dagger\left(\widehat{\widetilde{P}}^{(-k)} ; \gamma_t\right)(\widetilde{O}_i)\times P_{n/n_k}^{(-k)}(R=0)\ .$$

\begin{lemma}\label{onetoone}
Given $s_t(Y)$ is a monotone function of $Y$ (either increasing or decreasing) and $|Y|$ is bounded, there is a unique mapping between $\gamma_t'$ and $\gamma_t$ through their influence on $\E[Y(t)]$. That is, for each $\gamma_t$, there exist one and only one $\gamma_t'$ that minimize $\{\psi_t(\widetilde{P}; \gamma_t)-\psi_t'(\widetilde{P}; \gamma_t')\}^2$. Similarly, for each $\gamma_t'$, there exist one and only one $\gamma_t$ that minimize $\{\psi_t(\widetilde{P}; \gamma_t)-\psi_t'(\widetilde{P}; \gamma_t')\}^2$. 
\end{lemma}
\textbf{Proof of Lemma \ref{onetoone}} 
Under Assumptions (A1)-(A5), we can write (\ref{eq:parameter_id1}) as
$$\psi_t(\widetilde{P}; \gamma_t) = \E \bigg[g_1(X) \, \widetilde{\mu}_{t,1}(Y;X)+g_0(X) \Big\{ \widetilde{\mu}_{t,0}(Y;X) \pi_{t,0}(X) + \frac{ \widetilde{\mu}_{t,0}(Y \! \exp\{  \gamma_t s_t(Y)\} ;X) }{\widetilde{\mu}_{t,0}( \exp\{  \gamma_t s_t(Y)\} ;X) } \pi_{1-t,0}(X)  \Big\} \bigg].$$
Under Assumptions (A1-A3), (A5) and (\ref{exchangeability}), we re-write (\ref{eq:parameter_id_ex}) as
$$\psi_t'(\widetilde{P}; \gamma_t')=\E\left[ g_1(X) \, \widetilde{\mu}_{t,1}(Y;X)+g_0(X) \frac{\mu_{t, 1}(Y \exp(\gamma_t' s_t(Y)); X)}{\mu_{t, 1}(\exp(\gamma_t' s_t(Y)); X)}\right].$$

To prove this lemma, we will first prove $\psi_t(\widetilde{P}; \gamma_t)$ and $\psi_t'(\widetilde{P}; \gamma_t')$ are strictly monotone with respect to $\gamma_t$ and $\gamma_t'$, respectively. 
\begin{align*}
    \frac{\partial\psi_t(\widetilde{P}; \gamma_t)}{\partial\gamma_t}=\E \bigg[\frac{g_0(X)\pi_{1-t,0}(X)}{\widetilde{\mu}_{t,0}( \exp\{  \gamma_t s_t(Y)\} ;X)}\bigg\{&\widetilde{\mu}_{t,0}(Ys_t(Y) \! \exp\{  \gamma_t s_t(Y)\} ;X)\widetilde{\mu}_{t,0}(\exp\{  \gamma_t s_t(Y)\} ;X)\\
    &-\widetilde{\mu}_{t,0}(Y\! \exp\{  \gamma_t s_t(Y)\} ;X)\widetilde{\mu}_{t,0}(s_t(Y)\! \exp\{  \gamma_t s_t(Y)\} ;X)\bigg\}\bigg]
\end{align*}
Take $Y^\dagger$ to be an independent observation of $Y$, and since $\widetilde{\mu}_{t, 0}(\cdot)$ are conditional expected values of functions of $Y$, we can show that
\begin{align*}
    &\E\left[\exp(s_t(Y)+s_t(Y^\dagger))(Y-Y^\dagger)(s_t(Y)-s_t(Y^\dagger))\mid T=t, R=0, X, M=1\right]\\
    =&\E\left[Ys_t(Y)\exp(s_t(Y)+s_t(Y^\dagger))\mid T=t, R=0, X, M=1\right]\\
    &-\E\left[Ys_t(Y^\dagger)\exp(s_t(Y)+s_t(Y^\dagger))\mid T=t, R=0, X, M=1\right]\\
    &-\E\left[Y^\dagger s_t(Y)\exp(s_t(Y)+s_t(Y^\dagger))\mid T=t, R=0, X, M=1\right]\\
    &+\E\left[Y^\dagger s_t(Y^\dagger)\exp(s_t(Y)+s_t(Y^\dagger))\mid T=t, R=0, X, M=1\right]\\
    =& \E\left[Ys_t(Y)\exp(s_t(Y))\mid T=t, R=0, X, M=1\right]\E\left[\exp(s_t(Y^\dagger))\mid T=t, R=0, X, M=1\right]\\
    &-\E\left[Y\exp(s_t(Y))\mid T=t, R=0, X, M=1\right]\E\left[s_t(Y^\dagger)\exp(s_t(Y^\dagger))\mid T=t, R=0, X, M=1\right]\\
    &-\E\left[s_t(Y)\exp(s_t(Y))\mid T=t, R=0, X, M=1\right]\E\left[Y^\dagger \exp(s_t(Y^\dagger))\mid T=t, R=0, X, M=1\right]\\
    &+\E\left[\exp(s_t(Y))\mid T=t, R=0, X, M=1\right]\E\left[Y^\dagger s_t(Y^\dagger)\exp(s_t(Y^\dagger))\mid T=t, R=0, X, M=1\right]\\
    =&2\E\left[Ys_t(Y)\exp(s_t(Y))\mid T=t, R=0, X, M=1\right]\E\left[\exp(s_t(Y))\mid T=t, R=0, X, M=1\right]\\
    &-2\E\left[Y\exp(s_t(Y))\mid T=t, R=0, X, M=1\right]\E\left[s_t(Y)\exp(s_t(Y))\mid T=t, R=0, X, M=1\right]\\
    =&2\bigg\{\widetilde{\mu}_{t,0}(Ys_t(Y) \! \exp\{  \gamma_t s_t(Y)\} ;X)\widetilde{\mu}_{t,0}(\exp\{  \gamma_t s_t(Y)\} ;X)-\widetilde{\mu}_{t,0}(Y\! \exp\{  \gamma_t s_t(Y)\} ;X)\widetilde{\mu}_{t,0}(s_t(Y)\! \exp\{  \gamma_t s_t(Y)\} ;X)\bigg\}
\end{align*}
Since $s_t(Y)$ is a monotone function of $Y$ (either increasing or decreasing), $(Y-Y^\dagger)(s_t(Y)-s_t(Y^\dagger)$ is either strictly positive or strictly negative. And since $\exp(\cdot)>0$, $\widetilde{\mu}_{t,0}( \exp\{  \gamma_t s_t(Y)\} ;X)>0$. Together with $g_0(X)>0$ and $\pi_{1-t, 0}(X)>0$, we have $\frac{\partial\psi_t(\widetilde{P}; \gamma_t)}{\partial\gamma_t}>0$ or $<0$. We can similarly prove $\frac{\partial\psi_t'(\widetilde{P}; \gamma_t')}{\partial\gamma_t'}>0$ or $<0$. Thus, $\psi_t(\widetilde{P}; \gamma_t)$ and $\psi_t'(\widetilde{P}; \gamma_t')$ are strictly monotone with respect to $\gamma_t$ and $\gamma_t'$, respectively. 

Since $|Y|$ is bounded, $\psi_t(\widetilde{P}; \gamma_t)$ and $\psi_t'(\widetilde{P}; \gamma_t')$ are bounded as well. Under monotonicity of $\psi_t(\widetilde{P}; \gamma_t)$ and $\psi_t'(\widetilde{P}; \gamma_t')$, we will consider the following three scenario: 
\begin{itemize}
    \item Suppose $\psi_t(\widetilde{P}; \gamma_t)$ and $\psi_t'(\widetilde{P}; \gamma_t')$ have the same range. For each $\gamma_t$, there exist a unique $\gamma_t'$ where $\psi_t'(\widetilde{P}; \gamma_t')=\psi_t(\widetilde{P}; \gamma_t)$ and $\min\{(\psi_t(\widetilde{P}; \gamma_t)-\psi_t'(\widetilde{P}; \gamma_t'))^2\}=0$; and for each $\gamma_t'$, there exist a unique $\gamma_t$ where $\psi_t(\widetilde{P}; \gamma_t)=\psi_t'(\widetilde{P}; \gamma_t')$ and $\min\{(\psi_t(\widetilde{P}; \gamma_t)-\psi_t'(\widetilde{P}; \gamma_t'))^2\}=0$. 
    \item Suppose $\psi_t(\widetilde{P}; \gamma_t)$ and $\psi_t'(\widetilde{P}; \gamma_t')$ have overlapping range. Then, there is no guarantee $\min\{(\psi_t(\widetilde{P}; \gamma_t)-\psi_t'(\widetilde{P}; \gamma_t'))^2\}$ will reach 0. 

    If $\min\{(\psi_t(\widetilde{P}; \gamma_t)-\psi_t'(\widetilde{P}; \gamma_t'))^2\}=0$, we can use the above proof to show the unique mapping between $\gamma_t'$ and $\gamma_t$. 
    
    When $\min\{(\psi_t(\widetilde{P}; \gamma_t)-\psi_t'(\widetilde{P}; \gamma_t'))^2\}>0$, for each $\gamma_t$, $\{\psi_t(\widetilde{P}; \gamma_t)-\psi_t'(\widetilde{P}; \gamma_t')\}^2$ is monotone with respect to $\gamma_t'$. So for each $\gamma_t$, there must exist a unique $\gamma_t'$ that achieves $\min_{\gamma_t'}\{(\psi_t(\widetilde{P}; \gamma_t)-\psi_t'(\widetilde{P}; \gamma_t'))^2\}$. Similarly, for each $\gamma_t'$, there exist a $\gamma_t$ that achieve $\min_{\gamma_t}\{(\psi_t(\widetilde{P}; \gamma_t)-\psi_t'(\widetilde{P}; \gamma_t'))^2\}$.
    \item Suppose the range for $\psi_t(\widetilde{P}; \gamma_t)$ and $\psi_t'(\widetilde{P}; \gamma_t')$ do not overlap. For all $\gamma_t$, either $\gamma_t'=-\infty$ or $\infty$ would minimize $\{\psi_t(\widetilde{P}; \gamma_t)-\psi_t'(\widetilde{P}; \gamma_t')\}^2$, since $\psi(\widetilde{P}; \gamma_t')$ is monotone. Similarly, for all $\gamma_t'$, $\gamma_t=-\infty$ or $\infty$ would minimize $\{\psi_t(\widetilde{P}; \gamma_t)-\psi_t'(\widetilde{P}; \gamma_t')\}^2$. 
\end{itemize}
This conclude Lemma~\ref{onetoone}. 

\section{Induced estimates of $\E[Y(t)|1-t]$,  $\E[Y(t)|1-t, R=0]$ and $\E[Y(t)|1-t, R=1]$}
\label{app:data_analysis}

Since 
\begin{align*}
    \E[Y(t)|T=1-t] = \frac{\E[\I(T=1-t)Y(t)]}{P(T=1-t)}= \frac{\E[Y(t)]-\E[\I(T=t)Y(t)]}{P(T=1-t)},
\end{align*}
and 
\begin{align*}
    \E[\I(T=t)Y(t)] &= \E[\E(I(T=t)Y|R=1, X)g_1(X)+\E(\I(T=t)Y|R=0, X)g_0(X)]\\
    &= \E[\mu_{t, 1}(X)g_1(X)\pi_{t, 1}(X)+\mu_{t, 0}(X)g_0(X)\pi_{t, 0}(X)],
\end{align*}
we can estimate $\E[Y(t)|T=1-t]$ by 
\begin{align*}
\frac{\widehat{\psi}_t(\gamma_t)-\E_n[\widehat{\mu}_{t, 1}(X)\widehat{g}_1(X)\widehat{\pi}_{t, 1}(X)+\widehat{\mu}_{t, 0}(X)\widehat{g}_0(X)\widehat{\pi}_{t, 0}(X)]}{P_n(T=1-t)}
\end{align*}
where $P_n(T=t')$ is the observed proportion of individuals with $T=t'$ and $\E_n[\cdot]$ is the mean of a function of $\widehat{\widetilde{P}}$, taken across all $n$ individuals. 

Since 
\begin{align*}
    \E[Y(t)|T=t]=\E[Y(t)|T=t]=\frac{\E[I(T=t)Y(t)]}{P(T=t)},
\end{align*}
we can estimate $\E[Y(t)|T=t]$ by 
\begin{align*}
    \frac{\E_n[\widehat{\mu}_{t, 1}(X)\widehat{g}_1(X)\widehat{\pi}_{t, 1}(X)+\widehat{\mu}_{t, 0}(X)\widehat{g}_0(X)\widehat{\pi}_{t, 0}(X)]}{P_n(T=t)}.
\end{align*}

Similarly, we can estimate $\E[Y(t)|R=0, T=1-t]$ by
\begin{align*}
\frac{\widehat{\psi}_{t, 0}(\gamma_t)-\E_n[\widehat{\mu}_{t, 0}(X)\widehat{g}_0(X)\widehat{\pi}_{t, 0}(X)]/P_n(R=0)}{P_n(T=1-t|R=0)},
\end{align*}
and $\E[Y(t)|R=0, T=t]$ by
\begin{align*}
\frac{\E_n[\widehat{\mu}_{t, 0}(X)\widehat{g}_0(X)\widehat{\pi}_{t, 0}(X)]/P_n(R=0)}{P_n(T=t|R=0)}
\end{align*}
where $P_n(T=t'|R=0)$ is the observed proportion of individuals with $T=t'$ in the OBS arm. We can estimate $\E[Y(t)|R=1, T=1-t]$ by
\begin{align*}
\frac{\widehat{\psi}_{t, 1}-\E_n[\widehat{\mu}_{t, 1}(X)\widehat{g}_1(X)\widehat{\pi}_{t, 1}(X)]/P_n(R=1)}{P_n(T=1-t|R=1)}, 
\end{align*}
and $\E[Y(t)|R=1, T=t]$ by
\begin{align*}
\frac{\E_n[\widehat{\mu}_{t, 1}(X)\widehat{g}_1(X)\widehat{\pi}_{t, 1}(X)]/P_n(R=1)}{P_n(T=t|R=1)}
\end{align*}
where $P_n(T=t'|R=1)$ is the observed proportion of individuals with $T=t'$ in the RCT arm. 

\end{document}